\documentclass[a4paper,11pt,nosumlimits,reqno]{amsart}
\usepackage{amsmath}
\usepackage{amsfonts,amssymb}
\usepackage{cite}
\usepackage{latexsym,amsmath,amssymb,enumerate,bbm,fullpage,amscd,graphicx,color}
\usepackage[english]{babel}
\usepackage{tikz}
\usepackage{tikz-cd}
\usepackage{subcaption}
\usepackage{booktabs,array}
\usepackage{graphicx} 
\usepackage{textcomp}
\usepackage{simplewick}
\usepackage{float}
\usepackage{xifthen}
\usepackage{xcolor}
\usepackage{sansmath}
\usepackage{marvosym}

\usepackage{tensor}

 \RequirePackage{fix-cm}
\newtheorem{theorem}{Theorem}
\newtheorem{corollary}{Corollary}
\newtheorem{lemma}{Lemma}
\newtheorem{proposition}{Proposition}

\newtheorem*{acknowledgements}{Acknowledgements}

\theoremstyle{remark}
\newtheorem{remark}{Remark}
\newtheorem{example}{Example}

\theoremstyle{definition}
\newtheorem{definition}{Definition}

\usetikzlibrary{calc,decorations.pathmorphing,patterns}
\usetikzlibrary{decorations.text}
\usetikzlibrary{arrows}
\usetikzlibrary{decorations.pathmorphing,backgrounds,positioning,fit,petri}
\usetikzlibrary{automata}
\usetikzlibrary{graphs}
\usetikzlibrary{shadows}
\usetikzlibrary{trees}
\usetikzlibrary{decorations.pathmorphing}
\usetikzlibrary{decorations.markings} 

\colorlet{texto}{black!50!gray} 

\usepackage{hyperref}
 \hypersetup{
      colorlinks   = true,
      citecolor    = blue!90!black
 }
 \hypersetup{linkcolor=blue!90!black}
\newcommand{\titulote}[1]{%
  \ifodd\value{page}%
    \protect\parbox{0.97\linewidth}{#1}\hfill%
  \else%
    \hfill\protect\parbox{0.97\linewidth}{#1}%
  \fi%
}

\usepackage{enumitem}

\newcommand{\Aut}{\mathrm{Aut}}
\newcommand{\Autc}{\Aut_{\mathrm{c}}}
\renewcommand{\and}{\quad\mbox{and}\quad}

\newcommand{\T}{\mathbb{T}}
\newcommand{\Tr}{\mathrm{Tr}}
\newcommand{\al}{\alpha}
\newcommand{\mtr}[1]{\mathrm{#1}}

\newcommand{\A}{\mathcal{A}}
\newcommand{\dif}[1]{\mathrm{d}#1}

\newcommand{\re}{\mathbb{R}}

\newcommand{\dervpar}[2]{\frac{\partial #1}{\partial #2}}

\newcommand{\si}{\sigma}
\newcommand{\G}{\mathcal{G}}
\renewcommand{\H}{\mathcal{H}}
\newcommand{\B}{\mathcal{B}}
\newcommand{\C}{\mathbb{C}}

\newcommand{\ii}{\mathrm{i}}
\newcommand{\ee}{\mathrm{e}}
 
\newcommand{\inv}{^{-1}} 
\newcommand{\mtc}[1]{\mathcal{#1}} 
\newcommand{\Z}{\mathbb{Z}}

\newcommand{\V}{\mtc{V}}

\newcommand{\with}{\,\,\mtr{with} \,\,}
\newcommand{\where}{\qquad\mbox{where}\,\,}

\newcommand{\K}{\mathcal{K}} 
\newcommand{\La}{\Lambda}

\newcommand{\R}{\mathcal R}

\newcommand{\Df}{\mathbb{\mathcal{D}}}

\newcommand{\andd}{\quad \mtr{and}\quad}	 

\newcommand{\Sym}{\mathfrak{S}}

\newcommand{\hp}[1]{^{(#1)}}

\renewcommand{\phi}{\varphi}
\newcommand{\fey}{\mathfrak{Feyn}}
\newcommand{\feyn}{\fey_3(\phi^4)}
\newcommand{\feymel}{\fey_D(\phi^4_{\mathsf{m}})}

\newcommand{\Grph}[1]{\mathsf{Grph}_{\mathrm{c},#1}}
\newcommand{\tcol}{\Grph{3}}

\newcommand{\suml}{\sum\limits}

\newcommand{\Rb}{\mathcal{R}}

\newcommand{\bJ}{\bar J}

\newcommand{\kthreeg}{\raisebox{-.2\height}{\includegraphics[width=.43cm]{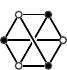}}}
\newcommand{\fder}[2]{\frac{\delta #1}{\delta #2}}
\newcommand{\wh}{_{\mathrm{w}\vphantom{b}}}
\newcommand{\bl}{_{\mathrm{b}}}

\usepackage{multicol,array}
\makeatletter
\def\moverlay{\mathpalette\mov@rlay}
\def\mov@rlay#1#2{\leavevmode\vtop{%
   \baselineskip\z@skip \lineskiplimit-\maxdimen
   \ialign{\hfil$\m@th#1##$\hfil\cr#2\crcr}}}
\newcommand{\charfusion}[3][\mathord]{
    #1{\ifx#1\mathop\vphantom{#2}\fi
        \mathpalette\mov@rlay{#2\cr#3}
      }
    \ifx#1\mathop\expandafter\displaylimits\fi}
\makeatother

\newcommand{\cupdot}{\charfusion[\mathbin]{\cup}{\cdot}}

\makeatletter
\def\leqno{\tagsleft@false}
\def\reqno{\tagsleft@false}
\def\fleqn{\@fleqnfalse}
\def\cneqn{\@fleqnfalse}
\makeatother 

 \usepackage[foot]{amsaddr}

 \DeclareCaptionLabelFormat{mylabel}{#1 #2 \hspace{1ex}}
\begin{document}
%
%

\fontsize{11.2}{14.3}\selectfont  

\title{{The full Ward-Takahashi identity for \\ colored tensor models}} 
\author{Carlos I. P\'erez-S\'anchez} 
\address{Mathematisches Institut der Westf\"alischen Wilhelms-Universit\"at, Einsteinstra\ss e 62,
48149 M\"unster, Germany
}
\email{perezsan@uni-muenster.de}



\maketitle
\begin{abstract}
  Colored tensor models (CTM) is a random geometrical approach to
  quantum gravity.  We scrutinize the structure of the connected
  correlation functions of general CTM-interactions and organize them
  by boundaries of Feynman graphs.  For rank-$D$ interactions including,
  but not restricted to, all melonic $\varphi^4$-vertices---to wit, 
  solely those quartic vertices that can lead to dominant spherical contributions in  
  the large-$N$ expansion---the aforementioned boundary graphs are shown to be
  precisely all (possibly disconnected) vertex-bipartite 
  regularly edge-$D$-colored graphs. The
  concept of CTM-compatible boundary-graph automorphism is introduced
  and an auxiliary graph calculus is developed.  With the aid of 
  these constructs, certain $\mathrm U(\infty)$-invariance
  of the path integral measure is fully exploited in order to derive a
  strong Ward-Takahashi Identity for CTMs with a symmetry-breaking
  kinetic term. For the rank-$3$ $\varphi^4$-theory, we get the
  exact integral-like equation for the 2-point function. Similarly,
  exact equations for higher multipoint functions can be readily
  obtained departing from this full Ward-Takahashi identity. Our
  results hold for some Group Field Theories as well. Altogether, our
  non-perturbative approach trades some graph theoretical methods 
  for analytical ones. We believe that these tools can be extended 
  to tensorial SYK-models. 
  \end{abstract}
\captionsetup[figure]{labelfont={bf},labelformat=mylabel }
\captionsetup[table]{labelfont={bf},labelformat=mylabel ,labelsep=none ,name={Table}} 
\setcounter{tocdepth}{4}
\tableofcontents
\section{Introduction}

The term \textit{colored random tensor models} is a collective for
random geometries obtained from quantum field theories for tensor
fields.  Aiming at a theory of quantum gravity in dimension $D\geq 2$,
these models are machineries of weighted triangulations of piecewise
linear manifolds, the weights being defined by certain path
integrals. In that probabilistic ambit, which we shall leave soon,
what we obtain here is, crudely, recursions for the connected
correlation function $G^{(2k+2)}$ in terms of $G^{(2k)}$---the
\textit{Ward-Takahashi Identities} (WTI)--- which are a consequence of
$\mathrm U(\infty)$-symmetries in the measure of their generating
functional, that is to say the free energy $\log Z[J,\bar J]$ (see below).
 Here we do not specialize in
the construction of those measures, nor use the probability
terminology, but we adhere to the physical one (e.g. we tend to use
\textit{propagator} instead of \textit{correlation}, etc.; this does
not imply that their probabilistic meaning could not be tracked back,
though). Accordingly, we drop qualificative ``random'' and stick to
\textit{colored tensor models} (CTM).
These correlation functions reflect, as we shall prove, some of the
structure of the tensor fields.  The tensors have forbidden
symmetries, which has been deemed \textit{color}.  In the
arbitrary-dimensional setting the coloring is needed in order for the
Feynman expansion to restrict to exactly those graphs one can
associate a sensible $\Psi$-complex to \cite[Lemma
1]{Gurau:2009tw}\footnote{ Pseudosimpicial or $\Psi$-complexes allow
  simplices to have more than a common face. Moreover, ostensibly, the
  coloring of GFTs is not absolutely necessary \cite{Smerlak}, but we
  stick in this paper to colors, as they more easily permit a
  systematic identification of graphs as spaces. Later on, we discuss
  models which drop coloring or part of.}.  As a byproduct of this
coloring, these theories might have several, say $a_{k}(D)$, independent
correlation functions of the same number $2k$ of points:
$G\hp {2k}_1,\ldots, G\hp{2k}_{a_{k}(D)}$.\par  This is not a feature exclusively of 
the complex tensor models that we analyze, but it will also be present 
in the (real) tensorial SYK-models (after Sachdev-Ye-Kitaev \cite{kitaev,SachdevYe})
that have been studied lately \cite{WittenSYK,Bonzom:2017pqs}
if one considers them \textit{not} as a $0+1$ field theory (as in \cite{Gross:2017aos}), but allows
spacial degrees of freedom, e.g. as in \cite{SYKTGFT}. In this sense, the 
present article could be useful if one wants to solve the (melonic 
sector of) that theories. 
\par

The initial idea in the primitive versions of random tensor models was
to reproduce, in higher dimensions, the success of random matrices in
modelling 2D-quantum gravity \cite{Ambjorn,dFGZ}.  The consummation of
this generalization had to wait long, however, until the analogue of
the large-$N$ expansion, which, as in matrix models, is bedrock of
most physical applications, was found \cite{Nexpansion}.  For these
higher dimensional analogues of random matrices, what empowered the
$1/N$-expansion is an integer called \textit{Gur\u{a}u's degree},
which for rank-$2$ tensor models (complex matrix models), coincides
with the genus (see Def. \ref{ex:exampleribbon}).  Crucially, for
dimensions greater that two, the degree is not a topological
invariant; in particular this integer has complementary information to
homology and is able to tell apart triangulations of homeomorphic
spaces.  Being tensor models a theory of random geometry, the fact
that their large-$N$ expansion relies on a non-topological quantity is
a rather wished feature, by which the theory of random tensors gains
reliability as a properly geometric quantum gravity framework for
dimensions $D\geq 2$. \par

The \textit{Tensor Track} \cite{track_update,track3,track4}
encompasses several classes of tensor models as study objects and
synthesizes these random-geometry-foundations in a
gravity-quantization program that has as watermark to leave the core
of quantum field theory intact---whenever possible.  Rooting itself in
Wilson's approach to renormalization and functional integrals, the
novelty in the tensor track is trading the locality of interactions
for invariance under certain large unitary groups
(Sec. \ref{sec:preliminaries}).  The origins of the Tensor Track are
also amends to the renormalization of Group Field Theory (GFT).  In
\cite{track_update}, Rivasseau stated Osterwalder-Schrader-like rules
that tensor models should satisfy.  One of the principal frameworks in
the Tensor Track is precisely that of CTM.  We provide in the next
paragraph an encapsulated description of alike settings sometimes
evoked by the name ``tensor model'' and studied also in the Tensor
Track.  \par
Belonging to the clade of Group Field Theories \cite{Freidel}, colored
tensor models were propelled by Gur\u{a}u. Roughly, GFTs are scalar
field theories on $D$-fold products of compact Lie groups (see
e.g. \cite{OritiGFT} for their relation to Loop Quantum Gravity, and
\cite{RovelliReisenberger} for the origin of the group manifold in the
context of spin foams).  Their Feynman diagrams encode simplicial
complexes: fields (interpreted as $D-1$ simplices) are paired by
propagators (see Sec. \ref{cgft}).  Monomials in the interaction part
of the action $S_{\mathrm{int}}$, typically of degree $D+1$ in the
fields, are understood as $D$-simplices, so Feynman graphs are gluigs
of these.  It was in that framework where the idea of coloring, which
facilitated the large-$N$ expansion, emerged \cite{Gurau:2009tw}. Ever
since, the auspicious tensor model family has dramatically grown: GFTs
with other unitary groups like $\mathrm{SU}(2)$ \cite{su2} and,
recently, with orthogonal groups \cite{On}. Another framework, not
addressed here, but which our results might be extended to, are
\textit{multi-orientable tensor models} \cite{MO}, which have some
symmetry of the
$\mathrm{U}(N)\times \mathrm{O}(N) \times \mathrm{U}(N)$-hybrid type
(assuming rank $3$). They retain still some of the graphs forbidden by
coloring and are still treatable with the large-$N$ expansion
\cite{reviewtanasa}.  Tensor Group Field Theories is another
GFT-related setting to which actually some of our results are extended
(Section \ref{cgft} for the $\mathrm U(1)^D$-group).  Concerning
renormalization of TGFTs a good deal of results pertaining the
classification of these models, has been undertaken specially by Ben
Geloun and Rivasseau \cite{4renorm} in $D=4$, and Ousmane Samary,
Vignes-Tourneret in $D=3$. The former model (BGR), a TGFT on
$\mathrm{U}(1)^4$, is one of the prominent $4$-dimensional models
which, moreover, as its authors themselves proved, is a renormalizable
field theory to all orders in perturbation theory.  Among all its
relatives, CTM render the best-behaved spaces. We choose to
temporarily constrain to this framework because geometric notions
become more transparent there.  We also ought to show the surjectivity
of certain map having as domain the Feynman graphs of a fixed tensor
model action.  That set is meagerer in the CTM-framework, where the
result becomes then stronger.
\par 
Relying on it, the main result of the present work, the \textit{full
   Ward-Takahashi Identity} (Theorem \ref{thm:full_Ward}), is
 non-perturbative QFT for tensor models, in essence.  Historically, the WTI
 appeared in matrix models in order to show the vanishing of the
 $\beta$-function of the Grosse-Wulkenhaar model---which had been
 already accomplished by other methods at one \cite{GWbeta} and three
 loops \cite{DR}---to all orders in perturbation theory.  The ultimate
 proof \cite{DGMR}, by Disertori, Gur\u{a}u, Magnen and Rivasseau,
 still perturbatively, was based on a Ward Identity (WI) also derived
 by them there \cite[Sec. 3]{DGMR}.  Later on, Grosse and Wulkenhaar
 \cite{gw12} retook the WI for their self-dual $\phi^{\star 4}_4$-model (see
 \eqref{GW} with $\Omega=1$) to give a non-perturbative proof that \textit{any} 
 quartic matrix model has a vanishing $\beta$-function.  We adapt the
 non-perturbative matrix model approach of \cite{gw12} to colored tensor
 models.  The full WTI (Theorem \ref{thm:full_Ward}) is proven for an
 arbitrary rank and for absolutely general CTM-interactions.  It holds
 for $\mathrm U(1)$-Group Field Theories (GFTs) as well, by
 Fourier-transforming them.  
  
\noindent
\textbf{\\ The strategy}.  We closely follow the treatment given in
\cite[Sec. 2]{gw12} to the Grosse-Wulkenhaar $({ \Omega=1 })$-model, a
$\phi^{\star 4}$-theory in Moyal $(\re^4,\star)$ which becomes, in the Moyal
matrix basis, a matrix model \cite[Sec. 2]{gw:matrixbase04}. In
\cite{gw12} the Ward identities are used to decouple the tower of
Schwinger-Dyson equations (SDE), which results in an integro-differential
equation for the two-point function, in terms of which, via algebraic
recursions, the theory can be solved, i.e. all $2k$-point-functions
(which are the non-vanishing ones) are thus determined.  Simply
stated, the strategy can be split in two tasks. First, to expand the
free energy $W[J,\bar J]=\log Z[J,\bJ]\sim \sum_{\mathbf p}
\sum_{\partial \mathcal F} (1/\sigma(\partial \mathcal F)) G_{\partial \mathcal F}(\mathbf p)
\cdot \partial \mathcal F (J,\bar J)(\mathbf p)$ in boundaries $\partial \mathcal F$ of
Feynman graphs $ \mathcal F$ of a specific model, with source-variables $J$ and
$\bar J$ and $\mathbf p$ being momenta and $\sigma(\partial\mathcal F)$ a symmetry factor.  
For matrix models this approach has an astonishing result and needs, in our setting, mainly three steps:
 \begin{itemize}
 \item Finding the right symmetry factors $\sigma(\partial  \mathcal F)$,
 which in turn requires
 the CTM-compatible concept of \textit{automorphism of colored graphs}. 
 This new concept, contrary to the existent in the literature of 
 graph encoded manifolds, precisely exhibits compatibility with the CTM-structure 
 (see Sec. \ref{sec:automorphisms}). Automorphism 
 groups are also computed. 
\item \textit{Non-triviality}.  Since $W[J,\bar
  J]=\log Z[J,\bJ]$ cancels out the disconnected Feynman graphs, one
  has to construct \textit{connected} Feynman graphs with possibly
  disconnected, arbitrary boundary graph $\B$.  This would ensure that
  each introduced correlation function $G_\B\hp n$ describes indeed
  a process in the model under study.  We develop first, in Section
  \ref{sec:graphtheory}, an operation introduced in \cite{cips} for
  rank $2$ and interpreted there as the connected sum, and take it
  further to arbitrary rank $D$. This operation sends two 
  Feynman graphs of a fixed model to a Feynman graph of the same model
  (Prop. \ref{thm:suma_grado}). Furthermore, the divergence degree
  that controls the large-$N$ expansion behaves additively with
  respect to it (Prop. \ref{thm:suma_grado}).
 
\item \textit{Completeness}.  The exact set of boundary graphs is
  expected to be model-dependent.  We determine it for
  \textit{quartic} (for $D\geq 4$ \textit{quartic melonic})
  interactions and show that it is the \textit{whole} set of
  $D$-colored graphs (see Section \ref{sec:completeness}).
 \end{itemize}
 The second task is to actually derive the WTI from these constructs.
 In order to be able to read off from $W$ any correlation function, a
 graph calculus is developed in Section \ref{sec:graphcalculus}.
 \\
 \noindent
 \\
\textbf{The results}.  
For tensor models, a version of the WTI was obtained in  
\cite{DineWard}, with emphasis on ranks $3$ and $4$.  Here we go a different,
considerably longer way that has the following advantages:
\begin{itemize}
\item it is a \textit{non-perturbative} treatment.  This approach
  shows a way out of treating single Feynman graphs in tensor models
  and proposes analytic methods instead.  We prove that the
  correlation functions are indexed by boundary graphs, though, so
  graph theory cannot be fully circumvented.
 \item  it exhibits the intricate, so far unknown structure of
   the Green's functions. 
      That the structure of the boundary sector of 
   \textit{single} models had not been studied underlies this shortcoming. 
   Green's functions are indexed by 
   all boundary graphs; for quartic interactions,
   namely by all $D$-colored graphs. Using \cite{counting_invariants} (see eqs. 
   \eqref{eq:invariants3} and \eqref{eq:invariants4} below) there are then
   in rank-$3$, four $4$-point functions, 
   eight $6$-point functions; for $D=4$, eight
   $4$-point, forty nine $6$-point functions and so on. 
 
 \item it is the \textit{full} WTI. Roughly speaking, the
   Ward-Takahashi identities contain a skew-symmetric tensor $E_{mn}$
   times a double derivative on the partition function. This double
   derivative splits in a part proportional to $\delta_{mn}$, which is
   annihilated by $E_{mn}$, and the rest. The existing WTI in
   \cite{DineWard} does not contain the former term.  It was enough
   for successfully treating a ``melonic-approximation'' \cite{us} and
   writing down a closed integro-differential equations for the
   lower-order correlation functions. Our aim, on the other hand, is
   the full theory. Accordingly, we compute here \textit{all terms}:
   non-planar contributions, in the matrix case, and non-melonic terms
   ---their tensor-model counterpart--- are all recovered.
 %
 %
 \end{itemize}  
 \begin{figure}\centering $\qquad$
 \includegraphics[height=3.1cm]{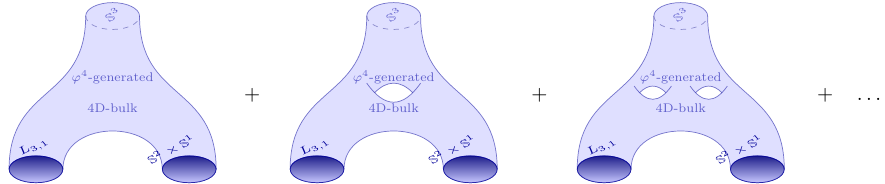}
\caption{Geometric picture of the expansion of a 
concrete Green's function in Gur\u au's degree for a particular 
correlation function. Gur\u au's degree is depicted by a handle 
(it is not a topological invariant, though)\label{fig:expgurau}
}
\end{figure}
After succinctly introducing the general setting of CTMs in next
section, we recap in Section \ref{sec:graphtheory} the main graph
theory of colored tensor models\footnote{A much more thorough
  exposition is given in \cite{cips} (keeping a very similar
  notation).} but adding some new definitions and results useful in
order to find, in Section \ref{sec:completeness}, the boundary sector
of quartic theories. This has a twofold application. On the one hand
it is basis for the expansion of the free energy in boundary graphs
(Section \ref{sec:ciclos}) which we use in Section \ref{sec:WI} to
obtain the full WTI.  On the other hand, it is useful in finding the
spectrum of manifolds that a specific CTM is able to generate.  We
offer some non-sphere examples of prime factors graphs generated by
boundaries of quartic CTMs---here a lens space and
$\mathbb S^2\times \mathbb S^1$.  Section
\ref{sec:geom_interpretation} serves to emphasize this and following
aspect about the results of Section \ref{sec:completeness}: If $\B$ is
a graph with $n$ vertices representing a manifold $M$, then the
multi-point function $G\hp{n}_\B$ is expected to have geometrical
information about all compact, oriented $4$-manifolds bounded by $M$.
In this bordism picture---here including the vacuum graphs to the picture,
for which $M$ is empty---some manifolds cannot be obtained 
from tensor model Feynman graphs, independently of the particular model, e.g. from the onset, 
Freedman's $E_8$ manifold cannot appear \cite{freedman1982}. 
Notice that in dimension $4$, the categories of topological and PL-manifolds (PL$_4$) are not equivalent,
so manifolds with non-trivial Kirby-Siebenmann class \cite{KirbySiebenmann}  cannot be tensor model graphs. 
Nevertheless\footnote{I thank the referee for the comments concerning the $4$-dimensional case}, the PL$_4$ category is the same as the category of smooth $4$-manifolds
\cite{simple4PL}. Therefore, in dimension $4$, tensor models
still can in principle access all smooth structures, and which of them are obtained,
is model dependent. (It is likely that the model given by the four ``pillow-like'' invariants 
in $D=4 $ colors, what we here call the $\phi^4_{4,\mathsf{m}}$-theory, suffices to generate them all.)  \par 
Each Green's function can be expanded in subsectors determined by
common value of Gur\u au's degree $\omega$, symbolically represented
as in Figure \ref{fig:expgurau} for
$M= L_{3,1}\sqcup (\mathbb S^1\times \mathbb S^2) \sqcup \mathbb S^3$
(see ex. \ref{ex:4bordism}). That expansion, as in the matrix case,
can lead to closed integro-differential equations for sectors such
sectors.  In particular, this paper provides techniques to find
integro-differential equations that these Green's functions obey.

  \normalsize


\section{Colored tensors models} \label{sec:preliminaries}  
The next setting describes a theory that works in certain high-energy
scale $\La$. With that resolution, an ordinary scalar vertex shows more
structure. For instance, this one:
\begin{equation}
\tikz[scale=0.7,  
    baseline=4ex,shorten >=.1pt,node distance=18mm, 
     auto,
    every state/.style={draw=black,inner sep=.4mm,text=black,minimum size=0},
]
{
\coordinate (O) at (0,0);
\foreach \ang in {45,135,-135,-45}
\draw (\ang:1.5) -- (O);
\node at (-6,0) {$\phi^4$-interaction at the energy scale $\La_0$};
;
\node[state, minimum size=0,circle, draw,fill=black!50] at (O) {};
\begin{scope}[shift={(3,0)}]
\end{scope}
\draw[dashed,gray] (O) -- (-30:4);
\draw[dashed,gray] (O) -- (29:4);
\node[state, minimum size=0,circle, draw,fill=black!50] at (O) {};
}\hspace*{-.992cm}
\raisebox{-0.7182\height}{\includegraphics[width=3.95cm]{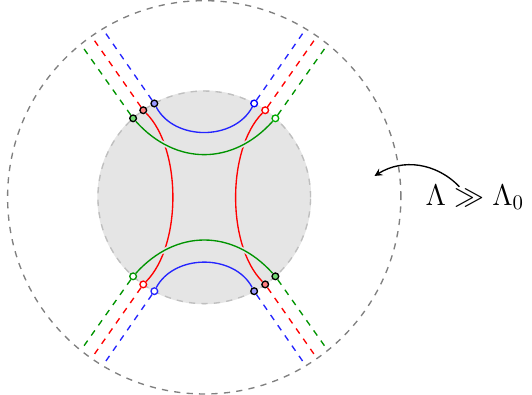}}
\label{general_three}
\end{equation}
At the energy scale $\Lambda$ there is a $\mathrm{U}(N_1\cdots
N_D)$-symmetry that is broken into
$\mathrm{U}(N_1)\times\cdots\times \mathrm{U}(N_D)$ giving 
rise to more invariants.  One postulates tensor fields $\phi$ and
$\bar\phi$ that transform independently under each unitary group
factor. Concretely, being $\H_c$ Hilbert spaces, usually
$\ell_2([1,N])$, but also $\ell_2([-N,N])$, our fields are tensors
$\phi,\bar\phi: \H_1 \otimes \H_2 \otimes \ldots \otimes \H_D\to \C$
that transform like
\begin{align*}
\phi_{a_1a_2\ldots a_D}\mapsto \phi'_{a_1a_2\ldots a_D}&=
\sum_{b_k}  W^{(k)}_{a_kb_k}
\phi_{a_1a_2 \ldots b_k \ldots a_D}\,,  
  \\
\bar\phi_{p_1p_2\ldots p_D}\mapsto\bar \phi'_{p_1p_2\ldots p_D}&= 
\sum_{q_k}
\overline W^{(k)}_{p_kq_k}
\bar \phi_{p_1p_2\ldots q_k \ldots p_D}\,,
\end{align*}
for every $W^{(k)}\in \mathrm U(N_k)$ and for each one of the
so-called \textit{colors} $k=1,\ldots, D$. Here, the rank of the
tensors, $D\geq 2$, is the dimension of the random geometry we want to
generate. For sake of simplicity, one sets $N_k=N$, for each color
$k$, but one insists in distinguishing each factor of the group
$\mathrm{U}(N)^{ D}$. Each such factor acts independently on a single
index of both $\phi$ and $\bar \phi$, which is refereed to as
\textit{tensor-coloring.}  The energy scale $\Lambda$ can be seen as
(a monotone increasing function of) this large integer $N$.
Symbolically we write the indices of each tensor in $\Z^D$, but one
should think of it as a cutoff-lattice $(\Z_N)^D$. \par

The classical action functional is build from a selection of connected
$\mathrm{U}(N)^{\otimes D}$-invariants, which are given by traces
$\{\Tr_{\mathcal{ B}_\alpha}(\phi,\bar\phi)\}_{\al}$ indexed by
regularly $D$-edge colored, vertex-bipartite graphs. We shorten this
term simply to $D$-\textit{colored graphs} (see
Sec. \ref{sec:graphtheory} for details). There is, in any rank, only
one quadratic invariant,
$\Tr_2(\phi,\bar\phi)=\sum_{\mathbf a\in \Z^D} \bar\phi_{\mathbf a}
\phi_{\mathbf a}$,
which is, as always, understood as the kinetic part.  Higher order
invariants as
\begin{align}\label{k33}
  \Tr_{K_{\mathrm{c}}{(3,3)}}(\phi,\bar\phi) &= \sum_{\mathbf{a,b,c,p,q,r}}    
                                               (\bar \phi_{r_1r_2r_3}\bar\phi_{q_1q_2q_3} \bar\phi_{p_1p_2p_3}) \cdot &\\
  \nonumber
                                             && \hspace{-5cm}  (\delta_{a_1p_1} \delta_{a_2r_2} \delta_{a_3q_3}  
                                                \delta_{b_1q_1} \delta_{b_2p_2} \delta_{b_3r_3}
                                                \delta_{c_1r_1} \delta_{c_2q_2} \delta_{c_3p_3}) \cdot 
                                                (\phi_{a_1a_2a_3} \phi_{b_1b_2b_3}  \phi_{c_1c_2c_3})\,,
\end{align} 
 are the \textit{interaction vertices\footnote{ Due to the common
     occurrence of the word \textit{vertex} both by field theory and
     graph theory, we cannot opt, unfortunately, for a concise
     terminology.},} in this rank-$3$ example $\Tr_{K_{\mathrm{c}}{(3,3)}}(\phi,\bar\phi) $ being
  of sixth degree, and
 the sum being carried over \textit{momenta} $\mathbf{a,\ldots,r}\in \Z^{
   3}$. The $D$-colored graph $\B$ that indexes a generic interaction
 vertex $\Tr_{\B}$ is obtained by the prescription in Table
 \ref{equivalencias}.
\begin{table}\centering
 \includegraphics[width=.80\textwidth]{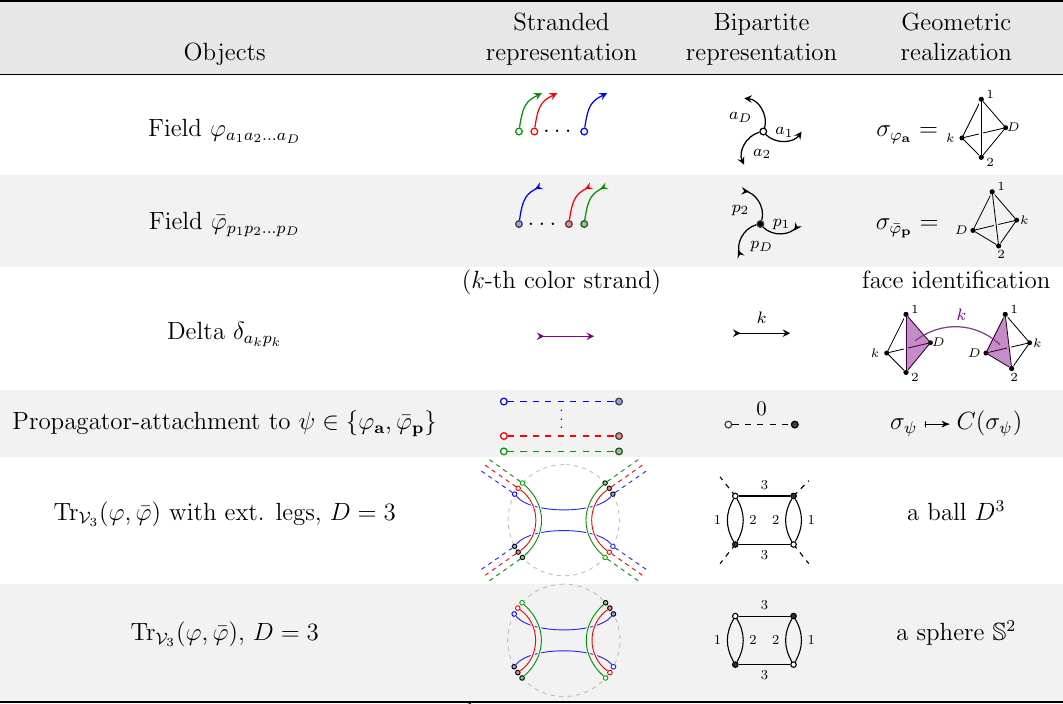}
 \caption{\label{equivalencias}  A dictionary between two
   equivalent representations of graphs and their associated geometric realization is shown. A
   more detailed construction is exposed in Section
   \ref{sec:geom_interpretation}.  Here $C$ denotes the cone of a
   simplex  }
\end{table}
Thus, for instance in $D=3$ colors,
the colored graph $K_{\mathrm c}(3,3)$ that indexes the 
interaction vertex \eqref{k33} is 
\begin{equation} \label{k33stranded}
 \centering
  \raisebox{-0.5\height}{\includegraphics[height=3.3cm]{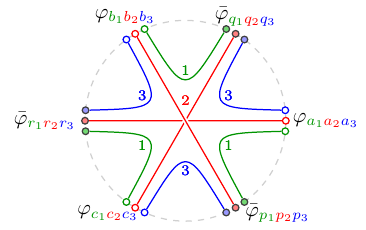}}
\end{equation}
This somehow obsolete notation is the so-called \textit{stranded
  representation}. We shall now use an equivalent, simpler notation of
these graphs: the \textit{bipartite representation}. This transition
is summarized in Table \ref{equivalencias} and allows a connection
with the graph theoretical representation of piecewise-linear
manifolds \cite{survey_cryst}, as we explain later in Section
\ref{sec:geom_interpretation}, which is the main link to the geometry
of CTMs. However, the graphs one actually associates a
(pseudo)manifold-meaning to arise in the Feynman expansion of
\begin{equation*} 
 Z[J,\bar J]=\frac{\int\Df [\phi, \bar\phi] \,\ee^{\Tr{(\bar J\phi)}+
\Tr{(\bar\phi J)}-N^{D-1}S[\phi,\bar\phi]}}
 {\int\Df [\phi, \bar\phi]\, \ee^{-N^{D-1}S[\phi,\bar\phi]}}\,,\!\!\!\!\!\!\where
\Df [\phi, \bar\phi]:=\!\!\!\prod\limits_{\,\,\,\,\mathbf {a} \in
\Z^D} \!\!\! N^{D-1}\frac{\dif\varphi_{\mathbf a} 
\dif\bar\phi_{\mathbf{a}}}{2\pi \ii}\ee^{-\Tr_{2}(\varphi,\bar\phi)},
\end{equation*}
and have one extra edge between any pair of Wick-contracted fields.
Associated to these Wick's contractions is the \textit{$0$-color},
drawn always dashed (or in the stranded representation, $D$ parallel
lines as reads in Table \ref{equivalencias}) and the graph one remains
with turns out to be $(D+1)$-colored (open or closed) as explained in
the next example.
\begin{example} \label{ex:inicial}
 We will study a particular model: \textit{the
   $(\phi_{D=3}^4)$-theory}.  Its interaction vertices are
\begin{align} 
\label{vertices}
  \mathcal V_1  &=\lambda \cdot \,\,  \raisebox{-.46\height}{\includegraphics[width=1.68cm]{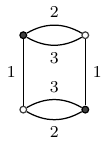}},   
  & & \mathcal{V}_2 \hspace{- .63cm} 
  &= \lambda\cdot \,\, \raisebox{-.46\height}{\includegraphics[width=1.68cm]{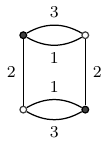}}\, , & & \mathcal{V}_3 
  &= \lambda \cdot \,\,\raisebox{-.46\height}{\includegraphics[width=1.68cm]{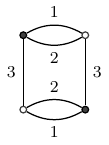}}
    \,.
\end{align} 
We have chosen directly the bipartite representation but, in order to
clarify the switch of notations explained in Table
\ref{equivalencias}, we consider one of the $\mathcal
O(\lambda^2)$-vacuum-graph contributions to the integral 
$\int
\Df[\phi,\bar \phi] \exp{(-S_0)} (\Tr_{\V_3}(\phi,\bar\phi)
\Tr_{\V_1}(\phi,\bar\phi)) $, given in Figure
\ref{fig:feynd3}.  It will be seen thereafter that this graph is a
(pseudo)simplicial complex that triangulates the sphere $\mathbb S^3$
with eight $3$-simplices.
\begin{figure}\centering\includegraphics[width=.72 \textwidth]{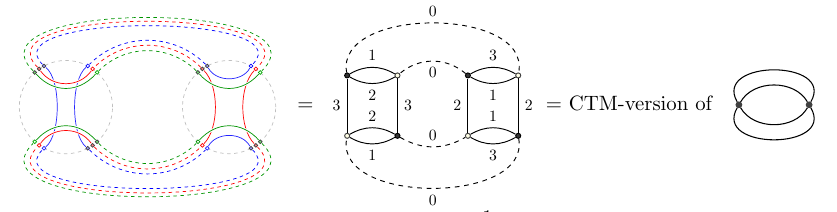} 
\caption{\label{fig:feynd3} Change of notations. Forgetting the tensor
  structure these both are an enriched version of the leftmost Feynman
  diagram in an ordinary scalar theory }
\end{figure}
\end{example}

\begin{remark}\label{rem:SA}
  Tensor field theory also has propagators that break the invariance
  in the action, in this case under the unitary groups.  It is
  therefore sensible to consider a slightly modified trace with a
  symmetry-breaking term $E$ in the quadratic term:
  $S[\phi,\bar\phi]=\Tr_2{(\bar \phi,E\phi)}+
    \sum_\alpha\Tr_{\mathcal{
        B}_\alpha}(\phi,\bar\phi)$, 
  with $E:\H_1 \otimes \H_2\otimes \H_3\to
  \H_1 \otimes \H_2\otimes \H_3$ \textit{self-adjoint},
  $\Tr_2(\bar\phi,E\phi)=\Tr_2(\overline{E\phi},\phi)$.  The first term is
  distinguished, and represents the
  kinetic part of the action, where $E$ could be interpreted as the
  Laplacian.\end{remark}

\fontsize{11.2}{14.3}\selectfont  

\section{Colored graph theory}\label{sec:graphtheory}
In this section we intersperse examples aimed at explaining a series
of definitions that concern the CTM-graphs. Each Feynman graph will be
taken connected, but boundary graphs of these need not to be so,
whence the occurrence of the disconnected graphs in our definitions.

\begin{definition}
  A $D$-\textit{colored graph} is a finite graph
  $\G=(\G\hp 0,\G\hp 1)$ that is vertex-bipartite and regularly
  edge-$D$-colored in the following sense:
\begin{itemize}
\item the vertex-set of $\G$, denoted by $\G\hp 0$, is composed by
  \textit{black} $\G\hp 0\bl$ and \textit{white} vertices
  $\G\hp 0 \wh$: $\G\hp 0= \G\hp 0\wh\cupdot\G\hp 0\bl$,
\item   any edge $e\in \G\hp 1$ is attached to precisely one white
  vertex $a$ and one black one $w$, which we denote by $t(e)=a, s(e)=w$
  or, alternatively, $e=\overline{aw}$ (thus the number 
  of white and black vertices is the same; loops are forbidden),  
\item the edge set is regularly $D$-colored, i.e.
  $\G\hp 1=\cupdot_{k=1}^D\G\hp 1_k$, where $\G_k\hp 1$ are the
  \textit{color-$k$ edges}. Moreover at each vertex there are $D$
  differently colored incident edges.
\end{itemize}
We write $\Grph{D}$ for the set of all \textit{connected} $D$-colored
graphs and $\amalg\Grph{D}$ for the set of (possibly) disconnected
graphs with finite number of connected components. Each connected
component of the subgraph of $\G$ with edges colored by a subset
$I=\{i_1,\ldots,i_q\}\subset \{1,\ldots,D\}$ of cardinality $q$ is
called $q$-\textit{bubble}.  On top of the edge-color set $I$, one
needs a vertex $v$ or edge $e$ of $\G$ that specifies the connected
component. The notation for a bubble is therefore $\G^I_{v}$,
$\G^I_{e}$ or, if specifying the colors that do not appear is easier,
$I=\{1,\ldots,D\} \setminus \{c_1,\ldots,c_r\}$, say,
$\G^{\hat c_1,\ldots, \hat c_r}_v$.  We write $\G\hp q$ for the set of
$q$-bubbles of $\G$; in particular, $\G\hp 2$ is the set of the
\textit{faces} of $\G$. \end{definition}

Graphs in either set $\Grph{D}$ or $\amalg \Grph D$ are said to be
\textit{closed}, in contrast to:

\begin{definition}
  A graph $\G$ is an \textit{open} $(D+1)$-colored graph if, first,
  its vertex-set is bipartite in the sense of (i) and (ii) below and
  if the edge set $\G^{(1)}=\cupdot_{c=0}^D\G^{(1)}_c$ is
  quasi-regularly (up to the $0$ color) $(D+1)$-\textit{colored} in
  the sense of (a) and (b):
\begin{itemize}
\item[(i)] the vertex-set is bipartite,
  $\G^{(0)}=\G^{(0)}_{\mathrm{w}\vphantom{b}}\cupdot\G^{(0)}_{\mathrm{b}}$,
  where $\G^{(0)}_{\mathrm{w\vphantom{b}}}$ are the \textit{white},
  and $\G^{(0)}_{\mathrm{b}}$ the \textit{black} vertices, and any
  edge $e$ is adjacent to precisely one vertex in $\G\hp 0\bl$ and a
  vertex in $\G\hp 0\wh$. Therefore one has the same number of black
  and white vertices,
\item[(ii)] any vertex is either \textit{internal} or \textit{external},
  $\G^{(0)}=\G^{(0)}_{\mathrm{inn}}\cupdot\G^{(0)}_{\mathrm{out}}$;
  moreover, the set $\G^{(0)}_{\mathrm{inn}}$ of internal vertices is
  regular with valence $D+1$ and external vertices have valence $1$,
\end{itemize} 
and, denoting by $ \G^{(1)}_c$ the edge-set of color-$c$:
\begin{itemize}
 \item[(a)] for each color $c$ and each inner vertex  
   $v\in\G^{(0)}_{\mathrm{inn}}$, there is exactly one color-$c$
   edge, $e\in \G^{(1)}_c$, attached to $v$,
 \item[(b)] external vertices $v \in\G_{\mathrm{out}}\hp 0$ are
   attached only to color-$0$ edges. We call both a vertex in
   $\G_{\mathrm{out}}\hp 0$ and the edge attached to it an
   \textit{external leg.}
\end{itemize}
As is common in QFT, we sometimes drop the external vertices and keep
only the external (in this case $0$-colored) edge.  Given an open
graph $\G$ one can extract a (in general non-regularly) colored graph
$\mathrm{inn}(\G)$ defined by
 $\mathrm{inn}(\G)\hp 0 = \G\hp 0_{\mathrm{inn}}$  and $ 
 \mathrm{inn}(\G)\hp 1 = \G\hp 1  \setminus \{\mbox{external legs of }\G\}$.
 The graph $\mathrm{inn}(\G)$ is called \textit{amputated} graph.  For
 any $p\in \Z_{\geq 0}$, we set
\begin{equation}
 \Grph{D+1}\hp {2p}:= \big\{ \G \mbox{ open }
 (D+1) \label{eqn:legs} \mbox{-colored }\big|\,\, \#\big(\G\hp
 0_{\mathrm{out}}\big)=2p \big\}\,.
\end{equation}
The factor $2$ arises from vertex-bipartiteness.  Here for $p=0$, of
course $\Grph{D+1}\hp {0}:=\Grph{D+1}$.
\end{definition}

\begin{example}\label{ex:colgraph} 
In this example $K_{\mtr c}(3,3)$ is the colored utility graph, which,  
tangentially, is the ``bipartite'' version of the stranded
representation of \eqref{k33stranded}. One has:
\[\label{eq:k33BW}
 K_{\mtr c}(3,3)=\raisebox{-.46\height}{
  \includegraphics[height=1.7cm]{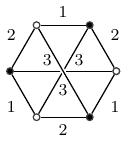}} \in \Grph{3}
  \quad \mbox{ and } \quad
  \raisebox{-.46\height}{
  \includegraphics[height=1.7cm]{gfx/k33BandW_labels_weiss.pdf}}  \sqcup 
  \raisebox{-.46\height}{
  \includegraphics[height=1.5cm]{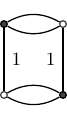}}  \in \amalg\Grph{3}. 
\]

\end{example}
\begin{example}\label{ex:open_graph}
The  graph $\mathcal K$ below is open and lies in $\Grph{3+1}\hp 6$. 
We depict also its amputation, $\mathrm{inn}(\mathcal K)$:
\begin{equation*} 
\mathcal{ K}= \raisebox{-0.485\height}{\includegraphics[height=4.2cm]{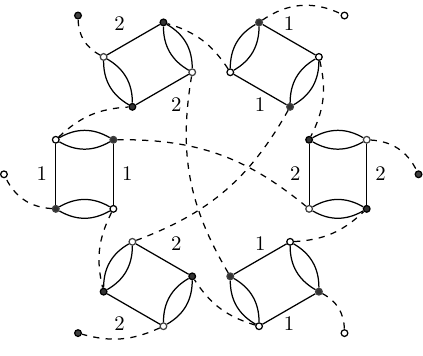}} 
\qquad 
\mathrm{inn}(\mathcal K) = \raisebox{-0.48\height}{
\includegraphics[height=4.0cm]{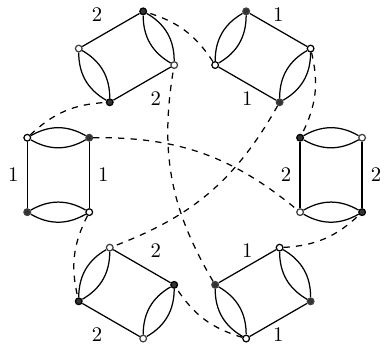}} 
\end{equation*} 
\end{example}

\begin{definition}
The \textit{boundary graph} $\partial \G$ of a $(D+1)$-colored graph $\G\in \Grph{D+1}\hp{2k}$
 defined by:
\begin{itemize}
 \item  its vertex set is $(\partial\G)^{(0)}=\G_{\mathrm{out}}$,
 inheriting the bipartiteness of the $\G_{\mathrm{out}}$.
\item the edge set $(\partial\G)^{(1)}$ is partitioned by colors
  $k\in\{1,\ldots,D\}$. For each color $k$, one sets
  $(\partial\G)^{(1)}_k:=\{(0k)\mbox{-colored paths in }\G\}$.  The
  incidence relations are given by the following rule: a white vertex
  $a\in (\partial\G)^{(0)}\wh$ is connected to a black vertex
  $x\in(\partial\G)^{(0)}\bl $ by a $k$-colored edge
  $e_k\in \partial \G\hp 1_k$ if and only if there is a
  $(0k)$-bicolored path in $\G$ between the external vertices $a$ and
  $x$.
\end{itemize}
\end{definition}
One can easily see that $\partial \G \in \amalg \Grph{D}$ by 
identifying $(0i)$-bicolored edges with $i$-colored edges, for $i=1,\dots,D$.
\begin{example}
The next graph is the \textit{cone} of $K_{\mtr c}(3,3)$,
\[ 
C(K_{\mtr c}(3,3)) =\raisebox{-0.45\height}{\includegraphics[height=2.3cm]{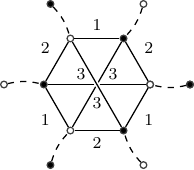}}
\]
The boundary of $C(K_{\mtr c}(3,3))$ is obviously $K_{\mtr c}(3,3)$
itself (ex. \ref{ex:colgraph}).  In our construct, it will be
important to be able to generate arbitrary graphs
$\B\in \amalg\Grph{D}$ as boundaries of a certain theory with
\textit{fixed} interaction vertices. Then, generating them by
\textit{coning} $\B$ ---that is, by adding an external color-$0$ leg
to each vertex of $\B$--- is not an option, for one would need to add
to the classical action the interaction vertex given by the connected
components of $\mathrm{inn}(C\B)\in \amalg \Grph D$ (and thereby
additional coupling constants should in principle be measured).  The
boundary graph $\partial \K$ of $\K$ in example \ref{ex:open_graph} is
$K_{\mtr c}(3,3)$.  This is the ``right'' type of graph for us,
e.g. obtained solely from a $\phi^4$-theory.
\end{example}

\begin{definition}\label{ex:exampleribbon}
Given a graph $\G\in\Grph{D+1}$, each cycle $\sigma \in \mathfrak{S}_{D+1} $ 
a ribbon graph $\mathcal J_\sigma$ called \textit{jacket}, which 
is specified by:
\begin{align*}
\mathcal J_\sigma \hp 0 = \G\hp 0 , \quad
\mathcal J_\sigma\hp 1 = \G\hp 1, \quad 
\mathcal J_\sigma\hp 2 = \{ f\in \G\hp 2 : 
f \mbox{ has colors }  \sigma^q(0)\mbox{ and }\sigma^{q+1}(0) ,  q\in\Z\}.
\end{align*}
Here $\sigma^q(0)$ is the $q$-fold application of $\sigma$ to $0$.
Obviously $\sigma$ and $\sigma\inv$ lead to the same jacket.  Moreover
each jacket, being a ribbon graph, has a genus \cite{cips} and the sum
of the genera of the $D!/2$ jackets of $\G$ is called \textit{Gur\u au's degree} 
and denoted by $\omega(\G)$. If a graph has a vanishing
degree, it is called \textit{melon}.  For $D=2$ then Gur\u au's degree
is the genus of the graph, as the only jacket is the graph itself;
melons in rank-$2$ are planar ribbon graphs. In any degree, melons
triangulate spheres \cite{GurauRyan}.
\end{definition}

\begin{example}\label{ex:jackets} 
  The necklace graph $\mathcal N$ defined by eq. \eqref{eq:necklace}
  has two spherical jackets $\mathcal J_{(1234)}$ and
  $\mathcal J_{(1423)}$ and a toric jacket $\mathcal J_{(1324)}$ (see
  \cite{cips} for the full computation).  Jackets are the
  graph-version of surfaces corresponding to Heegaard splittings
  \cite{GurauRyan}.  Hence the geometric realization of $\mathcal N$
  has a genus-$0$ Heegaard splitting and is therefore a sphere.  Also
  $\mathcal J_{(1324)}$ in $\G$ is the ``Clifford torus'' $\T^2$ in
  $\mathbb S^3$.
\end{example}

One way to
determine Gur\u au's degree \cite[App. A, Prop 1]{uncoloring} of a 
graph $\G\in\Grph{D+1} $ is to count its faces 
$\G\hp2$ and to use the formula 
\begin{equation} \label{eq:facecounting}
|\G\hp 2|= \frac12{D \choose 2}\cdot {|\G\hp 0|} +D-\frac{2\omega(\G)}{(D-1)!} \,.
\end{equation} 
The relevance of this integer relies in the analytic control it gives
to the theory of random tensor models. Here, the amplitude $\A(\G)$ of
Feynman graphs $\G$ in CTMs has the following behavior
 $\A(\G) \sim N^{D-\frac{2 \omega(\G)}{(D-1)!}} \,.$
 
\begin{definition}\label{def:general_model}
  A \textit{colored tensor model} $V(\phi,\bar\phi)_D$ is determined
  by three items. First, an integer $D \geq 2$, called
  \textit{dimension} of the model. This integer $D$ is the rank the
  tensors. Secondly, by an \textit{action}
  \[ V(\phi,\bar\phi)_D= \sum_{\B\, \in\, \Omega} \lambda_\B
  \,\Tr_\B(\phi,\bar \phi)\,,\]
  where $\Omega\subset \Grph{D}$, $|\Omega|<\infty $, and
  $\lambda_\B\in \re$.  Finally, by a \textit{kinetic term}
  $E:\bigotimes_{c=1}^D \H_c \to \bigotimes_{c=1}^D \H_c$ that is
  self-adjoint in the sense of
  $\Tr_2(\overline{E\phi},\phi)=\Tr_2(\bar\phi,E\phi)$.  Usually terms
  $E\neq 1$ are employed to make connection with GFTs and TGFTs, as
  the Laplacian boils down to such a term. We will often obviate $E$
  and specify the model only by the potential. The set of
  \textit{(connected) Feynman diagrams} of the \textit{model}
  $V(\phi,\bar\phi)_D$ is denoted by $\fey_D (V)$ and satisfies
\[
\fey_D (V)= \big\{ \G \in \cupdot _{k=0}^\infty 
\Grph{D+1}\hp{2k}\,\big|\, \mathrm{inn}(\G)^{\hat 0} 
\in \Omega \mbox{ and }  (\mathrm{inn}(\G))_0\hp 1 \neq \emptyset  \big\} \,.
\]
The graphs in $\Grph{D+1}\hp 0 \cap \fey_D(V)$ are called \textit{vacuum graphs}
of the model $V$. We are interested in honest Feynman graphs,
that is, those having internal propagators
(in other words, those that are not 
the cone of an interaction vertex). This explains the mysterious
restriction $(\mathrm{inn}(\G))_0\hp 1 \neq \emptyset$.
\end{definition}

\begin{definition}\label{thm:surg_ribb}
  Let $\Rb$ and $\mathcal Q$ be connected $(D+1)$-colored graphs,
  $\Rb \in \Grph{D+1}\hp{2k}$ and $ \mathcal Q\in\Grph{D+1}\hp{2l}$.
  Let $k$ be any color and let $e$ and $f$ be color-$k$ edges in $\Rb$
  and $\mtc Q$, respectively, i.e.  $e\in \Rb\hp 1_k$ and
  $f\in \mathcal Q_k\hp 1$.  We define the graph
  $\Rb \tensor*[_e]{\#}{_{\!f}} \mtc Q $ as follows:
\begin{align*}
(\Rb  \tensor*[_e]{\#}{_{\!f}} \mtc Q) \hp 0 & = \Rb\hp 0\cup \mtc Q\hp 0,\\
(\Rb  \tensor*[_e]{\#}{_{\!f}} \mtc Q) \hp 1 & = 
(\Rb\hp 1 \setminus \{e\})\cup (\mtc Q \hp 1 \setminus\{ f\}) \cup \{ E,F\},
\end{align*}
being $E$ and $F$ new $k$-colored edges defined by $s(E)=s(e)$,
$t(E)=t(f)$ and $s(F)=s(f)$, $t(F)=t(e)$ (see Figure
\ref{fig:ST}). Otherwise, the incidence relations and coloring are
inherited from those of $\mathcal R$ and $\mathcal Q$.  This implies
that $\Rb \tensor*[_e]{\#}{_{\!f}} \mtc Q$ is a connected graph in
$\Grph{D}\hp{2l+2l}$. \end{definition}

It is obvious that if one chooses only color-$0$ edges $e$ and $f$,
one can restrict $\#$ to a well-defined binary operation on the set of
Feynman graphs,
\[ \tensor*[_e]{\#}{_{\!f}}
: \fey_D(V)\times \fey_D(V)\to \fey_D(V),
\]
 for arbitrary rank-$D$ colored (complex) tensor model $V(\phi,\bar\phi)$.

\begin{figure}\centering
\includegraphics[height=3.0cm]{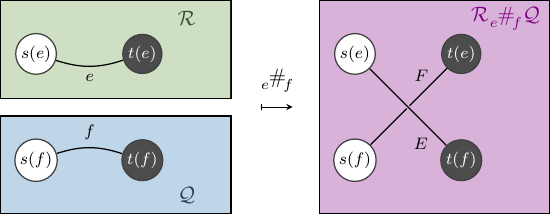} 
\caption{ On the definition of $\#$. Here $s$ and $t$ are source and target,
respectively  \normalsize \label{fig:ST}
}
\end{figure}
This operation $\#$ was defined in \cite{cips} for $3$-colored graphs
that are Feynman diagrams of rank-$2$ tensor models. It is
straightforward to check that $\#$ is associative. The notation is due
to the fact that on $\Grph{3}\times \Grph{3}$, $\#$ is the
graph-theoretical connected sum.  We use it now in higher dimensions,
but for $D\geq 3$, we (still) do not interpret $\#$ as connected sum.
We have, nevertheless the following result, which for $D=2$ has been
proven in \cite[Lemma 3]{cips}.

\begin{proposition} \label{thm:suma_grado}
For arbitrary edges $e\in \G_c\hp 1$, $f\in\K_c\hp 1$ of any color $c$, 
the operation $\tensor*[_e]{\#}{_{\!f}}$ behaves additively with respect to 
Gur\u{a}u's degree i.e. 
$\omega(\G \tensor*[_e]{\#}{_{\!f}} \K) =\omega(\G)+\omega(\K)$, for 
any graph $\G,\K \in \Grph{D+1}$.
\end{proposition}

\begin{proof}
  We use the face-counting formula \eqref{eq:facecounting} to
  calculate Gur\u{a}u's degree and compute how it changes after
  $ \tensor*[_e]{\#}{_{\!f}}$.  First, notice that the vertices of
  $\G$ and $\K$ add up exactly to those of
  $\G \tensor*[_e]{\#}{_{\!f}} \K$.  Concerning faces, in $\G$ there
  are exactly $D$ two-bubbles containing the edge $e$, namely the
  connected component $\G_{e}^{(cd)}$, where $d$ is any color but $c$
  itself.  By the same token, there are $D$ faces of $\K$ whose
  boundary loop contains $f$. Erasing $e$ and $f$ in favor of the
  edges $E$ and $F$ in $\G \tensor*[_e]{\#}{_{\!f}} \K$ puts the
  bubbles $\G_{e}\hp{cd}$ and $\K_f\hp{cd}$ together in a single
  one. This happens for each color $d\neq c$, whence
  $ |\G\hp2| + |\K\hp2|-D=|(\G \tensor*[_e]{\#}{_{\!f}} \K)\hp2 | ={D
    \choose 2 } |(\G \tensor*[_e]{\#}{_{\!f}} \K)\hp 0|/2+D-
  2\omega(\G \tensor*[_e]{\#}{_{\!f}} \K)/(D-1)!\,$.
  Then using formula \eqref{eq:facecounting} for both $\G$ and $\K$
  yields the result. \end{proof}

\begin{example}
  We consider two copies of the $(D+1)$-colored graph with two
  vertices, $\mathcal M$. It has only planar jackets, whence its Gur\u
  au's degree is zero. Therefore, if $e_i$ denotes the only color-$i$
  edge of $\mathcal M$, by Proposition \ref{thm:suma_grado}, one has
  $\mathcal P= \mathcal M \tensor*[_{e_1}]{\#}{_{\!{e_1}}}\mathcal M
  \tensor*[_{e_0}]{\#}{_{\!{e_0}}}\mathcal M
  \tensor*[_{e_D}]{\#}{_{\!{e_D}}}\mathcal M $
  is a melon, for
  $\omega(\mathcal P)=\omega(\mathcal M
  \tensor*[_{e_1}]{\#}{_{\!{e_1}}}\mathcal M)+ \omega(\mathcal M
  \tensor*[_{e_D}]{\#}{_{\!{e_D}}}\mathcal M)= 4\omega(\mathcal
  M)=0$.
  This graph will be handy in the sequel (in
  Eq. \eqref{eq:boundary_csum}, specifically) in order to separate
  boundary components (see Lemma \ref{thm:separatrix}). By a similar
  argument one can see that the vacuum graph in example
  \ref{ex:inicial} is a melon. Since melons triangulate spheres
  \cite{GurauRyan}, our claim there is proven.
\end{example} 
 
\subsection{Colored graph automorphisms} \label{sec:automorphisms} The
available concept of automorphism in the theory of manifold
crystallization \cite[Sec. 1]{survey_cryst} and graph-encoded
manifolds of the late 70s and early 80s cannot be used here, for
boundary graphs $\partial \mathcal C$ have a bipartite-vertex set
(which is moreover labeled by the momenta corresponding to the ones
carried by open legs of $\mathcal C$; see Sec. \ref{expansion_rk3});
here we introduce the concept that discloses the compatibility with
the whole CTM-structure.
\begin{definition} 
  An \textit{automorphism} $\Theta$ of a graph $\G\in \Grph{D}$ is a
  couple of permutations $\Theta=(\theta,\tilde \theta)$ of the set of
  vertices $\theta \in \mathrm{Sym}(\G\hp0)$ and the set of edges
  $\tilde \theta \in \mathrm{Sym}(\G\hp1)$ that respects
  \begin{itemize}
 \item  \textit{bipartiteness}:
   $\theta|_{\G\hp0_{\vphantom{b}\mathrm{w}}}\in
   \mathrm{Sym}(\G\hp0_{\vphantom{b}\mathrm{w}})$ and
   $\theta|_{\G\hp0_{\vphantom{b}\mathrm{b}}}\in
   \mathrm{Sym}(\G\hp0_{\vphantom{b}\mathrm{b}})$,
 \item  \textit{edge-coloring}: for any color $c$ and 
 $e_c\in\G\hp1_c$, then $\tilde \theta(e_c)\in \G\hp1_c$,
 \item  \textit{adjacency:} let $s:\G\hp 1\to
   \G\hp0_{\vphantom{b}\mathrm{w}}$ and $t:\G\hp 1\to
   \G\hp0_{\vphantom{b}\mathrm{b}}$ respectively denote the source and
   target maps. Then the following diagrams are commutative:
 \[
 \includegraphics[width=6.6cm]{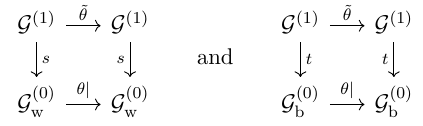}
\]
\end{itemize}
\end{definition}
 
We denote by $\Autc(\G)$ the group of automorphisms of the colored
graph $\G$.  Notice that $\Theta\in \Autc(\G)$ has no more information
than a permutation of white (or black) vertices plus ``preserving the
structure of colored graph''. That is to say, let $r=|\G\hp0|/2$ and
suppose that $\tau\in \Sym_{r}$ is such that there exists an
automorphism $\Theta=(\theta, \tilde \theta) \in\Autc(\G)$ that
restricts to $\tau$, $\theta|_{\G\hp0_{\vphantom{b}\mathrm{w}}}=\tau$.
We construct the other pieces of $\Theta$, beginning with
$\tilde \theta$. For an arbitrary color $j$, let $e_j$ be an edge in
$ \G\hp1_j$.  Set then
\[
\tilde\theta (e_j) := 
\mbox{ the only $j$-colored edge in $s\inv (\tau(s(e_j)))$ .}
\]
In terms of $\tilde\theta$, we define $\theta$ for black vertices: let
$p\in\G\hp0_{\vphantom{b}\mathrm{b}}$ and let, for arbitrary color
$j$, $f_j\in\G_j\hp1$ be the edge with $p=t(f_j)$. Then set
$ \theta(p):=t(\tilde\theta(f_j)).  $ That is, $\theta$ and
$\tilde \theta$ can be constructed from $\tau$.  We conclude that for
connected graphs $\G\in\Grph{D}$, if $\tau$ can be lifted to a
$\Theta\in\Autc(\G)$, then $\Theta$ is unique and (whenever it exists)
it will be denoted by $\hat \tau$.  This way we can see $\Autc(\G)$ as
a subgroup of $\Sym_r=\mathrm{Sym}(\G_{\vphantom{b}\mathrm{w}}\hp0)$.
In particular, the following bound holds:
\begin{equation}
|\Autc(\G)| \leq  (|\G\hp 0 \wh|) ! =\big( |\G \hp0| /2 \big)!\,\,\,.
\end{equation}

\begin{example}
  By contrast with the `uncolored' utility graph $K(3,3)$, for which
  $|\Aut(K_{ }(3,3))|=2(3!)^2$, one has for its color version
  $K_{\mathrm c}(3,3)$ a quite modest
  $\Autc(K_{\mathrm c}(3,3)) \cong \Z_3$.  The two non trivial
  elements of $\Autc(K_{\mathrm c}(3,3))$ are rotations by
  $\pm2\pi/ 3$.  The rotations by $\pm \pi/3, \pi$ are forbidden by
  edge-coloring preservation.  On the other hand, reflections about
  the depicted axes do preserve edge-coloring but not the
  bipartiteness of the edges:
\[ K_{\mathrm c}(3,3)= 
\raisebox{-.47\height}{ \includegraphics[width=2.6cm]{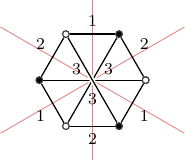} }
\]
The following lines complete the short list of automorphism groups of
connected graphs in $\leq 6$ vertices; there $d=1,2,3$ and $R_\theta$
means anti-clockwise rotation by $\theta$:
\begin{align*}
\Autc \big( 
\raisebox{-.4\height}{\includegraphics[width=.8cm]{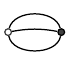}}
\big)&=\{*\}\,, & 
\Autc \Big( 
\raisebox{-.38\height}{\includegraphics[width=1.2cm]{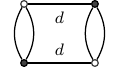}} 
\Big)&=  \langle R_\pi\rangle  \simeq\Z_2\,, \\
\Autc \Big(
\raisebox{-.4\height}{\includegraphics[width=1.0cm]{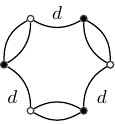}} 
\Big)&=\langle R_{2\pi/3}\rangle\simeq \Z_3 \,, & 
\Autc \Big( 
\raisebox{-.4\height}{\includegraphics[width=1.4cm]{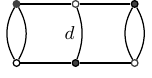}} 
\Big)&=\{*\}\,.
\end{align*}
 \end{example}
More extensive tables of automorphism groups of connected
colored graphs, as well as their Gur\u au's degree, can be found in \cite{SDE}.
If $\G\in \amalg \Grph{D} $
is the disjoint union of $m_i$ copies of pairwise distinct 
types of connected graphs $\{\Gamma_i\}_{i=1}^s\subset \Grph{D}$, $\G= 
(\Gamma_1\sqcup \ldots\sqcup\Gamma_1 )\sqcup 
\ldots \sqcup (\Gamma_{s}\sqcup\ldots \sqcup \Gamma_{s}) $, then
\begin{equation}\label{eq:corona}
 \Autc(\G)=(\Autc(\Gamma_1)^{ } \, \wr \, \Sym_{m_1}) 
\times (\Autc(\Gamma_2)^{ } \, \wr \, \Sym_{m_2}) \times 
\ldots \times (\Autc(\Gamma_s)^{ }\,\wr \, \Sym_{m_s})\,,  
\end{equation}
where $\wr$ is the wreath product of groups. 
Hence $|\Autc(\G)|=\prod_{i=1}^s(m_i)!\cdot |\Autc(\Gamma_i)|^{m_i }$.

\section{Completeness of the boundary sector for quartic interactions}
\label{sec:completeness}

\begin{definition}\label{def:boundarysector}
  The \textit{boundary sector} of a rank-$D$ colored tensor model
  $V(\phi,\bar\phi)$ is the image of the map
  $\partial: \fey_D(V) \to \amalg\Grph{D}$.
\end{definition}

In \cite{cips} it has been shown, constructively, that the geometric
realization of the boundary sector of the $\phi_3^4$-theory is enough
to reconstruct all orientable, closed (possibly disconnected)
surfaces.  Here we present, first, a stronger result in Section
\ref{sec:phi34} for $D=3$. A similar statement with a similar proof
for the rank $D>3$ case follows in Section \ref{sec:BS_D}. Both
results are needed for the Ward-Takahashi Identity.

\subsection{The boundary sector of the $\phi_3^4$-theory} 
\label{sec:phi34}

\begin{lemma} \label{thm:completeness3} 
  Every connected $3$-colored graph is the boundary of (at least) one Feynman
  diagram of the $\phi_3^4$-theory. In other words, 
  the boundary sector contains $\tcol$.
\end{lemma}
\begin{proof}
  Let $\mathcal R$ be a connected $3$-colored graph.  If
  $\R=\emptyset$ is the empty graph, trivially, one can pick any
  closed (or vacuum) graph $\tilde \R$ of the model.  Assume then,
  that $\R$ is not the empty graph.  We construct
  $\tilde {\mathcal R}$ so that
  $\partial \tilde {\mathcal R} =\mathcal R$. To each white
  (resp. black) vertex $d\in \R\hp 0_{\mathrm w \vphantom b}$ (resp.
  $ x\in\R\hp 0_{\mathrm b}$) we associate the following contractions:
\[
\tilde d(c_1,q_1,c_3,c_2,q_2)= \raisebox{-.569\height}{\includegraphics[height=2.3cm]{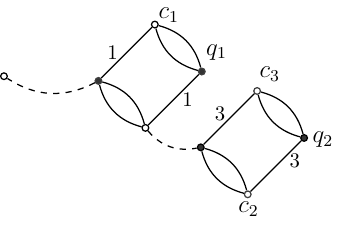}},\quad
\tilde x(b_2,p_2,p_3,b_1,p_1)=\raisebox{-.569\height}{\includegraphics[height=2.3cm]{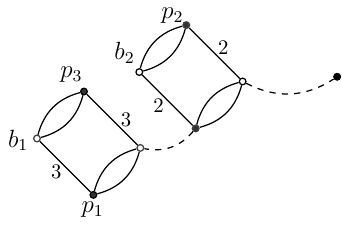}}
\]
The edges of any of the three colors are associated the following
Wick-contractions. For $e_1 \in \mathcal R_1\hp1$ 
\[
e_1\mapsto  
\contraction{\tilde d(}{i_1}{,p_1,b_3,a_2,p_2) 
\tilde  x(c_2,q_2,r,c_1,}{d_1}
\contraction[2ex]{\tilde d(a_1,}{i_1}{,b_3,a_2,p_2) 
\tilde  x(c_2,q_2,r,}{d_1}
\tilde d(c_1,q_1,c_3,c_2,q_2) 
\tilde  x(b_2,p_2,p_3,b_1,p_1) = 
\raisebox{-.569\height}{\includegraphics[height=2.3cm]{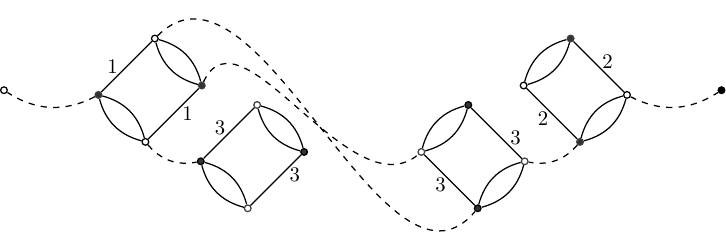}}
\]
Similarly to given $f_2 \in \mathcal R_2\hp1$ and
$g_3\in \mathcal R\hp1_3$, one associates, respectively, the following
graphs:
\[
f_2\mapsto  
\contraction{\tilde d(a_1,p_1,b_3,}{a}{,p_2) 
\tilde  x(c_2,}{q_2}
\contraction[2ex]{\tilde d(a_1,p_1,b_3,a_2,}{p}{a\tilde  x(}{d}
\tilde d(c_1,q_1,c_3,c_2,q_2) 
\tilde  x(b_2,p_2,p_3,b_1,p_1)= \raisebox{-.569\height}
{\includegraphics[height=2.2cm]{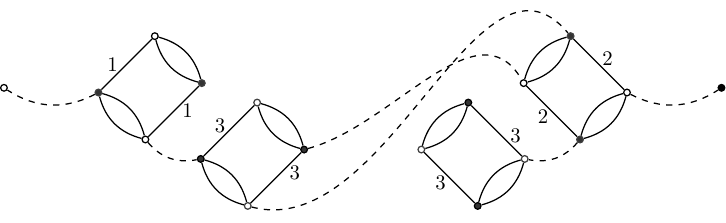}}
\]
\[g_3\mapsto 
\contraction{\tilde d(a_1,p_1,}{c}{,a_2,p_2)\tilde  x(c_2,q_2,}{p}
\tilde d(c_1,q_1,c_3,c_2,q_2)
\tilde  x(b_2,p_2,p_3,b_1,p_1) = 
\raisebox{-.629\height}{\includegraphics[height=1.8cm]{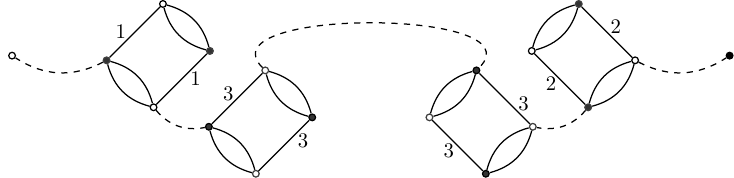}}
\]
Since each vertex $v\in \R$ is regularly $3$-colored, the five Wick
contractions added to $\tilde v$ saturate all but one irregularly
colored vertices in $\tilde v$ and make them regularly colored.  The
only one that remains is a leaf and will be an open leg.  Of course,
connectedness of two vertices $d,x\in\mathcal R\hp 0$ by an edge of
color $i$ (with $i=1,2,3$) is transferred to the connectedness of the
(unmarked) external vertices of $\tilde d$ and $\tilde x$ in
$\tilde {\mathcal R}$ by a $(0i)$-colored path in that graph. Thus, by
construction, $\partial \tilde{\mathcal R}=\mathcal R$.
\end{proof}
\begin{remark} 
  In the proof of Lemma \ref{thm:completeness3} the vertex $\V_3$ has
  been used. We suspect, there is an optimal construction, which
  it only uses $\V_1$ and $\V_2$.  
  The optimization of this proof would use the
  dipole contraction \cite[Lemma 4]{4renorm} (in that setting for
  rank-$4$ TGFTs) but we defer this proof.
\end{remark}

\subsection{The boundary sector of the $\phi_{D,\mathsf{m}}^4$-theory} \label{sec:BS_D}

\label{completeness_general}
In two dimensions there is a single (complex) quartic model; in three
dimensions, there are three interaction vertices. Both in two and
three dimensions quartic vertices are all melonic. The situation
changes from $4$ dimensions on.  For instance, in $D=4$, the
interaction vertex $\mathcal N$ given by eq. \eqref{eq:necklace} and
$\mathfrak S_4$ permutations thereof are not melonic, for their
Gur\u{a}u's degree is $\omega(\mathcal N)=1$ (see computation
\cite[Sec 2]{cips}).  For arbitrary rank, $D\geq 2$, we use the
following shortcut: the $\phi^4_{D,\mathsf{m}}$-theory denotes the
model with the following $D$ melonic vertices
$\{\Tr_{\V_k}(\phi,\bar\phi)\}_{k=1,\ldots,D}$, being
\[
\Tr_{\V_k}(\phi,\bar\phi)=\raisebox{-.45\height}{\includegraphics[width=1.15cm]{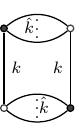}} \,\,\,.
\]
Here some of the edges with colors $\hat k=\{1,2, \ldots,D\}\setminus
\{k\}$ are shortened by dots. We also abbreviate the Feynman diagrams of
the $\phi^4_{D,\mathsf{m}}$-theory as
$\fey_D(\phi^4_{\mathsf{m}})$. For $D=3$ the subindex $\mathsf{m}$
denoting melonicity is redundant. There, the
$\phi^4_{3,\mathsf{m}}$-theory is the $\phi^4_{3}$-theory and
$\fey_3(\phi^4_{\mathsf{m}})=\fey_3(\phi^4)$, according to previous remarks.

\begin{theorem} \label{thm:completeness_disconn}
For arbitrary rank $D$, 
the boundary sector $\partial \feymel$ of the $\phi^4_{D,\mathsf{m}}$-theory
is all of $\amalg\Grph{D}$.
\end{theorem}

We need first two lemmas. The first one is most of the work and
concerns the connected case. The second lemma tells how glue
$\phi^4_{D,\mathsf m}$-Feynman graphs into a connected
$\phi^4_{D,\mathsf m}$-Feynman that has a custom (disconnected)
boundary.
\begin{lemma} \label{thm:BS_D}
The boundary sector of the $\phi^4_{D,\mathsf{m}}$-theory
contains $\Grph{D}$.
\end{lemma}

The idea is to associate, to each vertex $v$ of $\B$, a partially
Wick-contracted ``raceme'', $\tilde v$, of interaction-vertices of the
$\phi^4_{D,\mathsf{m}}$-theory. Each raceme has a marked
(graph-theoretical) vertex.  Among all the associated racemes, one
contracts with a $0$-color all but the marked vertex, in such a way
that one has a $0i$-bicolored path in $\tilde \B$ between two such
preferred vertices at racemes $\tilde x$ and $\tilde d$, whenever
there is an $i$-colored edge in $\B$ between $x$ and $d$.

\begin{proof}
Let $\B\in\Grph{D}$. We construct a graph $\tilde\B \in
\fey_D(\phi^4_{\mathsf{m}})$ with $\partial\tilde \B=\B$.  Concretely,
we assemble $\tilde \B$ from $\B$ as follows.  Only after \textit{Step
  2} we will have a well-defined Feynman graph.  \\
\textit{Step 1:}
Replace any black vertex $x\in\B\hp 0\bl$ and any white vertex
$d\in\B\hp 0 \wh$ by $\tilde x$ and $\tilde d$, respectively:
\begin{equation}\label{xtilde}
x\mapsto \tilde x=\,\,\raisebox{-.45\height}{\includegraphics[width=7.0cm]{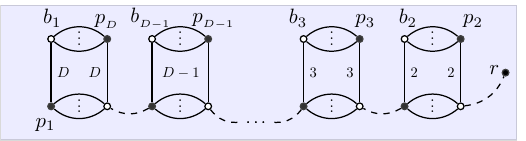}}
\end{equation}

\begin{equation}\label{dtilde}
d\mapsto \tilde d= \,\,\raisebox{-.45\height}{\includegraphics[width=7.0cm]{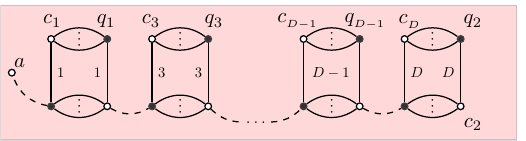}}
\end{equation}
At this stage, $\tilde \B$ consists of the following connected
components $\{ \tilde d\}_{d\in\B\hp 0\wh} \cup \{\tilde x
\}_{x\in\B\hp 0\bl} $, which, altogether, have the following set of
vertices that are \textit{not} contracted with the $0$-color:
\begin{equation}
\{b_1^x,p_1^x,\ldots,b_{D-1}^x,p_{D-1}^x,
p_D^x\}_{x\in\B\hp 0\bl}  \cup   
\{c_1^d,q_1^d,\ldots,c_{D-1}^d,q_{D-1}^d,c_D^d\}_{d\in\B\hp 0\wh}\,.
\label{abiertos}
\end{equation}
\noindent
\textit{Step 2:} We shall contract all open vertices \eqref{abiertos} 
as follows:
Whenever $x = t(e_i)$ and $d=s(e_i)$, for $e_i$ 
an edge of color $i\neq D$, $e_i\in\B\hp 1_i$, 
one Wick-contracts $b_i^x$ with $q_i^d$ and $p_i^x$ with $c_i^d$:
\begin{equation}
\contraction{\tilde x(\ldots,}{b_i^x}{,p_i^x,\ldots,b_{D-1}^x,p_{D-1}^x,
p_D^x,r^x) \,\,
\tilde d
(\ldots,c_i^d,}{q_i^d}
\contraction[1.7ex]{\tilde x(\ldots,b_i^x,}{p_i^x}{,\ldots,b_{D-1}^x,p_{D-1}^x,
p_D^x,r^x) \,\,
\tilde d
(\ldots,}{c_i^d}
\tilde x(\ldots,b_i^x,p_i^x,\ldots,b_{D-1}^x,p_{D-1}^x,
p_D^x,r^x) \,\,
\tilde d
(\ldots,c_i^d,q_i^d,c_{D-1}^d,q_{D-1}^d,c_D^d,a^d
) \label{contraction_k}
\end{equation}
Whenever  $x = t(e_D)$ and $d=s(e_D)$ for 
 $e_D\in\B\hp 1_D$, contract $p_D^x$ with $c_D^d$:
 \[ 
\contraction[1.1ex]
{\tilde x(\ldots,b_i^x,p_i^x,\ldots,b_{D-1}^x,p_{D-1}^x,}{p_D^x}{,r^x) \,\,
\tilde d
(\ldots,c_i^d,q_i^d,c_{D-1}^d,q_{D-1}^d,}{c_D^d}
\tilde x(\ldots,b_i^x,p_i^x,\ldots,b_{D-1}^x,p_{D-1}^x,
p_D^x,r^x) \,\,
\tilde d
(\ldots,c_i^d,q_i^d,c_{D-1}^d,q_{D-1}^d,c_D^d,a^d
)
\]
The regularity and the bipartiteness of $\B$ imply the
well-definedness of $\tilde\B$ as open ($D+1$)-colored graph. We now
see that $\partial \tilde\B=\B$. Indeed, for each black vertex $x$
(resp. white vertex $d$) in $\B$, there exactly is a black
(resp. white) external leg, namely $r^x$ (resp. $a^d$) which is mapped
by $\partial$ to a black vertex $\partial r^x$ (resp. white vertex
$\partial a^d$).  Therefore, $\B$ and $\partial\tilde \B$ have the
same bipartite vertex set.  To conclude, we remark that for every
$k$-colored edge $e_k$ in $\B$ between $x$ and $d$, there is indeed a
$(0k)$-bicolored path in $\tilde\B$ between $r^x$ and $a^d$, and this
ensures that there is a $k$-colored edge between $\partial a^d$ and
$\partial r^x$, by the mere definition of the boundary graph:
\begin{itemize}
\item $k=1$: From right to left in the following graph, notice that
  since the vertex $\V_1$ does not appear in $\tilde x$, there is (in
  the bottom part) there is a ($01$)-bicolored path between $r$ and $p_1$. That
  path can be concatenated with $\overline{p_1c_1 a}$, which is
  also ($01$)-bicolored.  
 (Notice that from the two Wick-contractions, only one lies on such a path. 
 The other is secondary.)  
 \[
  \includegraphics[width=11.0cm]{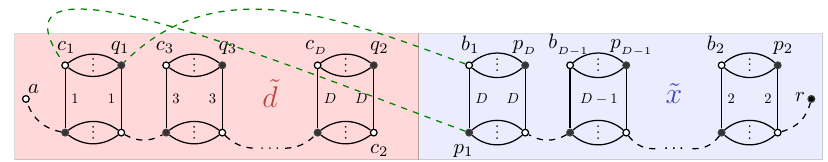}\]

 \item $k=2$: By a similar token, there is a ($02$)-colored
 path between $a$ and $c_2$. 
 \[
  \includegraphics[width=11cm]{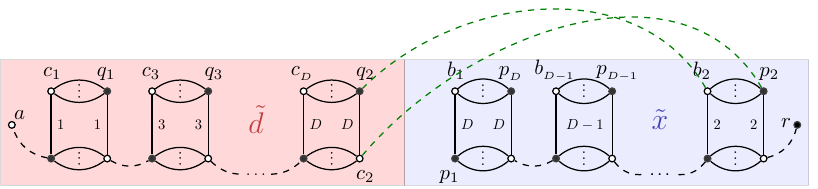}\]
 \item $k=3,\ldots,D-1$: The Wick-contraction \eqref{contraction_k}
   connects $c_k$ and $p_k$ with a color-$0$ edge. It is evident that
   from graphs \eqref{xtilde} and \eqref{dtilde}, that there is a
   $0k$-bicolored path through it that connects $r$ and $a$.
 \item $k=D$. There are $\V_D$-vertices 
 neither to the left of $c_D$ nor to the right
 of $p_D$, so there is a ($0D$)-bicolored path $\overline {ac_D}$, which 
  can be concatenated with $\overline{c_Dp_D}$ and subsequently with $\overline{p_Dr}$.
  \[
  \includegraphics[width=11cm]{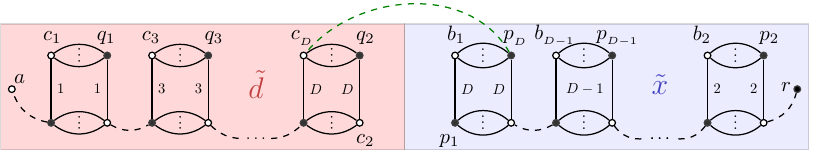} 
   \qedhere
  \]
\end{itemize}
\end{proof}

\begin{example} \label{ex:necklaceTilde} In rank-$4$, the construction
  in Theorem \ref{thm:BS_D} (Step 1) associates to each vertex
  $\{f,d,x,z\}$ of the necklace-graph
\begin{equation}\label{eq:necklace}
\mathcal N=
\raisebox{-.46\height}{\includegraphics[width=2.1cm]{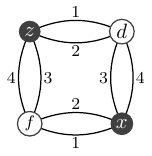}} 
\end{equation}
``racemes'' $\{\tilde f, \tilde d,\tilde x,\tilde z\}$.  
According to Step 2,
they are contracted by $0$-colored edges to form the following 
$\phi^4_{\mathsf m}$-Feynman diagram $\tilde{\mathcal N}$,
which obviously satisfies  $\partial \tilde {\mathcal N}=\mathcal N$.
\[
\tilde{\mathcal N}=\raisebox{-.5\height}{
  \includegraphics[width=11cm]{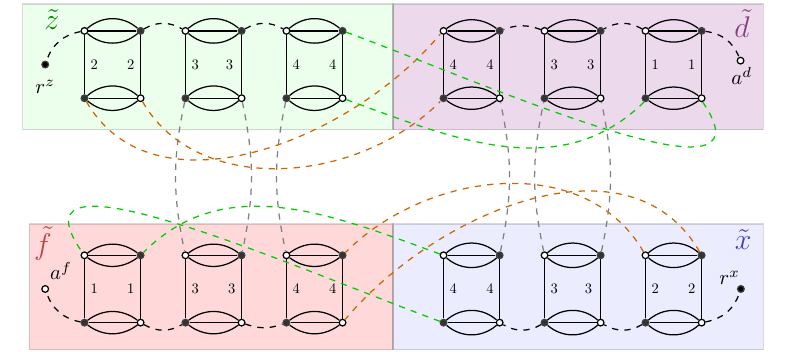} }
  \]
\end{example}

\begin{lemma} \label{thm:separatrix}
 The $\phi^4_{\mathsf{m}}$-graph $\mathcal S$ given by
\begin{equation}
 \label{separatrix}
 \mathcal S(g,v;h,w):=\raisebox{-.45\height}{
 \includegraphics[width=6.0cm]{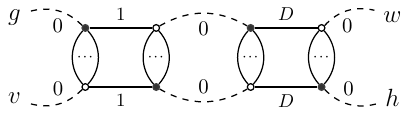}
 }
\end{equation}
separates
 boundary components. 
More precisely: Given two open graphs $\G_1,\G_2 \in \fey_D(\phi^4_{\mathsf m})$,
with $2p_1 $ and $2p_2 $ external legs, respectively, 
\[\G_1(c^{(1)},c^{(2)},\ldots,c^{(p_1)};
r^{(1)},r^{(2)},\ldots,r^{(p_1)}) \and \G_2
(d^{(1)},d^{(2)},\ldots,d^{(p_2)};
s^{(1)},s^{(2)},\ldots,s^{(p_2)}),\]
for any $1\leq i,j\leq p_1$ and $1\leq k,l\leq p_2$, being $c\hp i$
(resp. $r\hp j$) any outer white (resp. black) vertex of $\G_1$ and
$d\hp i$ (resp. $s\hp j$) any outer white (resp. black) vertex of
$\G_2$, we claim that
\[ \mathcal C:=
\bcontraction{\G_1(\ldots,}{c^{(i)}}{,\ldots ; \ldots, r^{(j)},\ldots) 
\mathcal S(a,b;}{p}
\bcontraction[1.8ex]{
\G_1(\ldots,c^{(i)},\ldots ; \ldots,}{r^{(j)}}{,\ldots) 
\mathcal S(}{a}
\contraction{\G_1(\ldots,c^{(i)},\ldots ; \ldots, r^{(j)},\ldots) 
\mathcal S(a,}{b}{
;p,q) 
\G_2(\ldots,d^{(k)}\ldots;\ldots,}
{s^{(l)}}
\contraction[2.0ex]{\G_1(\ldots,c^{(i)},\ldots ; \ldots, r^{(j)},\ldots) 
\mathcal S(a,b;p,}{q\vphantom{d}}{) 
\G_2(\ldots,}{d^{(k)}}
\G_1(\ldots,c^{(i)},\ldots ; \ldots, r^{(j)},\ldots) 
\mathcal S(g,v;h,w) 
\G_2(\ldots,d^{(k)},\ldots;\ldots,s^{(l)},\ldots) 
\]
is a Feynman graph in $\fey_D(\phi^4_{\mathsf{m}})$, whose boundary is given by
\begin{equation}
\partial (
\contraction{\G_1(c^{(1)},\ldots,}{c}{{}^{(i)},\ldots ;r^{(1)} \ldots,}{r}
\G_1(c^{(1)},\ldots,c^{(i)},\ldots ; r^{(1)}\ldots, r^{(j)},\ldots) )
 \sqcup \partial(
 \contraction{\G_2(d\hp{1}\ldots,}{\hphantom{d}d}{{}^{(k)},\ldots;s\hp{1}\ldots,}{s}
\G_2(d\hp{1},\ldots,d^{(k)},\ldots;s\hp{1},\ldots,s^{(l)},\ldots) 
) \,.
\label{dos_fronteras}
 \end{equation}
\end{lemma}
Thus, if $\mathcal{C} $ is given by
\[
  \begin{tikzpicture}%
    [
    baseline=-1.2ex,shorten >=.1pt,
    semithick,auto,
    every state/.style={fill=white,draw=texto,inner sep=.3mm,text=texto,minimum size=0},
    accepting/.style={fill=white,text=black},
    initial/.style={white,text=texto}
]
    \draw [dashed] (0,0) to[bend right] (1,-.6);
    \draw [dashed] (0,0) to[bend left]  (1,.6);
    \draw [dashed] (0,0) to[bend right] (-1,.6);
    \draw [dashed] (0,0) to[bend left] (-1,-.6);
    \draw [dashed] (-2,0) to[bend left] (-1,.6);
    \draw [dashed] (-2,0) to[bend right] (-1,-.6);
    \draw [dashed] (2,0) to[bend right] (1,.6);
    \draw [dashed] (2,0) to[bend left] (1,-.6);
    \draw [dashed] (2,0) -- (3,.6);
   \draw [dashed] (2,0) -- (3,-.6);
       \draw [dashed] (-2,0) -- (-3,.6);
   \draw [dashed] (-2,0) -- (-3,-.6);
   \node at (-3.,0.1) {\scriptsize $\vdots$};
   \node at (3.,0.1) {\scriptsize $\vdots$};
   \node[state,accepting  
   ] at (0,0) {$\,\, {\mathcal S} \,\,$};
    \node[state,accepting 
    ] at (2,0) {$\,\, {\mathcal \G}_2\,\,$};
    \node[state,accepting 
    ] at (-2,0) {$\,\, {\mathcal \G}_1\,\,$};
  \end{tikzpicture}
\]
the dots listing uncontracted external legs, Lemma \ref{thm:separatrix} says that
\[
\partial \mathcal{C} =\partial\left(\raisebox{-.15\height}{
 \begin{tikzpicture}
    [
    baseline=-1.2ex,shorten >=.1pt,  
    semithick,auto,
    every state/.style={fill=white,draw=texto,inner sep=.3mm,text=texto,minimum size=0},
    accepting/.style={fill=white,text=black},
    initial/.style={white,text=texto}
]
 
   \draw [dashed] (-1.2,0) circle (.5);
   \draw [dashed] (-2,0) -- (-3,.6);
    \node at (-3.,0.1) {\scriptsize $\vdots$};
   \draw [dashed] (-2,0) -- (-3,-.6);
   \node[state,accepting  
   ] at (-2,0) {$\,\, {\mathcal \G}_1\,\,$};
  \end{tikzpicture}
 }\right) \sqcup \partial
\left(\raisebox{-.15\height}{
 \begin{tikzpicture}%
    [
    baseline=-1.2ex,shorten >=.1pt,
    semithick,auto,
    every state/.style={fill=white,draw=texto,inner sep=.3mm,text=texto,minimum size=0},
    accepting/.style={fill=white,text=black},
    initial/.style={white,text=texto}
]
    \draw [dashed] (2,0) -- (3,.6);
   \node at (3.,0.1) {\scriptsize $\vdots$};
   \draw [dashed] (2,0) -- (3,-.6);
   \draw [dashed] (1.2,0) circle (.5);
  \node[state,accepting  
  ] at (2,0) {$\,\, {\mathcal \G}_2\,\,$};
  \end{tikzpicture}
 }\right)
\]

\begin{proof}
  This is a restatement of  \cite[Lemma 6]{cips}.
  \end{proof}
One can restate a more general result by considering
$\mathcal P=\contraction{\mathcal{S} (}{g}{\,,}{} \contraction{\mathcal{S}
  (g,v;}{h}{\,}{w} \mathcal{S} (g,v;h,w)$
and taking ($D+1$)-colored
graphs $\K$ and $\G$ that might even be closed. By taking edges
$e \in \K_0\hp 1$ and $f\in \G_0\hp 1$ and letting
$k=\overline{gw}$ and $l=\overline{h w}$ one has 
\begin{equation}\label{eq:boundary_csum}
\partial (
\K \tensor*[_{e}]{\#}{_{k}} \mathcal P
\tensor*[_{l}]{\#}{_{\!f}} \G
) = \partial \K \sqcup \partial \G\,. 
\end{equation}
\begin{proof}[of theorem \ref{thm:completeness_disconn}]
  Let $\B$ be an arbitrary graph in $\amalg\Grph{D}$.  We decompose
  $\B$ in its connected components
  $\{\R^\alpha\}_{\alpha=1}^B \subset \Grph{D}$,
  $\B=\sqcup_{\alpha=1}^B \R^\alpha$. For each connected component
  $\alpha$, we consider the graphs $\tilde \R^\alpha$ given by Lemma
  \ref{thm:completeness3} if $D=3$ or by Lemma \ref{thm:BS_D} if
  $D\geq 4$. Fix two arbitrary vertices
  $d^\alpha \in(\R^\alpha)\wh\hp 0$ and
  $x^\alpha\in(\R^\alpha)\hp0\bl$ and consider the vertices
  $c_1^\alpha,q_1^\alpha$ and $b_1^\alpha,p_1^\alpha$ that lie on the
  racemes $\tilde d^\alpha$ and $\tilde x^\alpha$ of $\tilde\R^\alpha$
  respectively.  One considers also the $0$-colored edges
  $e_\alpha=\overline{c_1^\alpha p_1^\alpha}$ and
  $f_\alpha=\overline{q_1^\alpha b_1^\alpha}$ that connect the racemes
  $\tilde d^\alpha$ with $\tilde x^\alpha$.  Consider, $B-1$ copies of
  $\mathcal P$,
\[\mathcal P^i= \contraction{\mathcal{S} (}{g_i}{\,,}{}
\bcontraction{\mathcal{S} (g_i,v_i;}{h_i}{\,}{w_i} \mathcal{S}
(g_i,v_i;h_i,w_i)\,  \qquad  (i=1,\ldots, B-1)\,,\] and denote by $k_i=\overline{g_i\,v_i}$ and
$l_i=\overline{h_i\,w_i}$ the $0$-edges arising from the Wick-contracting. Then 
\[
\mathcal T = (\tilde \R^1)  \tensor*[_{e_1}]{\#}{_{\!k_1}} (\mathcal P^1) 
\tensor*[_{l_1}]{\#}{_{\!f_2}}
(\tilde \R^2)  \tensor*[_{e_2}]{\#}{_{\!k_2}} (\mathcal P^2)
\tensor*[_{l_2}]{\#}{_{\!f_3}}
(\tilde \R^3)  \tensor*[_{e_3}]{\#}{_{\!k_3}} 
\cdots (\mathcal P^{B-1}) 
\tensor*[_{l_{B-1}}]{\#}{_{\!\,f_B}}(\tilde \R^B) 
\]
implies, after repetitively using eq. \eqref{eq:boundary_csum},
\[
\partial \mathcal T= \partial \tilde\R^1 \sqcup\ldots \sqcup \partial \tilde \R^B
=\R^1\sqcup \ldots \sqcup \R^B=\B. 
\qedhere
\]
\end{proof}

\subsection{Geometric interpretation} 
\label{sec:geom_interpretation}
Graphs in $\amalg\Grph{D+1}$ serve to construct triangulations
$\Delta(\G)$ of $D$-(pseudo)manifolds, i.e. a (pseudo)complex as
stated in \cite{survey_cryst}:
\begin{itemize}
 \item  for each vertex $v\in \G\hp 0$, add a
   $D$-simplex $\sigma_v$ to $\Delta(\G)$
 \item  one labels the vertices of $\sigma_v$ by
   the colors $\{0,1,\ldots,D\}$
 \item  for each edge $e_c\in \G\hp1_c$ of
   arbitrary color $c$, one identifies the two $(D-1)$-simplices
   $\sigma_{s(e_c)}$ and $\sigma_{t(e_c)}$ that do not contain the
   color $c$.
 \end{itemize} 
 \begin{corollary}
  The boundary sector of the $\phi^4_{D,\mathsf{m}}$-model generates all orientable, closed
  piecewise linear manifolds.  Thus, for $D=4$, it generates all
  orientable, closed $3$-manifolds.
\end{corollary}
\begin{proof}
  By Pezzana's theorem \cite{pezzana,survey_cryst}, all compact,
  connected PL-$(D-1)$-manifolds possess a suitable crystallization.
  Crystallizations are, in particular, $D$-colored graphs, all of
  which are generated by certain boundary $\partial \G \in \partial
  \fey_D(\phi^4_{\mathsf{m}})$, by Theorem
  \ref{thm:completeness_disconn}. The second statement follows from
  Moise's theorem \cite{moise} on equivalence of topological and PL
  $3$-manifolds.
\end{proof}

\begin{example} \label{ex:4bordism} This result implies that there
  exist a $4$-dimensional $\Psi$-manifold that is represented by the
  $\phi^4_{\mathsf m}$-theory, whose boundary is any closed,
  orientable (honest) $3$-manifold.  In particular, for instance, the
  $3$-manifold with the following, say, three connected components: a
  lens space, $L_{3,1}$; the $3$-manifold with cyclic infinite
  fundamental group, $\mathbb{S}^2\times \mathbb S^1$; and a more
  common prime factor, $\mathbb S^3$.  First one needs to crystallize
  them. The next three are crystallizations of said manifolds, in
  which we represent the color $4$ by a waved line and suppress
  redundant labels:
\begin{equation} \label{eq:crystallizations}
 \Gamma= \hspace{-0.50cm}\raisebox
{-.475\height}{  
\includegraphics[width=3.2cm]{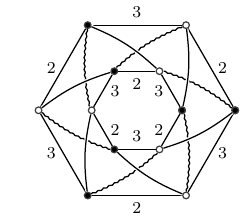}
}
\qquad \qquad
 \mathcal{C}=\hspace{-.5006cm}\raisebox
{-.45\height}{
\includegraphics[width=3.0cm]{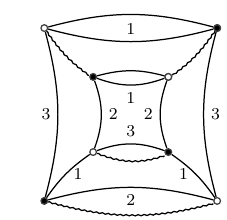}
}
\hspace{0.270cm}\qquad
 \mathcal{M}=\hspace{-0.70cm}\raisebox
{-.4015\height}{
\includegraphics[width=2.5cm]{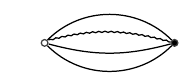}
}
\end{equation}
We compute the fundamental group of these  
crystallizations in Appendix \ref{app:fundamental}. 
Theorem \ref{thm:completeness_disconn} states that 
 $\G=\tilde \Gamma \# \mathcal P \# \tilde{ \mathcal M} \# \mathcal P\# \tilde{\mathcal C}$ 
 has as boundary the disjoint union of these graphs. Therefore 
 the geometric realization $|\Delta(\G)|$ of $\G$ has 
 as boundary $L_{3,1}\sqcup (\mathbb S^2\times \mathbb S^1) \sqcup \mathbb S^3$.
 \end{example}
 For further applications it might be important to modify Gur\u au's
 degree of a graph while simultaneously sparing its boundary. This is
 also due to the relevance of the difference
 $\tilde\omega(\G)-\omega(\partial \G)$, where $\tilde\omega$ is the
 degree for open graphs defined as the sum of the genera of its
 \textit{pinched jackets} \cite{4renorm}.  On closed graphs
 $\tilde \omega$ is the same as $\omega$ (closed jackets cannot be
 ``pinched'').  The remark is that one can modify any graph
 $\G\in\fey_D(\phi^4_{\mathsf m})$ into a graph $\G'$ of the same
 quartic model, so that $\partial\G'=\partial\G$.  The only ingredient
 one needs is a vacuum graph
 $\mathcal L \in \fey_D(\phi^4_{\mathsf m})$ with
 $\omega(\mtc{L})>0$. (e.g. in $D=3$, the necklace graph with the
 color 4 equal to 0, and being thus in $\feyn$,
 cf. ex. \ref{ex:necklaceTilde}). Then
\[
\G'=\G\# \mathcal P \# \mathcal L=\raisebox{-.150\height}{
  \begin{tikzpicture}%
    [
    baseline=-1.2ex,shorten >=.1pt,
    semithick,auto,
    every state/.style={fill=white,draw=texto,inner sep=.3mm,text=texto,minimum size=0},
    accepting/.style={fill=white,text=black},
    initial/.style={white,text=texto}
]
    \draw [dashed] (0,0) to[bend right] (1,-.6);
    \draw [dashed] (0,0) to[bend left]  (1,.6);
    \draw [dashed] (0,0) to[bend right] (-1,.6);
    \draw [dashed] (0,0) to[bend left] (-1,-.6);
    \draw [dashed] (-2,0) to[bend left] (-1,.6);
    \draw [dashed] (-2,0) to[bend right] (-1,-.6);
    \draw [dashed] (2,0) to[bend right] (1,.6);
    \draw [dashed] (2,0) to[bend left] (1,-.6);
       \draw [dashed] (-2,0) -- (-3,.6);
   \draw [dashed] (-2,0) -- (-3,-.6);
   \node at (-3.,0.1) {\scriptsize $\vdots$};
   \node[state,accepting  
   ] at (0,0) {$\,\, {\mathcal S} \,\,$};
    \node[state,accepting 
    ] at (2,0) {$\,\, {\mathcal L} \,\,$};
    \node[state,accepting 
    ] at (-2,0) {$\,\, {\mathcal \G} \,\,$};
  \end{tikzpicture}} 
  \]
  has degree $\tilde\omega(\G')=
  \tilde\omega(\G) +\omega(\mathcal P)+ \tilde\omega(\mathcal L)= 
   \tilde\omega(\G) +  \tilde\omega(\mathcal L)> \tilde \omega (\G)$, by 
Theorem \ref{thm:suma_grado}, and 
  $\partial(\G')=\partial \G \sqcup \partial \mathcal L= \partial\G$
  by Theorem \ref{thm:separatrix} and because $\mathcal L$
  is a vacuum graph. Notice that 
  the degree cannot be increased by an 
  arbitrary amount, but only by multiples of $2 /(D-1)!$ .

\section{The expansion of the free energy in boundary graphs} \label{sec:ciclos}

Before tackling the main problem, it will be useful to recall the
expansion of the free energy for real matrix models.  The reader in a
hurry might accept eq.  \eqref{expansion_W} and go to
eq. \eqref{fuentes_graficas} for notation.

\subsection{The free energy expansion for a general real matrix
  model} \label{general_matrix} As background, consider the following
model, whose objects are compact operators $M:\H\to\H$ (``matrices''),
with $\H$ a separable Hilbert space. The interactions are described by
a polynomial potential, $P(M)$. The partition function reads
\begin{equation}
\frac{Z[J]}{Z[0]}=\frac{\int \Df{M}\,\, \ee^{\Tr{(JM)}-\Tr{(EM^2)}-
\Tr\, P(M)}}{\int \Df{M}\,\, \ee^{-\Tr{(EM^2)}- \Tr\,P(M)}}\,, \label{GW_as_matrix} 
\end{equation}
where $E$ is a Hermitian operator on $\H$. The free energy,
$W_{\mathrm{matrix}}[J]\varpropto \log (Z[J]/Z[0])$, generates the
connected Green's functions. To expand in terms of the combinatorics
of the sources' indices, we shall use a multi-index notation, with
$\mathbb P_{}^m$ having \textit{length} $m=|\mathbb P_{}^m|$.  This
just means that $\mathbb P^m$ is an $m$-tuple
$\mathbb P_{}^m=(p_1p_2\ldots p_m)\in I^m$ for given index set
$I$. $I^m$ is often the integer lattice, and $m$ will not be a fixed
integer, but we will deal with multi-indices of arbitrary length.  To
enumerate multi-indices of the same length we use a subindex, so
$\mathbb{P}^m_1,\mathbb{P}^m_2,\ldots,\mathbb{P}^m_{n_m}$ are all
length-$m$ cycles.  Sums over multiple multi-indices are understood as
follows:

\begin{table}\centering
\includegraphics[width=.8916\textwidth]{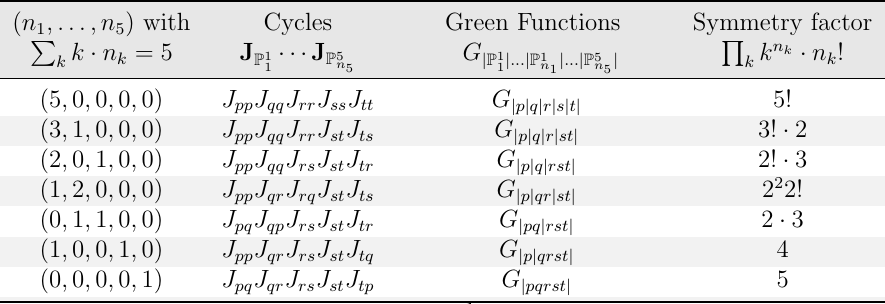}
\caption{  The boundary-graph--expansion's fifth order  \label{tab:fifth}}
\end{table} 
\[\sum\limits_{\mathbb P_{}^k,\ldots, \mathbb Q_{}^m}=\sum_{p_1}\cdots\sum_{p_k}\cdots\sum_{q_1} \cdots \sum_{q_m}
\with \mathbb{P}^k={(p_1\ldots p_k)} \and \mathbb{Q}^m=(q_1\ldots q_m)\,.\]
The $J$-cycles of size $\ell$, namely $J_{p_1p_2}J_{p_2p_3}\cdots J_{p_{\ell-1}p_\ell} J_{p_{\ell}p_1}$, 
are here for sake of briefness denoted by 
$\mathbf{J}_{\mathbb{P}^\ell} := \prod_{i=1}^\ell J_{p_{i}p_{i+1}}\,\with\, \mathbb 
P^\ell=(p_1\ldots p_\ell) \mbox{ and } p_{\ell+1}:=p_{1}\,.$
With that notation, the free energy can be expanded \cite[Sec. 2.3]{gw12} in
length-$\ell$ cycles, with $\ell$ variable, as follows:
 \begin{equation} \label{madre}
\sum\limits_{\ell=1}^\infty
 \sum\limits_{n_\ell\geq 1}^\infty  
\sum\limits^\infty_{\substack{{n_1= 0} \\ \ldots \\
            n_{\ell-1}=0\,\,\,\,\,}} \!\!\!\!
        \left[ \,\prod\limits_{j=1}^\ell\frac{1}{n_j! j^{n_j}} \right]
\sum\limits_{\substack{\mathbb{P}_1^1,\ldots,\mathbb{P}^1_{n_1}\\ \!\phantom{\!\!}_{{\phantom{a}}^{\hdots}}  
\\ \mathbb{P}^\ell_1,\ldots,\mathbb{P}^\ell_{n_\ell}  }}  \left\{  
G^{(\mathcal N_{\mtr{matrix}})}_{|\mathbb{P}^1_1|\ldots |\mathbb{P}^1_{n_1}|\ldots |\mathbb{P}^\ell_1|\ldots |\mathbb{P}^\ell_{n_\ell}|} 
\prod_{k=1}^\ell (\mathbf{J}_{\mathbb{P}^k_{1}}\cdots \mathbf{J}_{\mathbb{P}^k_{n_k}}) 
\right\}.
\end{equation}
One word more on notation: Fixed the $\ell$ by the first sum, for
$1\leq k \leq \ell$, the non-negative integer $n_k$ stands for the
number of boundary components with $k$ sources (whence $n_\ell\neq 0$
in the second sum is precisely a way to paraphrase the decomposition
in the longest cycle).  The number of boundary components
$B_{\mtr{matrix}}$, and the number of sources,
$\mathcal N_{\mtr{matrix}}$ (i.e. the order of the Green's function)
are
$B_{\mtr{matrix}}=\sum_{j=1}^\ell n_j \andd \mathcal N_{\mtr{matrix}}=
\sum_{j=1}^\ell j\cdot n_j.$
Instead of expanding by longest-cycles, we can also rephrase
\eqref{madre} as an explicit Taylor expansion to $\mathcal{O}(J^6)$,
\begin{align*}
 W[J]=&\sum_p G_{|p|}J_{pp}+\frac{1}{2}\sum_{p,q} 
 \left( G_{|pq|}J_{pq}J_{qp}+G_{|p|q|}J_{pp}J_{qq}\right) \\
&+\sum_{p,q,r}\Big(\frac13 G_{|pqr|}J_{pq}J_{qr}J_{rp}+ \frac12 G_{|pq|r|}J_{pq}J_{qp}J_{rr}+
 \frac{1}{3!}G_{|p|q|r|}J_{pp}J_{qq}J_{rr}\Big) \\&+\sum_{p,q,r,s} \Big(\frac{1}{4}G_{|pqrs|}J_{pq}J_{qr}J_{rs}J_{sp}+ 
\frac{1}{3}G_{|pqr|s|}J_{pq}J_{qr}J_{rp}J_{ss} 
\\ 
 \hphantom{ W[J]=}& +\frac1{8 } G_{|pq|rs|}J_{pq}J_{qp}J_{rs}J_{sr} 
 + \frac1{4} G_{|p|q|rs|}J_{pp}J_{qq}J_{rs}J_{sr} +\frac{1}{4!}G_{|p|q|r|s|}J_{pp}J_{qq}J_{rr}J_{ss}\Big) 
 \\
 \hphantom{ W[J]=}& 
 +\sum_{p,q,r,s,t} \bigg(	
 \frac{1}{5} G_{|pqrst|}J_{pq}J_{qr}J_{rs}J_{st}J_{tp}  
 +\frac{1}{4}G_{|p|qrst|}J_{pp}J_{qr}J_{rs}J_{st}J_{tq}  
  \\
 \hphantom{ W[J]=}& 
 +\frac{1}{2\cdot 3} G_{|pq|rst|} J_{pq}J_{qp}J_{rs}J_{st}J_{tr}
 +\frac{1}{2^2 2!} G_{|p|qr|st|} J_{pp}J_{qr}J_{rq}J_{st}J_{ts} \\
 &  
 +\frac{1}{2! 3} G_{|p|q|rst|} J_{pp}J_{qq}J_{rs}J_{st}J_{tr}  
+\frac{1}{3!2}G_{|p|q|r|st|}J_{pp}J_{qq}J_{rr}J_{st}J_{ts}\\
 \hphantom{ W[J]=}&
+\frac{1}{5!} G_{|p|q|r|s|t|}J_{pp}J_{qq}J_{rr}J_{ss}J_{tt}\bigg)+\mathcal{O}(J^6). 
\end{align*}
Table \ref{tab:fifth} shows how to read off from
\eqref{madre}, say, the fifth power in $J$. 
The Green's function for a fixed cycle can be furthermore 
expanded in genus-$g$ sectors:
\begin{equation} \label{eq:exp_genus_matrix}  
G^{(\mathcal N_{\mtr{matrix}})}_{|\mathbb{P}^1_1|\ldots |\mathbb{P}^1_{n_1}|\ldots |\mathbb{P}^\ell_1|\ldots |\mathbb{P}^\ell_{n_\ell}|} = \sum_{g\geq 0}  G^{(\mathcal N_{\mtr{matrix}}, \,g)}_{|\mathbb{P}^1_1|\ldots |\mathbb{P}^1_{n_1}|\ldots |\mathbb{P}^\ell_1|\ldots |\mathbb{P}^\ell_{n_\ell}|} \,. 
\end{equation}
For the 5-tuple $(n_1,\dots,n_5)=(3,1,0,0,0)$, here chosen only 
to exemplify the genus expansion's meaning, 
$G\hp 5_{|p|q|r|st|}=\sum_{g\geq 0}G_{|p|q|r|st|}^{(5,\,g)}$
reads
\[
G_{|p|q|r|st|}J_{pp}J_{qq}J_{rr}J_{st}J_{ts} =
\raisebox{-.39\height}{\includegraphics[height=3.4cm]{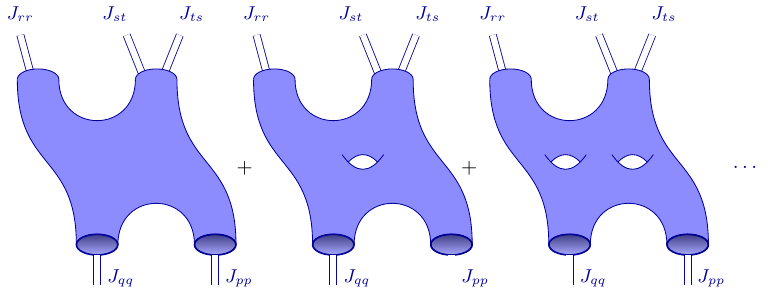}}
\]

\subsection{The general expansion for rank-$3$ CTMs} \label{expansion_rk3}

The combinatorics of the matrix-sources just shown in Section
\ref{general_matrix} gets modified for the rank-$3$ colored tensors
because of the coloring; moreover, because the theory is now complex
the non vertex-bipartite graphs are forbidden.  The first implication
of coloring is that the sources do not exhibit repeated indices in the
same source, e.g. none of the following terms is allowed in the
expansion of $W[J,\bar J]=\log(Z[J,\bar J])$:
\[
J_{\ldots a\ldots a\ldots },\,\bar J_{\ldots a\ldots a\ldots}, \bar J_{aaa}, \,J_{aaa},\,
\bar J_{aab} {J}_{bcc},\, J_{aab} \bar {J}_{bcc},\ldots \qquad\qquad \mbox{ (terms forbidden by coloring).}
\]
Whilst for the lowest order correlation functions this seems to be
quite restrictive, the expansion shows intricacy 
as one goes to higher order ones.  \par
We now consider 
a graph 
$\G\in\feyn$ and set the first convention. 
We fix the indices of the $J$-sources (the external
lines connected to the black vertices) and let $\G$ 
yield the indices of the $\bar J$-sources. 
For any $i$, both index types $\mathbf{a}^i,\mathbf p^i \in\Z^3$, are known as \textit{momenta.}
\begin{equation}
\G =
\raisebox{-.35\height}{\includegraphics[width=2.8cm]{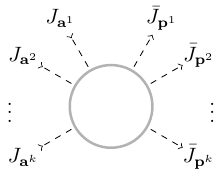}} 
\label{comb_graph}
\end{equation}
We let the notation for the $2k$-point function $G^{(2k)}_{\ldots}$
that describes the ``process'' $\G$ reflect this combinatorics via
another graph $\B$ to be constructed shortly.  The resulting
$G^{(2k)}_{\B}$ ought to encompass all graphs in $\feyn$ that lead to
the same combination of indices in the $\bar J$-sources.  From the
$\bar J$-sources, for each $\alpha=1,\ldots,k$,
$\mathbf p^{\alpha}=\mathbf p^{\alpha}(\mathbf{a}^1,\ldots, \mathbf
a^k)$
is a triple index that depends on $\mathbf{a}^1,\ldots, \mathbf a^k$.
The $j$-th color of $\mathbf p^{\alpha}$ will be denoted by
$p_j^{\alpha}$ and to fix the enumeration of $\mathbf p^{\alpha}$, we
will ask $p_1^\al:=a^\al_1$, for each $\al=1,\ldots, k$.  Moreover,
regularity and coloring of the graph implies that
$\{p_j^{\alpha}\}_{\alpha=1}^k$ and $\{a_j^1,a_j^2,\ldots,a^k_j\}$
coincide \textit{as sets, also} for the colors $j=2,3$.
\par

A crucial step in order to find the generalization of the expansion \eqref{madre}, is to
notice that\footnote{The author is indebted to Raimar Wulkenhaar for this valuable remark.}  
that very equation is a sum over boundaries of $\fey^\re_2(\phi^4)$.  
In order to adapt \eqref{madre} to 
$\feyn$, we take each monomial $G^{(2k)}_{\B} (\mathbf
a^1,\ldots, \mathbf a^k)J_{\mathbf a^1}\cdots J_{\mathbf a^k} \bar
J_{\mathbf p^1}\cdots \bar J_{\mathbf p^k},$ which in all generality
looks like in eq. \eqref{comb_graph} and notice that the structure of the sources is,
of course, encoded by the boundary graph $\B=\partial\G$.
Parenthetically, this is not an uncommon practice in (scalar) QFT,
where the boundary graph is just a graph without edges,
i.e. a finite set whose cardinality gives the
number of points of the correlation function.  The graph $\B$
and said monomial are uniquely, mutually determined as follows:
\begin{itemize}
\item  a source $J_{\mathbf a^{ s }}$ determines a
  white vertex in $\B$; a source $\bar J_{\mathbf{p}^{ j }}$, a black
  vertex in $\B$; 
\item two vertices are joined by a $c$-colored edge in $\B$ if and
  only if there exists a $(0c)$-bicolored path \textit{in} $\G$
  between the (vertices associated to the) external lines
  $J_{\mathbf a^{ s }}$ and $\bar J_{\mathbf{p}^{j}}$. Then set
\end{itemize}
\begin{equation} \label{fuentes_graficas}
\big(\mathbb J{(\B)}\big)(\mathbf{a}^1,\ldots,\mathbf a^k):=J_{\mathbf a^1}\cdots J_{\mathbf a^k}
\bar J_{\mathbf p^1}\cdots \bar J_{\mathbf p^k}= J_{a^{ 1 }_1a^{1 }_2a^{1}_3} \ldots J_{a^{k}_1a^{k}_2a^{k}_3} 
\bar J_{a^1_1p^1_2p^1_3} \ldots \bar J_{a^k_1p^k_2p^k_3}\,.
\end{equation}
Here the momenta $\mathbf p^\alpha$ are determined as in the graph
\eqref{comb_graph} and the convention below it, and $\mathbb J({\B})$
is a function of the momenta $\{\mathbf a\}=(\mathbf a^1,\ldots,\mathbf a^k)\in (\Z^3)^k$.
Thus, the expansion can be recast as
\begin{align}
 W[J,\bar J] & = \sum\limits_{k=1}^\infty \nonumber
 \sum\limits_{\substack{\B \in \partial(\feyn)  \\ 
k=\frac12|\B^{(0)}|}} \sum\limits_{\{\mathbf a\}}\frac{1}{|\Autc(\B)|} 
G^{(2k)}_{\B}(\{\mathbf a\})\,\cdot \,\mathbb J(\B) (\{\mathbf a\})\, .
\end{align}  
It will be convenient to define a pairing $\star$ of
functions $F:(\Z^3)^k\to \C$ with boundary graphs $\B\in \partial(\feyn) $:
\[
F\star \mathbb J(\B) := \sum\limits_{\mathbf a^1, \ldots, \mathbf a^k} 
F(\mathbf a^1, \ldots, \mathbf a^k) \mathbb J(\B(\mathbf a^1, \ldots, \mathbf a^k))\, .
\]
With this notation, $W$ takes the neater form:
\begin{align} \label{expansion_W}
 W[J,\bar J] & = \sum\limits_{k=1}^\infty
  \sum\limits_{\substack{\B \in \amalg\tcol  \\ 
2k=|\B^{(0)}|}}  \frac{1}{|\Autc(\B)|} 
G^{(2k)}_{\B} \star \mathbb J(\B) \, . 
\end{align}  
The fact that \textit{all $3$-colored graphs} appear listed
in eq. \eqref{expansion_W} is consequence of 
Theorem \ref{thm:completeness_disconn}.

\begin{remark}  
  A conspicuous difference with matrix models' expansion \eqref{madre}
  ---where the boundary of each graph is topologically ``uniform'',
  all being triangulations of $\sqcup^B \mathbb S^1$--- is that in
  rank-$3$ tensor field theories, the analogous connected components
  of the boundary have a non-trivial topology, since these are
  $3$-colored graphs and therefore \cite{cips} define closed
  orientable surfaces,
  $\Delta_\B \cong \bigsqcup^B_{\beta=1} \Sigma^{g_\beta}$ with
  $g_\beta \in \Z_{\geq 0}$ and $\Sigma^g=\#^g \T^2$ (being
  $\Sigma^0:=\mathbb S^2 $ for the empty connected sum $g=0$). As
  shown here, an analogous result holds for higher dimensions. Details
  on the expansion of $W$ in disconnected boundary graphs are
  presented in Appendix \ref{app:B}. \end{remark}

To illustrate the expansion, we derive the first terms 
in powers of the sources:
\begin{align}
\label{expansion_congraficas}
W_{D=3}[J,\bar J]&=G\hp 2 _{\includegraphics[height=1.9ex]{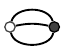}} \star \mathbb{J} (
\raisebox{-.3\height}{\includegraphics[width=.55cm]{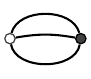}} ) + 
\frac1{2!}
G\hp 4 _{|\raisebox{-.4\height}{\includegraphics[height=1.9ex]{gfx/Item2_Melon.pdf}}|
\raisebox{-.4\height}{\includegraphics[height=1.9ex]{gfx/Item2_Melon.pdf}}|} \star \mathbb{J} (
\raisebox{-.3\height}{\includegraphics[width=.55cm]{gfx/icono_melonM.pdf}}
\sqcup 
\raisebox{-.3\height}{\includegraphics[width=.55cm]{gfx/icono_melonM.pdf}} ) + 
\frac1{2}G\hp4_{\raisebox{-.4\height}{\includegraphics[width=.5cm]{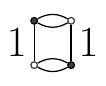}}} \star \mathbb{J}
\big( \,
\raisebox{-.34\height}{\includegraphics[width=0.8cm]{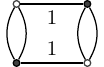}} \,\big) \\  
\nonumber
& \,\, \,\,
+\frac1{2}G\hp4_{\raisebox{-.4\height}{\includegraphics[width=.5cm]{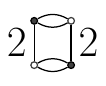}}}  
\star \mathbb{J}
\big(\, 
\raisebox{-.34\height}{\includegraphics[width=0.8cm]{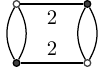}}\, \big)
+\frac1{2}G\hp4_{\raisebox{-.4\height}{\includegraphics[width=.5cm]{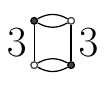}}} 
\star \mathbb{J}
\big( \,
\raisebox{-.34\height}{\includegraphics[width=0.8cm]{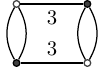}}\, \big)
+ \sum_{c=1}^3\frac1{3}G\hp 6_{\raisebox{-.4\height}{\includegraphics[width=.5cm]{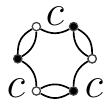}}} 
\star \mathbb{J}
\Big(
\raisebox{-.34\height}{\includegraphics[width=.8cm]{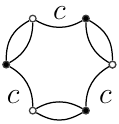}}\Big)
\\  \nonumber
&\,\,\,\,+
\frac1{3} G\hp6 _{\raisebox{-.4\height}{\includegraphics[width=1.75ex]{gfx/logo_k33_blanco.pdf}}} \star \mathbb{J}
\Big(
\raisebox{-.3\height}{\includegraphics[width=.66cm]{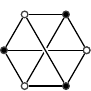}}\Big) 
+ \sum_{c=1}^3 G\hp 6_{
\raisebox{-.34\height}{\includegraphics[width=0.72cm]{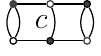}}} \star \mathbb{J}
\Big(
\raisebox{-.34\height}{\includegraphics[width=1cm]{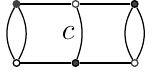}} \Big)
+ \frac1{ 3!}
G\hp 6 _{|\raisebox{-.4\height}{\includegraphics[height=1.9ex]{gfx/Item2_Melon.pdf}}|
\raisebox{-.4\height}{\includegraphics[height=1.9ex]{gfx/Item2_Melon.pdf}}|  
\raisebox{-.4\height}{\includegraphics[height=1.9ex]{gfx/Item2_Melon.pdf}}|}  \star 
 \mathbb{J} (
\raisebox{-.3\height}{\includegraphics[width=.55cm]{gfx/icono_melonM.pdf}}
^{\sqcup 3})
\\   
\nonumber
& \,\, \,\,
+\sum_{c=1}^3 
\frac12 G\hp 6 _{|\raisebox{-.4\height}{\includegraphics[height=1.9ex]{gfx/Item2_Melon.pdf}}|
\raisebox{-.4\height}{\includegraphics[width=.5cm]{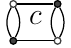}}|} 
\star  \mathbb{J}
\big(
\raisebox{-.3\height}{\includegraphics[width=.55cm]{gfx/icono_melonM.pdf}}
\sqcup 
\raisebox{-.34\height}{\includegraphics[width=.8cm]{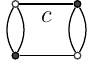}}\,
\big) + \mathcal O(8) \,.
\end{align}

In this expansion, the monomial $\mathbb J(\B)$ in the sources $J$ and
$\bar J$ is defined by formula \eqref{fuentes_graficas}.  Thus, for
instance the term in $W[J,\bar J]$ for the trace indexed by the
colored complete graph $K_{\mathrm c}(3,3)$ is
\begin{align*}
\frac1{3} G\hp6 _{\raisebox{-.4\height}{\includegraphics[width=1.75ex]{gfx/logo_k33_blanco.pdf}}} \star \mathbb{J}
\Big(
\raisebox{-.3\height}{\includegraphics[width=.66cm]{gfx/k33BandWg}}\Big) 
 = \frac13 \!
\sum\limits_{\mathbf{a,b,c}\in \Z^3}\!\!\!
G\hp6_{\raisebox{-.4\height}{\includegraphics[width=1.75ex]{gfx/logo_k33_blanco.pdf}}}
(\mathbf{a,b,c})
 J_{a_1a_2a_3}
 \bar J_{a_1b_2c_3}
 J_{b_1b_2b_3}
  \bar J_{a_3b_1c_2}
  J_{c_1c_2c_3}
 \bar J_{a_2b_3c_1} \,.
\end{align*}
This monomial $\mathbb{J}(\B)$ should not be confused with the trace
$\Tr_\B(J,\bar J)$, which would imply sums over all indices inside the
graph.  Actually $\Tr_{\B}(\phi,\bar\phi)= \boldsymbol{1}\star \mathbb
J (\B)$ holds, being $\boldsymbol 1$ the constant function $\Z^{|\B
  \hp 0|/2}\to \C$, $ \mathbf{a} \mapsto 1$.  Notice that formula
\eqref{expansion_congraficas} pairs the momenta
indices\footnote{Recall that we give the white indices and let the
  graph determine the black ones (see above eq.  \eqref{comb_graph}).}
of $\mathbb J(\B)$ with those of the corresponding Green's function
$G_\B^{(2k)}$. This seemingly redundant notation will
pay off not before the WTI below. The next short section explains why
those factors have been chosen, and how to recover each Green's
functions $G\hp{2k}_\B$ in the expansion of $W[J,\bJ]$.  \\

The number of correlation functions in rank $D=3,4$ theories are
counted.  In \cite{counting_invariants}, Ben Geloun and Ramgoolam found
the generating function $Z_{D,\mtr {conn.}}(x)=\sum_p a\hp D_{p,\mtr {conn.}} x^p$ of the number $a\hp D_{p,\mtr {conn.}} $ of \textit{connected}
graphs $\Grph{D}$ of a fixed number of vertices $2p$. It has the
following behavior\footnote{The OEIS series numbers \cite{OEIS} for 
$Z_{3,\mathrm{conn.}}$ and $Z_{4,\mathrm{conn.}}$ are A057005 and A057006, respectively.}:
\begin{equation}\label{eq:invariants3}
Z_{3,\mtr{conn.}}(x)=\sum_p a\hp 3_{p,\mtr {conn.}}x^p =x +3 x^2 + 7x^3 + 26 x^4 +
97x^5+ 624x^6 +4163x^7+ 34470 x^8 +  \ldots
\end{equation}
The first three terms of this series are 
evident in eq. \eqref{expansion_congraficas}. 
For $D=4$, they also computed 
\begin{align}\label{eq:invariants4}Z_{4,\mtr{conn.}}(x) &=
x + 7x^2+ 41x^3+ 604x^4+ 13753x^5   \\&
\qquad\qquad\quad\!\! + 504243x^6 + 24824785x^7+ 1598346352 x^8+\ldots  \nonumber
\end{align}
From those expressions one can readily compute the 
number $a_{p}(D)$ of disconnected $D$-colored graphs in $2p$
vertices. That integer is the number of correlation 
 $2p$-point functions.

\subsection{Graph calculus} \label{sec:graphcalculus}

Let $\R,\mathcal Q \in \tcol$ and $\mathbf{a}^1,\ldots,\mathbf{a}^ r,
\mathbf{c}^1,\ldots,\mathbf{c}^ q\in \Z^3$. In view of the discussion
above, we associate to those graphs, respectively, the monomials
$\mathbb{J}(\R)(\mathbf{a}^1,\ldots,\mathbf{a}^ r)$ and
$\mathbb{J}({\mtc{Q}})(\mathbf{c}^1,\ldots,\mathbf{c}^ q)$.  Here, the
white vertices of the graph $\R$ have incoming momenta
$\mathbf{a}^1,\ldots,\mathbf{a}^ r$ and similarly for $\mtc Q$.  So we
can derive one with respect to the other:
\[
\frac{\partial   \R(\mathbf{c}^1,\ldots,\mathbf{c}^ p)  }
{\partial\mathcal Q(\mathbf{a}^1,\ldots,\mathbf{a}^ q)} :=\left.
\frac
{\partial \mathbb{J} (\R)(\mathbf{a}^1,\ldots,\mathbf{a}^ p)}
{\partial \mathbb{J} ({\mtc{Q}})(\mathbf{c}^1,\ldots,\mathbf{c}^ q)}
\right|_{J=0=\bar J}. 
\]
This can be straightforwardly   
computed. First notice that trivially, if $p\neq q$,  
 automatically $\partial \R/\partial \mtc{Q} \equiv 0$.
Otherwise we have:
\begin{lemma} \label{lem:graph_independence}
Let $\mathbf{a}^1,\ldots,\mathbf{a}^ r\in \Z^3$ be colorwise, 
pairwise different, i.e. such that for each $\al,\beta=1,\ldots,r$,
and for each color $c=1,2,3$, $a_c^ \al \neq a_c^ \beta$ holds
whenever $\alpha\neq \beta$. Then for 
connected graphs $\R,\mathcal Q  \in \tcol$,
\begin{align} \label{graph_calc}
\frac{\partial   \R(\mathbf{c}^1,\ldots,\mathbf{c}^ r)  }
{\partial\mathcal Q(\mathbf{a}^1,\ldots,\mathbf{a}^ r)}=
\begin{cases} 
 \sum\limits_{\hat\sigma\in\Autc(\R)}
\delta^{ \mathbf{ c}^{\sigma (1)} ,\ldots,\mathbf{c}^ {\sigma(r)}}
      _{\mathbf{a}^1,\mathbf{a}^2, \ldots,\mathbf{a}^ r}
  & \qquad \mbox{if } \,\,\, \R \cong \mathcal Q\, , 
 \\
  \,\,\,\quad 0 & \qquad \mbox{if } \,\, \R \ncong \mathcal Q\,.  
 \end{cases}
\end{align} 
\end{lemma}	
Here $\hat\sigma \in \Autc(\R)$ means the  automorphism 
$\hat\sigma:\R\to \mtc Q$ whose restriction to white vertices
satisfies 
$\hat\sigma|_{\R\hp0_{\vphantom{b}\mathrm{w}}}=
\sigma\in\mathrm{Sym}({\R\hp0_{\vphantom{b}\mathrm{w}}})=\Sym_r$.
Also ``$\cong$'' denotes isomorphism in the sense of colored graphs, and
the $\delta$-function is shorthand for the following product of $3r$
Kronecker-deltas:
$
\delta^{a_1,a_2,a_3,\ldots,b_1,b_2,b_3}_{c_1,c_2,c_3,\ldots,d_1,d_2,d_3}=
\delta^{a_1}_{c_1}\delta^{a_2}_{c_2}\delta^{a_3}_{c_3}
\cdots \delta^{b_1}_{d_1}\delta^{b_2}_{d_2}
\delta^{b_3}_{d_3}
.$
\begin{proof}
If we compute directly using 
\begin{equation}
\label{JbJderivatives}
\frac{\partial J^\sharp_{\mathbf u}}{\partial J^\natural_{\mathbf w}}= 
\delta{^\natural_\sharp}\delta^{u_1}_{w_1}\delta^{u_2}_{w_2}\delta^{u_3}_{w_3} 
\where J^\sharp,J^\natural\in \{J,\bar J\},\end{equation}
then one splits this in the $J$-derivatives and the $\bar J$-terms
\begin{align}  \nonumber
\frac{\partial   \R(\mathbf{c}^1,\ldots,\mathbf{c}^ r)  }
{\partial\mathcal Q(\mathbf{a}^1,\ldots,\mathbf{a}^ r)}= \left.
\frac{\partial^r(J_{\mathbf c^ 1} \ldots J_{\mathbf c^ r})}
{\partial J_{\mathbf p^ 1} \ldots \partial J_{\mathbf p^ r}}\right|_{J=0} \left.
\frac{\partial^r(\bar J_{\mathbf q ^ 1} \ldots \bar J_{\mathbf q^ r})}
{\partial \bar J_{\mathbf p^ 1} \ldots \partial \bar J_{\mathbf p^ r}} \right|_{\bar J=0}
, 
\end{align} 
being the labels of the sources 
fully determined by
\begin{equation} \label{coloruno}
 p^\al_1=a_1^\al \and q^\al_1=c^\al_1  \quad \mbox{for all }  \al=1,\ldots,r. 
\end{equation} 
One can again use eq. \eqref{JbJderivatives} and compute each of these terms:
 \begin{align}  \nonumber
 \left.
\frac{\partial^r(J_{\mathbf c^ 1} \ldots J_{\mathbf c^ r})}
{\partial J_{\mathbf a^ 1} \ldots \partial J_{\mathbf a^ r}}\right|_{J=0}=
\sum\limits_{\sigma\in\Sym_{p}}   \delta^{\mathbf{c}^{\sigma(1)}}_{\mathbf{a}^1} 
\delta^{\mathbf{c}^{\sigma(2)}}_{\mathbf{a}^2} \cdots\,
\delta^{\mathbf{c}^{\sigma(r)}}_{\mathbf{a}^ r} \, , 
 \end{align}
 and
  \begin{align}  \nonumber
 \left.
\frac{\partial^r(\bar J_{\mathbf q^ 1} \ldots \bar J_{\mathbf q^ r})}
{\partial \bar J_{\mathbf p^ 1} \ldots \partial \bar J_{\mathbf p^ r}}\right|_{\bar J=0}=
\sum\limits_{\sigma\in\Sym_{p}}   \delta^{\mathbf{q}^{\sigma(1)}}_{\mathbf{p}^1} 
\delta^{\mathbf{q}^{\sigma(2)}}_{\mathbf{p}^2} \cdots\,
\delta^{\mathbf{q}^{\sigma(r)}}_{\mathbf{p}^ r} \, . 
\end{align}
Then 
\begin{align}  \nonumber
\frac{\partial   \R(\mathbf{c}^1,\ldots,\mathbf{c}^ r)  }
{\partial\mathcal Q(\mathbf{a}^1,\ldots,\mathbf{a}^ r)} & =  
\sum\limits_{\sigma\in\Sym_{p}}  \sum\limits_{\tau\in\Sym_{p}} 
\delta^{\mathbf{c}^{\sigma(1)}}_{\mathbf{a}^1} 
\delta^{\mathbf{q}^{\tau(1)}}_{\mathbf{p}^1} 
\delta^{\mathbf{c}^{\sigma(2)}}_{\mathbf{a}^2} 
\delta^{\mathbf{q}^{\tau(2)}}_{\mathbf{p}^2} \cdots\,
\delta^{\mathbf{c}^{\sigma(r)}}_{\mathbf{a}^ r}    
\delta^{\mathbf{q}^{\tau(r)}}_{\mathbf{p}^ r}
\\
& = \sum\limits_{\substack{\sigma,\tau \in \Sym_r \\ \sigma=\tau}}
\delta^{\mathbf{c}^{\sigma(1)}}_{\mathbf{a}^1} 
\delta^{\mathbf{q}^{\tau(1)}}_{\mathbf{p}^1} 
\delta^{\mathbf{c}^{\sigma(2)}}_{\mathbf{a}^2} 
\delta^{\mathbf{q}^{\tau(2)}}_{\mathbf{p}^2} \cdots\,
\delta^{\mathbf{c}^{\sigma(r)}}_{\mathbf{a}^ r}    
\delta^{\mathbf{q}^{\tau(r)}}_{\mathbf{p}^ r}
\end{align} 
where the restriction to sum only over the diagonal $\tau=\sigma$ is
derived from the color-$1$ deltas by using the index-definition
\eqref{coloruno}:
$\delta_{a_1^\al}^{c_1^{\si(\al)}}\delta_{p_1^\al}^{q_1^{\tau(\al)}}=
 \delta_{a_1^\al}^{c_1^{\si(\al)}}\delta_{a_1^\al}^{c_1^{\tau(\al)}}=\delta^{\tau(\al)}_{\si(\al)}$
 for arbitrary $\al=1,\ldots,r$.
The second equality follows from the condition  $c_j^ \al \neq c_j^ \gamma$ 
if $\al\neq\gamma$,  for each color $j=1,2,3$.  
Now suppose that $\R\ncong \mtc Q$ and consider,
for an arbitrary  $\sigma\in \Sym_r$, the  following term in the sum: 
\begin{equation}
  \delta^{\mathbf{c}^{\sigma(1)}}_{\mathbf{a}^1} \label{termino}
  \delta^{\mathbf{q}^{\sigma(1)}}_{\mathbf{p}^1} 
  \delta^{\mathbf{c}^{\sigma(2)}}_{\mathbf{a}^2} 
  \delta^{\mathbf{q}^{\sigma(2)}}_{\mathbf{p}^2} \cdots\,
  \delta^{\mathbf{c}^{\sigma(r)}}_{\mathbf{a}^ r}    
  \delta^{\mathbf{q}^{\sigma(r)}}_{\mathbf{p}^ r} \end{equation}
By assumption $ \mtc Q \neq \hat \sigma(\R)$. 
That is, there is a
white vertex (marked by) $\mathbf a^\alpha$, and a color $j\neq 1$,
with the following property: 
\begin{itemize}
\item[-] if $\mathbf p^\nu\in \mtc Q\hp0\bl$ denotes the black
  vertex where the $j$-colored edge $e_j$ beginning at
  $\mathbf a^\al$ ends (i.e. $t(e_j)=\mathbf p^\nu$); and, moreover,
 if $\mathbf q^\gamma\in \R_{\vphantom{b}\mathrm{b}}\hp0$ denotes the 
vertex where the $j$-colored edge at 
 $\mathbf c^{\si(\al)}$ ends; 
  \textit{then}  $\hat \sigma\inv(\mathbf q^\gamma)\neq \mathbf p^\nu$. 
\end{itemize}
This means that the following deltas 
are contained in the term \eqref{termino}:
\begin{equation} \label{term_dos}
\delta^{a_j^\alpha}_{p^\nu_j} \delta^{\mathbf a^\al}_{\mathbf c^{\sigma(\al)}}
\delta^{c^{\sigma(\al)}_j}_{q^\gamma_j}  \delta^{\hat\sigma\inv(\mathbf q^\gamma)}_{\mathbf q^\gamma} 
\end{equation}
On the other hand, consider the $j$-colored edge $g_j$ with
$t(g_j)=\hat\sigma\inv(\mathbf q^{\gamma})$ and the vertex $\mathbf
a^\mu$ with $s(g_j)=\mathbf a^\mu$.  Because of $\hat
\sigma\inv(\mathbf q^\gamma)\neq \mathbf p^\nu$ and as consequence of
the regularity of the coloring of the graph one has $\mu\neq
\alpha$. Thus, the term \eqref{termino} contains, on top of
\eqref{term_dos},  $\delta^{\mathbf a^\mu}_{\hat \sigma\inv(\mathbf
  q^\gamma) }$.  By using the assumption, $a_j^ \al \neq a_j^ \mu$ one
gets
$
\delta^{a_j^\alpha}_{p^\nu_j} \delta^{\mathbf a^\al}_{\mathbf c^{\sigma(\al)}}
\delta^{c^{\sigma(\al)}_j}_{q^\gamma_j}  \delta^{\hat\sigma\inv(\mathbf q^\gamma)}_{\mathbf q^\gamma}  \delta^{\mathbf a^\mu}_{\hat \sigma\inv(\mathbf q^\gamma) } =0.
$
Since this holds for arbitrary $\sigma$, then $\R\ncong \mtc Q$
implies that $\partial \R/\partial \mtc Q\equiv 0$.  Hence
 \begin{align}  \nonumber
\frac{\partial   \R(\mathbf{c}^1,\ldots,\mathbf{c}^ r)  }
{\partial\mathcal Q(\mathbf{a}^1,\ldots,\mathbf{a}^ r)}=
\sum\limits_{\sigma\in\Sym_{r}}\delta (\hat \sigma(\R), \mathcal{Q}) 
\delta^{\mathbf{ c}^{\sigma(1)},\ldots, \mathbf{c}^{\sigma(r)}}
      _{\mathbf{a}^1,\mathbf{a}^2, \ldots,\mathbf{a}^ r}\, , 
   \mbox{ with }
 \delta (\hat \sigma (\R), \mathcal{Q}) 
 :=\begin{cases} 
  1
  & \qquad \mbox{if } \,\,\, \hat\sigma(\R)  =\mathcal Q\, , 
 \\
   0 & \qquad \mbox{if } \,\,\, \hat\sigma(\R) \neq \mathcal Q\,.  
 \end{cases}
\end{align} 
 The sole non-vanishing terms are precisely the automorphisms of $\R$
 and the result follows.
\end{proof}

To better comprehend this formula, notice that the derivative
$\partial/\partial \mathcal Q$ still has momentum dependence, for
$\mathcal Q$ has external lines as vertices.  For instance,
\[
\frac{\partial}{\partial\,{\raisebox{-.33\height}{\includegraphics[width=.7cm]{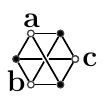}}}}
=\frac{\partial^6}{
\partial J_{a_1a_2a_3}
\partial \bar J_{a_1b_2c_3}
\partial J_{b_1b_2b_3}
\partial \bar J_{a_3b_1c_2}
\partial J_{c_1c_2c_3}
\partial \bar J_{a_2b_3c_1}
}\,.
\]
For $\{\mathbf{a,b,c}\}$ and $\{\mathbf{e,f,g}\}$ subsets of $\Z^3$ satisfying 
the hypothesis of Lemma \ref{lem:graph_independence},   
\[
\frac{\partial}{\partial\,{\raisebox{-.33\height}{\includegraphics[width=.7cm]{gfx/externo2_k33.pdf}}}}
\bigg(\raisebox{-.4\height}{\includegraphics[width=1.2cm]{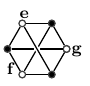}}  \!\!\!\bigg)=
\delta_{\mathbf{a}}^{\mathbf{e}} \delta_{\mathbf{b}}^{\mathbf{f}} \delta_{\mathbf{c}}^{\mathbf{g}}
+\delta_{\mathbf{a}}^{\mathbf{g}} \delta_{\mathbf{b}}^{\mathbf{e}} \delta_{\mathbf{c}}^{\mathbf{f}}
+\delta_{\mathbf{a}}^{\mathbf{f}} \delta_{\mathbf{b}}^{\mathbf{g}} \delta_{\mathbf{c}}^{\mathbf{e}}\,,
\]
holds for $\{a_d\neq b_d \neq c_d \neq a_d\}_{d=1,2,3}$
and has the same information as $\Autc(\kthreeg)\simeq \Z_3$.
 \par 
 Formula \eqref{graph_calc} can be directly generalized to
 non-connected graphs.  Any graph $\R\in \amalg\tcol$ that has $s$
 different connected components, each of multiplicity $m_i$,
 $i=1,\ldots,s$, can be split according to
\begin{equation} \label{tipo_genearal}
\R= \R_1^1 \sqcup \R^1_2\sqcup \ldots\sqcup \R^1_{m_1}\sqcup 
\cdots\sqcup \R^s_{1}\sqcup \R^s_{2}\sqcup \ldots \sqcup \R^s_{m_s}\,,
\end{equation}
where the subindices only label copies of the same graph $\R^k_{\,\!^{_*}}$.  
Using a similar expression for $\mtc{Q}\in \amalg\tcol$,
one finds, 
\begin{align}\label{graph_calc_disc}
\frac{\partial   \R(\mathbf{c}\hp1,\ldots,\mathbf{c}\hp r)  }
{\partial\mathcal Q(\mathbf{a}\hp1,\ldots,\mathbf{a}\hp r)}=
\left.\frac{\partial^{|\mathcal{Q}\hp0|}}{\partial  \mathbb{J}(\mathcal Q)} 
\mathbb{J}(\mathcal R)\right|\normalsize_{\substack {J=0 \\
\bar J=0}}=
 \begin{cases}
\sum\limits_{\sigma\in\Autc(\R)} 
\delta^{\sigma(\mathbf{ c}\hp1),\ldots,\sigma(\mathbf{c}\hp r)}
      _{\mathbf{a}\hp1,\mathbf{a}\hp2, \ldots,\mathbf{a}\hp r} & 
      \,\, \mbox{if } \,\, \R \neq \mathcal Q \,,
 \\
 \!\! \qquad 0 & \,\, \mbox{if } \,\, \R \neq \mathcal Q \,, 
 \end{cases}
\end{align}
which we again denote by $ \partial\R / \partial\mathcal{Q}$.
For $\R$ of the type \eqref{tipo_genearal}, 
this derivative contains certain number $\sigma(\Rb)$ Kronecker deltas, being 
$\sigma(\Rb) := m_1!\ldots m_s! 
|\Autc(\R^1_{^{\,\!_\bullet}})|^{m_1}\cdots 
|\Autc(\R^s_{^{\,\!_\bullet}})|^{m_s}$. 
This explains the factors accompanying the Green's functions in 
the expansion \eqref{ansatz2}.
\begin{lemma}\label{thm:recover}
  For $\B \in \amalg\tcol$, let $\mathcal N$ be the number of vertices of
  $\B$. Then the $\mtc N$-point function corresponding to $\B$ is
  non-trivial and can be recovered from the free energy $W$ as
  follows:
\[
G^{(\mathcal N)}_{\B}=\left.\frac{\partial }{\partial  \B}
W[J,\bar J]\right|_{J=0=\bar J}
:=\left.\frac{\partial^{|\B\hp0|}}{\partial {\mathbb J(\B)}} 
W[J,\bar J]\right|_{J=0=\bar J}\,.
\]
\end{lemma}
\begin{proof}
In the expansion \eqref{expansion_W}, we single out $\B$ and
derive with respect to $\B$:
\[
\left.\frac{\partial^{|\B\hp0|}}{\partial  {\mathbb J(\B) }} 
W[J,\bar J]\right|_{J=0=\bar J}=
\frac{\partial }{ \partial \B} \left( W[J,\bar J] - 
\frac1{\sigma(\B)}G^{(\mathcal N)}_{\B}\mathbb J(\B) \right) + 
\frac{1}{\sigma(\B)}\frac{\partial }{ \partial \B}
G^{(\mathcal N)}_{\B}\B \,.
\]
The first summand vanishes, since $\B$ does not appear in that sum of
terms. The second term yields, after equation \eqref{graph_calc_disc},
precisely $G^{(\mathcal N)}_{\B}$.  In the $\phi^4_3$-theory, this
Green's function is non-trivial, for there exists at least one graph,
whose boundary is $\B$, as stated by Lemma \ref{thm:completeness3}.
\end{proof}

For instance, the $6$-point function $
G\hp6 _{\raisebox{-.4\height}{\includegraphics[width=2ex]{gfx/logo_k33_blanco.pdf}}}$ 
reads in full notation
\[
G\hp6 _{\raisebox{-.4\height}{
\includegraphics[width=2ex]{gfx/logo_k33_blanco.pdf}}}
(\mathbf{a},\mathbf{b},\mathbf{c}) 
=\left.\frac{\partial^6 W[J,\bar J]}{
\partial J_{a_1a_2a_3}
\partial \bar J_{a_1b_2c_3}
\partial J_{b_1b_2b_3}
\partial \bar J_{a_3b_1c_2}
\partial J_{c_1c_2c_3}
\partial \bar J_{a_2b_3c_1}}
\right|_{J=0=\bar{J}} \,.
\]

\subsection{Arbitrary-rank graph calculus}
As is it obvious from the proofs, the results in previous section do
not rely on the number of colors.  In fact, we claim that for any
rank-$D$ model $V(\phi,\bar\phi)$, the following expansion in boundary-graphs
holds:
\begin{align} \label{expansion_Wgeneral}
 W[J,\bar J] & = \sum\limits_{k=1}^\infty
  \sum\limits_{\substack{\B \in \partial\fey_D(V(\phi,\bar\phi))  \\ 
k=\frac12\#(\B^{(0)})}}  \frac{1}{|\Autc(\B)|} 
G^{(2k)}_{\B} \star \mathbb J(\B) \, . 
\end{align} 
Here, we have set $N=1$, which can be reverted by a rescaling of the kinetic term and of 
each of the correlation functions $G\hp{2k}_\B \to N^{\gamma(\B)}G\hp{2k}_\B$,
and $\gamma(\B)$ should be determined. We postpone this task and depart from the
simplified version, eq. \eqref{expansion_Wgeneral}.
For the $\phi^4_\mathsf{m}$-theory, the sum is over all
$\partial\fey_D(\phi^4_\mathsf{m})=\amalg\Grph{D}$, as consequence of Theorem
\ref{thm:completeness_disconn}.  For that model the free energy $W_{D=4}[J,\bar J]$ 
to $\mathcal O(6)$ is then given by 
\begin{align*}
 W_{D=4}[J,\bar J] &=
G\hp 2 _{\includegraphics[width=.6cm]{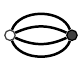}} \star \mathbb{J}\big(
\raisebox{-.34\height}{\includegraphics[width=0.6cm]{gfx/icono_melon4.pdf}}\big) + 
\frac1{2!}
G\hp 4 _{|\raisebox{-.4\height}{\includegraphics[width=.6cm]{gfx/icono_melon4.pdf}}|
\raisebox{-.4\height}{\includegraphics[width=.6cm]{gfx/icono_melon4.pdf}}|} \star \mathbb{J}\big(
\raisebox{-.34\height}{\includegraphics[width=0.6cm]{gfx/icono_melon4.pdf}}^{\sqcup \,2}\big) 
  \\  
  \nonumber
  & \,\,
+ \sum_{k=1}^4 \frac1{2}G\hp4_{\raisebox{-.4\height}{\includegraphics[width=.6cm]{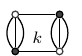}}} \star \mathbb{J}
\big(  
\raisebox{-.34\height}{\includegraphics[width=1.2cm]{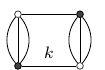}} \! \big)
+\sum_{k=2}^4 \frac1{2}G\hp4_{\raisebox{-.4\height}{\includegraphics[width=.9cm]{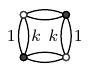}}} 
\star \mathbb{J}
\bigg(\!\!
\raisebox{-.4\height}{\includegraphics[width=1.3cm]{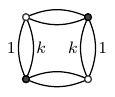}} \!\!\bigg) \,.
\end{align*}
We omit the next $49=41_{\mtr{conn.\, boundary }}+8_{\mtr{disconn. \,boundary}}$  \textit{connected} 
$\mathcal O(6)$-multipoint functions. 
The way to recover the correlation function $\G_\B$ from 
the free energy, in arbitrary rank is 
described by an obvious generalization of Lemma \ref{thm:recover}.

\section{The full Ward-Takahashi Identity}\label{sec:WI}
The Ward-Takahashi Identities for tensor models are inspired by 
those found for the Grosse-Wulkenhaar model given by the action
\begin{align} 
\int \dif{^4x}\;\left(   \frac{1}{2} (\partial_\mu \phi) \star
(\partial^\mu \phi) + \frac{\Omega^2}{2} 
(\tilde{x}_\mu \phi) \star (\tilde{x}^\mu \phi) + \frac{\mu^2}{2} \phi\star\phi 
 \label{GW}
+  \frac{\lambda}{4!} \phi \star \phi\star \phi \star \phi\right)(x) \,,
\end{align}
on Moyal $\re^4$. Here $x_\mu=2\Theta\inv_{\mu\nu}x^\nu$, being
$\Theta$ a $4\times4$ skew-symmetric matrix, that also parametrizes
the Moyal product
$ (f\star g)(x) = \int \frac{\dif{^4k}}{(2\pi)^4} \int \dif{^4} y
f(x{+}\tfrac{1}{2} \Theta {\cdot} k)\, g(x{+}y)\,
\mathrm{e}^{\mathrm{i} k \cdot y}$
on the Schwartz space, $f,g\in \mathcal S(\re^4,\C)$.  The
modification by a harmonic oscillator term makes the theory dual under
certain ``position-momentum''-duality, also known as
Langmann-Szabo-duality \cite{LS-duality}.  The authors of the model
have shown that their action \eqref{GW} with $\Omega=1$ can be grasped
as a matrix model in such a way that it fits in the setting of
\eqref{GW_as_matrix} by using the Moyal matrix base
\cite{gw:matrixbase04}.  We generalize the existent WTI in
\cite{DineWard} by following the non-perturbative strategy by
\cite[Sec. 2-3]{gw12}.

\subsection{Derivation of the Ward-Takahashi Identity}
We set $N=1$ from now on, which does not affect our analysis.
We consider an arbitrary colored tensor model in $D$ colors
$V(\phi,\bar\phi)=\sum_\B \lambda_\B \Tr_\B(\phi,\bar\phi)$, as
pointed out in Definition \ref{def:general_model}, with non-trivial
kinetic form
$S_0(\phi,\bar\phi)=\sum_{p\in\Z^D}\bar \phi_{p_1\ldots p_D}
E_{p_1\ldots p_D} \phi_{p_1\ldots p_D}$,
$E\neq \mtr{id}$, with $E$ self-adjoint. The measure
$\Df[\phi,\bar \phi]$ in the path integral of such a model is
invariant under the action of each factor of the group
$\mtr U(N)\times\ldots\times \mtr U(N)$.  We take an infinitesimal
transformation in its $a$-th factor
\[W_a\in \mtr U(N), \qquad W_a=1+\ii \alpha T_a + \mtc O(\al^2), 
\qquad T_a^\dagger=T\vphantom{^\dagger}_a,\] 
for any $a=1,\ldots, D$ and recast the invariance of 
the partition function 
with respect to this group action as the following matrix equation:
\begin{equation} \label{invariancia}
\fder{\log{Z[J,\bar J]}}{ T_a}=0.
\end{equation}
In the sequel, we drop the $N^{D-1}$ prefactors, which can be restored by rescaling $E$ and
the coupling constant(s). Denote by 
$F$ the source term 
$\Tr_2(\bar \phi,J) +\Tr_2(\bar J,\phi)$. 
One finds by using 
\begin{align*}
\fder{F(J,\bar J)}{(T_a)_{m_an_a}}
 = 
 \sum\limits_{p_i\in\,\Z} &\left[
\bar J_{p_1\ldots p_{a-1}m_a \ldots p_D}
\phi_{p_1\ldots p_{a-1}n_a \ldots p_D} \right. 
 -
\left.\bar\phi_{p_1\ldots p_{a-1}n_a\ldots p_D}
J_{p_1\ldots p_{a-1}m_ap_{a+1}\ldots p_D} 
\right]\, 
\end{align*}
and
\begin{align*}
\fder{S(\phi,\bar\phi)}{(T_a)_{m_an_a}}=
\fder{S_0(\phi,\bar\phi)}{(T_a)_{m_an_a}}&=
 \sum\limits_{p_i\in\,\Z}
 \left[
\bar \phi_{p_1\ldots p_{a-1}m_ap_{a+1}\ldots p_D}
E_{ {p_1\ldots p_{a-1}n_ap_{a+1}\ldots p_D}}
\phi_{p_1\ldots p_{a-1}n_ap_{a+1}\ldots p_D} \right. \\
&\qquad\, -\left.
\phi_{p_1\ldots p_{a-1}n_ap_{a+1}\ldots p_D}
E_{p_1\ldots p_{a-1}m_ap_{a+1}\ldots p_D}
\bar\phi_{p_1\ldots p_{a-1}m_ap_{a+1}\ldots p_D}\right]\,,
\end{align*}
that eq. \eqref{invariancia} implies 
 \begin{align*} & \!\!\!\!\!\!\!
 \int\Df[\phi,\bar \phi]\ee^{-S+F} 
   \sum\limits_{p_i\in\,\Z}
 \left[
\left(
E_{ {p_1\ldots  m _a \ldots p_D}}
-E_{ {p_1\ldots  n_a \ldots p_D}}
\right) \bar\phi_{p_1\ldots p_{a-1}m_ap_{a+1}\ldots p_D}
\phi_{p_1\ldots p_{a-1}n_ap_{a+1}\ldots p_D} \right]
\\
=& 
\int\Df[\phi,\bar \phi]
\ee^{-S+F}
\sum\limits_{p_i\in\,\Z}
\left(
\bar J_{ {p_1\ldots   m_a  \ldots p_D}}
\phi_{ {p_1\ldots   n _a \ldots p_D}}
-
\bar \phi_{ {p_1\ldots  m_a \ldots p_D}}
J_{ {p_1\ldots  n _a \ldots p_D}}
\right) \,.
 \end{align*}
Hence, after functional integration, one gets
  \begin{align*}
  &
 \sum\limits_{p_i\in\,\Z} \left(
E_{ {p_1\ldots  p_{a-1}m_a p_{a+1}\ldots p_D}}
-E_{ {p_1\ldots  p_{a-1}n_ap_{a+1} \ldots p_D}}
\right) 
\times \\
&
\quad \qquad \,\,\qquad\fder{ }{J}_{p_1\ldots p_{a-1}m_ap_{a+1}\ldots p_D}
\fder{}{\bar J}_{p_1\ldots p_{a-1}n_ap_{a+1}\ldots p_D} 
\exp\left(-S_{\mathrm{int}}(\delta/\delta \bar J,\delta/\delta J)\right)
\ee^{\sum_{\mathbf{q}} \bar J_{\mathbf q}
E _{\mathbf q}\inv J_{\mathbf q}}\\
=&
\sum\limits_{p_i\in\,\Z}
\left(
\bar J_{ {p_1\ldots  p_{a-1}m_a p_{a+1}\ldots p_D}}
\fder{ }{\bar J}_{ {p_1\ldots  p_{a-1}n_ap_{a+1} \ldots p_D}} 
\right.  \\  
 &\qquad \qquad   -
\left. 
J_{ {p_1\ldots  p_{a-1}n_ap_{a+1} \ldots p_D}}
\fder{ }{ J}_{ {p_1\ldots  p_{a-1}m_a p_{a+1}\ldots p_D}}
\right) 
 \exp{\left(-S_{\mathrm{int}}(\delta/\delta \bar J,\delta/\delta J)\right)
 }
\ee^{ \sum_{\mathbf{q}} \bar J_{\mathbf q}
E _{\mathbf q}\inv J_{\mathbf q}}\,.
 \end{align*} 
 We have identified $\delta/\delta J_{\mathbf{p}}$ 
 with $\bar \phi_{\mathbf p}$ and $\delta/\delta \bar J_{\mathbf{p}}$ 
 with $\phi_{\mathbf{p}}$. Thus, the preliminary WTI reads 
 \begin{align}  
& \!\!\!\!\! \sum\limits_{p_i\in\Z }\left(
E_{ {p_1\ldots  p_{a-1}m_a p_{a+1}\ldots p_D}}
-E_{ {p_1\ldots  p_{a-1}n_ap_{a+1} \ldots p_D}} 
 \right)
  \fder{^2 Z[J,\bar J]}{J_{{p_1\ldots  p_{a-1}m_ap_{a+1} \ldots p_D}}   \nonumber
\bar J_{{p_1\ldots  p_{a-1}n_ap_{a+1} \ldots p_D}}} 
\\
= &\sum\limits_{p_i\in\Z }\left\{  
\bar J_{p_1\ldots  p_{a-1}m_ap_{a+1} \ldots p_D} 
\fder{}{\bar J}_{p_1\ldots  p_{a-1}n_a  \ldots p_D}
-J_{p_1\ldots  p_{a-1}n_a p_{a+1}  \ldots p_D}
\fder{ }{J}_{p_1\ldots  p_{a-1}m_a \ldots p_D}\right\}Z[J,\bar J]\,. \label{ward1}
\end{align} 
As mentioned in the introduction, a Ward-Takahashi Identity
was obtained already in \cite{DineWard,us}. Namely, if we 
derive eq. \eqref{ward1} with respect to 
 \[\fder{ }{J}_{q_1\ldots  q_{a-1}n_aq_{a+1} \ldots q_D} 
 \fder{ }{\bar J}_{q_1\ldots  q_{a-1}m_aq_{a+1} \ldots q_D} \, ,   \]  
we obtain a relation between ``the $4$-point function''
 and the following difference of $2$-point functions (here
 also expressed in graphical notation, which we wish to surrogate 
 by $G_\B\hp{2k}$'s): 
\begin{equation}\label{eq:warddine}
\raisebox{-.455\height }{
\includegraphics[height=3.50cm]{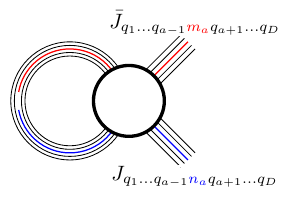} } 
\!\!\!\!\!\!\!\!\!\!\!\!\!\!\!\!\!
=
\raisebox{-.455\height }{
\includegraphics[height=3.5cm]{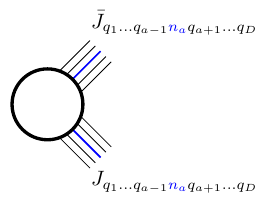} }
\!\!\!\!\!\!\!\!\!\!\!\!\!\!\!\!\!
- \quad
\raisebox{-.455\height }{
\includegraphics[height=3.5cm]{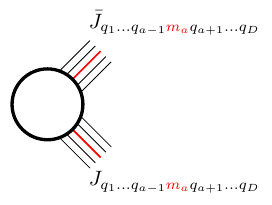} }
\end{equation}
The term in the LHS is defined as follows:
\[
\sum\limits_{p_i\in\,\Z} (E_{{p_1\ldots  p_{a-1}m_a p_{a+1}\ldots p_D}}  - 
E_{p_1\ldots  p_{a-1}n_a p_{a+1}\ldots p_D})\cdot\!\!\!\!\!\!\!
\raisebox{-.455\height }{
\includegraphics[height=3.5cm]{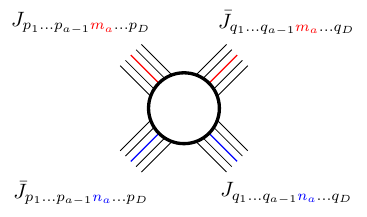} }
\] 
The issue is that for arbitrary degree $D$, there exist not only those
$4$-point functions. For instance, for $D=3$ there are four $4$-point
functions and for $D=4$ there are eight $4$-point functions. So we opt
for an analytical method that shows this missing structure. With that
aim, we need to solve eq. \eqref{ward1} for the double derivative of
$Z$. One can split $\delta^2Z/\delta J_{p_1\ldots m_a\ldots p_D}\delta
\bar J_{p_1\ldots n_a\ldots p_D}$ as a sum of a \textit{singular}
part, denoted by $Y\hp a_{m_a}$ and defined by being all the terms in
there proportional to $\delta_{m_an_a}$, and the \textit{``regular
  contribution''}. While the latter can be read off, no vestige from
$Y\hp a_{m_a}$ remains, for it is annihilated by $E_{{p_1\ldots m_a
    \ldots p_D}} - E_{{p_1\ldots n_a \ldots p_D}} $ in
eq. \eqref{ward1}. A direct approach with graphs does not consider
those contributions.  Of course, our method reduces to the result to
\eqref{eq:warddine} when $m_a\neq n_a$.  In order to find the singular
contributions, we need to introduce some terminology.

\begin{definition} \label{def:arista_weg} Let $\B\in \amalg \Grph{D}$ and
  $e\in \B\hp 1$.  The graph $\B\ominus e$ is defined as the graph
  that is formed after removal of \textit{all} the edges between the
  two vertices $e$ is attached at, and by subsequently colorwise gluing
  the remaining edges. More formally, we let
\begin{align*}
s\inv (s(e)) =: (s\inv (s(e))\cap t\inv(t(e))) \cupdot \,
A_{s(e)}, 
\quad 
t\inv(t(e))=: (s\inv (s(e))\cap t\inv(t(e)) ) \cupdot \, A_{t(e)}\,.
\end{align*}
We let $I(e)$ be the set of colors in of the edges $s\inv (s(e))\cap t\inv(t(e))$.
Then the coloring of $A_{s(e)}$ and of $A_{t(e)}$
agrees, both being equal to $\{1,\ldots,D\}\setminus I(e)$. 
We define $\B\ominus e$ by 
\begin{align*}
  (\B\ominus e)\hp0 &= \B\hp 0 \setminus \{s(e),t(e)\}\,, \\
  (\B  \ominus e ) \hp 1& = \B\hp 1 \setminus (s\inv (s(e))\cap t\inv(t(e)) ) / (A_{s(e)}  
  \sim_\mathrm{c} A_{t(e)})\,,
\end{align*}
where $f\sim_{ \mathrm{c}} g$ iff $f\in A_{s(e)}$ and $g\in A_{t(e)}$ have the
same color; see 
Figure \ref{escision}.  
By definition $\mathbb J(\emptyset)=1$, 
so $\mathbb J((\raisebox{-.3\height}
{\includegraphics[height=1.9ex]{gfx/Item2_Melon.pdf}})
\ominus e)=1$
for any edge $e$ of $\raisebox{-.3\height}
{\includegraphics[height=1.9ex]{gfx/Item2_Melon.pdf}}.$
\begin{figure}\centering 
\includegraphics[width=11cm]{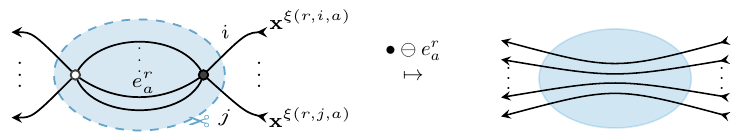}  
\caption{ On the definition of the graph $\B \ominus e_a^r$ in an arbitrary 
number of colors, being $a$ one of them. The graph 
on the left locally represents the $a$-colored edge $e_a^r$  and 
the vertices $s(e_a^r)$ and $t(e_a^r)$. The dipole that they form 
is removed and broken edges are colorwise glued (right graph)  
\label{escision}}
\end{figure}

\end{definition}

Keeping in mind the Ward-Takahashi Identity for a fixed color $a$ and
fixing the entries $(m_an_a)$ of the generator $T_a$, we shall define an
operator
$\Delta_{m_a,r}^\B : \C^{((\Z^D)^k)} \to \C^{((\Z^D)^{k-1})}$. 
In order to
do so, we need first to introduce more notation concerning the edge
removal $\B\ominus e^r_a$ as in Definition \ref{def:arista_weg}.  
Let $\B\in \partial \fey_D(V)$ with $|\B\hp 0|\geq 4$.
To stress the essence of the discussion 
we assume that the boundary graph $\B$ is connected and leave the extension 
of this discussion of the full disconnected boundary sector to Appendix \ref{app:B}.
We label the (say) white vertices of a boundary graph, $\B$, and of
  $\B\ominus e^j_c$ by momenta $\mathbf{x}^i\in\Z^D$ and denote by $e^j_c$
  the edge of color $c$ at the vertex $\mathbf{x}^j\in\Z^D$:
\begin{align} \label{nombrarvertices}
\B\hp 0\wh=(\mathbf x^1,\ldots, \mathbf x^k),\,\,\,\, (\B\ominus e^r_a)\hp0\wh =
(\mathbf x^1,\ldots,\widehat{\mathbf x^r},\ldots \mathbf x^k) =
(\mathbf y^1,\ldots, \mathbf y^{k-1}), \, (\mathbf y^l\in\Z^D)\,.
\end{align}

When $\B\ominus e^r_a$ is formed out of $\B$, one removes, in
particular, a single black vertex $t(e_a^r)$.  By regularity, certain
edge $e_i^{\xi(r,i,a)}$ of color $i$ is in $A_{t(e_a^r)}$, and this
edge determines again a unique white vertex $\mathbf x^{\xi(r,i,a)}$,
$1\leq \xi(r,i,a)\leq k$, by the relation
$s(e_i^{\xi(r,i,a)})=\mathbf x^{\xi(r,i,a)} \in\B\hp 0$.  If we pick
$ i\in \{1,\dots, D\}\setminus I(e_a^r)$, which is the color-set of
$A_{t(e_a^r)}$, one has $\xi(r,i,a)\neq r$ and, furthermore,
$\mathbf x^{\xi(r,i,a)}$ remains in $\B\ominus e_a^r$. This vertex is
renamed $\mathbf y^{\kappa(r,i,a)} \in (\B\ominus e_a^r)\hp{0}$
following \eqref{nombrarvertices}, whence
\begin{equation}
\kappa(r,i,a)= \begin{cases}
                \xi(r,i,a)  & 
                \mbox{ if } \xi(r,i,a) < r\,, \\
                \xi(r,i,a)-1   & 
                \mbox{ if } \xi(r,i,a) > r\,.
               \end{cases} \label{def:kappa}
\end{equation}

\begin{definition}\label{def:def_z}
  Keeping fixed the color $a$ and entries $(m_an_a)$ of
  a generator of the $a$-th summand $\mathfrak u(N)$ in the Lie
  algebra $\mathsf{Lie}(\mathrm{U}(N)\times \ldots \times \mathrm{U}(N))$, we consider
  the Green's function 
  $G_\B\hp {2k}:(\Z^D)^k\to \C$
  associated to a boundary-graph $\B\in \partial 
  (\fey_D(V))$.  For any integer $r$, $1\leq r\leq
  k$, we define the function $\Delta_{m_a,r}^\B G_\B\hp{2k}:
  (\Z^D)^{k-1}\to \C$ by
\[
(\Delta_{m_a,r}^\B G_\B\hp{2k})(\mathbf y^1,\ldots, \mathbf y^k )=
\begin{cases}
 \sum\limits_{\substack{q_h  }}
G_\B\hp{2k}(\mathbf y^1,\ldots,\mathbf y ^ {r-1}, {\mathbf z}^{r}(m_a,\mathbf q),\mathbf y^r,
\ldots,\mathbf y^ k ) 
& 
\mbox{ if } I(e_a^r)\neq \{a\} \,,
\\ \quad \,
G_\B\hp{2k}(\mathbf y^1,\ldots,\mathbf y ^ {r-1}, {\mathbf z}^{r}(m_a,\mathbf q),\mathbf y^r,
\ldots,\mathbf y^ k ) 
& \mbox{ if } I(e_a^r)=\{a\}\, ,
\end{cases}
\]
where $I(e_a^r)$ is the \textit{set of colors} of the edges
$s\inv(s(e_a^r))\cap t\inv (t(e_a^r))$ and $q_h$, for any
$h\in I(e_a^r)\setminus\{a\}$, is a dummy variable to be summed
over. The momentum ${\mathbf z}^ r\in\Z^D$ has entries defined by:
\[
z_i^ r=
\begin{cases}
m_a & \qquad \mbox{ if } \,\, i=a\,,\\
q_i & \qquad \mbox{ if }\,\, i\in I(e_a^r)\setminus \{a\}\,,\\
y_i^{\kappa(r,i,a)} & \qquad \mbox{ if }\,\, i \in \mbox{colors of } A_{t(e_a^r)} =
\{1,\ldots, D\}\setminus I(e_a^r) \,,
 \end{cases}
\]
where $\mathbf y^{\kappa(r,i,a)}$ ($1\leq \kappa(r,i,a)< k$) is the
white vertex $\B\ominus e^r_a$ defined in \eqref{def:kappa} (see also
Fig. \ref{fig:colapso_2}).
This definition depends on the labeling of the vertices. However, the
pairing $\langle\!\langle G\hp{2k}_{\B} , {\B} \rangle\!\rangle_{m_a}$
defined as follows does not:
\begin{equation}
 \langle\!\langle 
G\hp{2k}_{\B} , 
{\B} \rangle\!\rangle_{m_a} :=  
\sum\limits_{r=1}^k
\left(\Delta_{m_a,r}^\B G\hp{2k}_{\B}\right)\star \mathbb J
{(\B\ominus e_a^r)}\label{parejas} \,.
\end{equation}

\begin{figure}\centering\centering
\begin{subfigure}{0.58\textwidth}
\includegraphics[width=.7\linewidth]{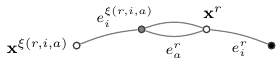}
\caption{ Locally, edge and vertex labeling in 
$\B$ prior to collapse into $\B\ominus e_a^r$} \label{fig:colapso_2a}
\end{subfigure}
\hspace*{\fill} 
\begin{subfigure}{0.34\textwidth}
\includegraphics[width=.99\linewidth]{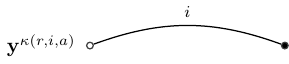} 
\caption{ Locally $\B\ominus e_a^r$} \label{fig:colapso_2b}
\end{subfigure}
\caption{ On the notation for the definition of
  $\Delta_{m_a,r}^\B$.  All gray edges in the leftmost figure
  disappear and merge into a single color-$i$ edge. The two vertices
  $t(e_a^r)$ and $\mathbf x^r$ disappear as
  well  \label{fig:colapso_2}}
\end{figure}
\end{definition}

\begin{example} \label{ex:Delta}
  We consider the \textit{empty graph} $\emptyset$ as colored.  According to Definition \ref{def:arista_weg} 
\begin{align*}
\big\langle  \hspace{-.60ex} \big\langle 
G\hp 2_{\raisebox{-.4\height}{\includegraphics[height=1.9ex]{gfx/Item2_Melon.pdf}}}, 
\raisebox{-.34\height}{\includegraphics[width=.55cm]{gfx/icono_melonM.pdf}}
\big\rangle \hspace{-.60ex}\big\rangle_{m_a}  = \big(
\Delta_{a,1} G\hp 2_{\raisebox{-.4\height}{\includegraphics[height=1.9ex]{gfx/Item2_Melon.pdf}}} \big)
\star \big(\, \raisebox{-.34\height}{\includegraphics[width=.55cm]{gfx/icono_melonM.pdf}}\ominus(e_1^a)\big)   = \!\!
\sum\limits_{\substack
{q_b,q_c \in\Z,  \\b\neq a\neq c \vspace{-.25cm}}} G\hp 2_{\raisebox{-.4\height}{\includegraphics[height=1.9ex]{gfx/Item2_Melon.pdf}}}
(\mbox{color-ordering of}\,(m_a,q_b,q_c)) \mathbb{J}(\emptyset)\,,
\end{align*}
so, for instance if $a=2$, one has:  
 \[
\big\langle\hspace{-.60ex} \big\langle 
G\hp 2_{\raisebox{-.4\height}{\includegraphics[height=1.9ex]{gfx/Item2_Melon.pdf}}}, 
\raisebox{-.34\height}{\includegraphics[width=.55cm]{gfx/icono_melonM.pdf}}
\big\rangle \hspace{-.60ex}\big\rangle _{m_2}=
\sum\limits_{q_1,q_3 \in\Z} 
G\hp 2_{\raisebox{-.4\height}{\includegraphics[height=1.9ex]{gfx/Item2_Melon.pdf}}}
(q_1,m_2,q_3)
\,. \vspace{-.2cm}\]
To clear up the notation in $\Delta^{\K}_{m_a,r}$,  
consider the graph $\K$ in Figure \ref{fig:Delta}
and examples concerning the edge removal, here for, say,   
$e^1_1$ and $e^3_2$. One has then
$\xi(1,3,1)=2$, $\xi(3,1,2)=5$ and $\xi(3,3,2)=1$. Therefore
$\kappa(1,3,1)=1, \kappa(3,1,2) =4 $ 
and $\kappa (3,3,2)=1 $. Accordingly:
\begin{align*}
(\Delta^{\mathcal K}_{\substack{m_1, r=1}}G_{\mathcal K}\hp{10} )(\mathbf y^1,\ldots,\mathbf y^4) &= \sum\limits_{q_2}
G_{\mathcal K}\hp{10}(\mathbf z^1,\mathbf y^1,\ldots,\mathbf y^4) = 
 \sum\limits_{q_2}
G_{\mathcal K}\hp{10}(m_1,q_2,y_3^1 ,\mathbf y^1,\ldots,\mathbf y^4)\, ,
\\
(\Delta^{\mathcal K}_{\substack{m_2,r=3}}G_{\mathcal K}\hp{10} )(\mathbf y^1,\ldots,\mathbf y^4)&= 
\,\,
G_{\mathcal K}\hp{10}(\mathbf y^1,\mathbf y^2,\mathbf z^3,\mathbf y^3,\mathbf y^4) =      
 \,\,G_{\mathcal K}\hp{10}(\mathbf y^1,\mathbf y^2,y_1^4,m_2,y_3^1 ,\mathbf y^3,\mathbf y^4)\, .
\end{align*}
The usefulness of this operation shall be clear in the 
proof of the next result.
\end{example}

\begin{figure}[H]\centering 
\raisebox{-.45\height}{\includegraphics[width=.80\textwidth]{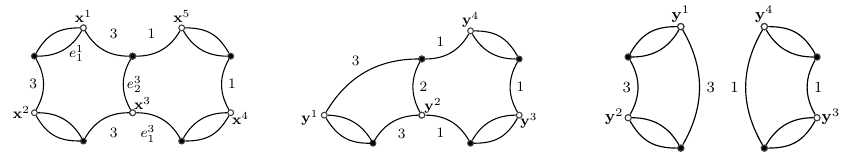}} 
\caption{Concerning example \ref{ex:Delta}, 
from left to right: $\K$, $\mathcal K\ominus e^1_1$
and $\mathcal K \ominus e^3_2$ 
\label{fig:Delta}}
\end{figure}

\begin{theorem}[Full Ward-Takahashi Identity] \label{thm:full_Ward}
  Consider an arbitrary rank-$D$ tensor model whose kinetic form $E$
  in $\Tr_2(\bar\phi,E\phi)$ (see Def. \ref{def:general_model}) obeys
\[
E_{p_1\ldots  p_{a-1}m_ap_{a+1}\ldots p_D}-
E_{p_1\ldots  p_{a-1}n_ap_{a+1}\ldots p_D}
=E(m_a,n_a) 
\quad \mbox{ for each } a=1,\dots ,D
\]
i.e. this difference does not depend on the momenta
$p_1,\ldots, \widehat p_a,\ldots ,p_D$.  Then the partition function
$Z[J,\bJ]$ of that model satisfies
\begin{align} \label{ward2}
 &\!\!\!\!\!\sum\limits_{p_i\in\Z } \fder{^2 Z[J,\bar J]}{J_{{p_1\ldots  p_{a-1}m_ap_{a+1} \ldots p_D}} 
\delta\bar J_{{p_1\ldots  p_{a-1}n_ap_{a+1}\ldots p_D}}} 
 -  \left(\delta_{m_an_a} Y\hp a_{m_a}[J,\bar J]\right) \cdot Z[J,\bar J]
\\
=& \sum\limits_{p_i\in\Z }
\frac1{ 
E_{ {p_1\ldots  m_a \ldots p_D}}
-E_{ {p_1\ldots   n_a  \ldots p_D}}  } 
 \left( 
\bar J_{p_1\ldots  m_a \ldots p_D} 
\fder{}{\bar J}_{p_1\ldots   n_a \ldots p_D}  
-J_{p_1\ldots   n_a  \ldots p_D}\nonumber 
\fder{ }{J}_{p_1\ldots   m_a  \ldots p_D}\right)Z[J,\bar J]
\end{align}   
where 
\begin{align} \label{termY}
Y\hp a_{m_a}[J,\bar J] &:= \sum\limits_{k=1}^\infty\,\,\sum\limits_{\substack{\B \in \partial\fey_D(V )  \\ 
2k= \#(\B^{(0)})}}    
\frac{1}{|\Autc(\B)|} \langle\!\langle 
G\hp{2k}_{\B} , 
{\B} \rangle\!\rangle_{m_a} \\
& =
\sum\limits_{k=1}^\infty\,\,\sum\limits  _{\substack{\B \in \partial\fey_D(V )  \\ 
2k= \#(\B^{(0)})}}  \frac{1}{|\Autc(\B)|} 
\sum\limits_{r=1}^k
\left(\Delta_{m_a,r}^\B G\hp{2k}_{\B}\right)\star \mathbb{J}
{(\B\ominus e_a^r)}\,. \nonumber
\end{align} 
\end{theorem}

\begin{proof}

In the next equation 
\begin{align} \label{dobleder}
\hspace{-.06cm}\!\!\!\!
 \frac1{Z[J,\bJ]} \fder{^2 Z[J,\bar J]}{J_{p_1\ldots m_a\ldots p_D }    
 \delta \bJ _{p_1\ldots n_a\ldots p_D }} = 
 \fder{^2 W[J,\bJ] }{J_{p_1\ldots m_a\ldots p_D } 
 \delta \bJ_{p_1\ldots n_a\ldots p_D }} +
 \fder{ W[J,\bJ]}{J_{p_1\ldots m_a\ldots p_D }}
 \fder{ W[J,\bJ]}{\bJ_{p_1\ldots n_a\ldots p_D }}  
\end{align}
we want to detect the terms that provide a singular contribution
$\delta_{m_an_a}$ to the double derivative in the LHS.  We group them
in $ \mathcal Y\hp a_{p_1,\ldots, {m}_a,\ldots ,p_D}[J,\bar J]$.
Notice that in the RHS, the product of derivatives cannot yield terms
proportional to $\delta_{m_an_a}$.  Hence, all terms in
$Y\hp a_{m_a} [J,\bJ]$ come from the sum over the momenta
$(p_1,\ldots, \widehat {p}_a,\ldots ,p_D)\in\Z^{D-1}$ of the term
$ \mathcal Y\hp a_{p_1,\ldots, {m}_a,\ldots ,p_D}$ in the double
derivative of $W[J,\bar J]$. By definition of
$ \mathcal Y\hp a_{p_1,\ldots, {m}_a,\ldots ,p_D}$, we can then write
\[\fder{^2 W[J,\bJ] }{J_{p_1\ldots m_a\ldots p_D } 
 \delta \bJ_{p_1\ldots n_a\ldots p_D }} = \delta_{m_a n_a}\mathcal
Y\hp a_{p_1,\ldots, {m}_a,\ldots ,p_D}[J,\bJ]\cdot Z[J,\bJ] + \mathcal
X\hp a_{p_1,\ldots, [{m}_an_a],\ldots ,p_D}[J,\bJ] \,, \]
where $\mathcal X\hp a _{p_1,\ldots, [{m}_an_a],\ldots ,p_D}[J,\bJ]$
contains only regular terms. We could compute them, but we are only
interested in the term proportional to $\delta_{m_an_a}$.  We let
the two derivatives act on the expansion of the free energy in
boundary graphs, eq. \eqref{expansion_Wgeneral}.
Ignoring the symmetry factor, the derivatives acting on a single term
$G\hp{2k}_{\B} \star \mathbb J(\B)$ lead to
\begin{align} 
&\suml_{p_i\in \Z} \suml_{\mathbf{x}^1,\ldots,\mathbf{x}^k}
G\hp{2k}_{\B}(\mathbf{x}^1,\ldots,\mathbf{x}^k)
\fder{^2 \mathbb J(\B\{\mathbf{x}^1,\ldots,\mathbf{x}^k\})} 
{J_{p_1\ldots m_a\ldots p_D} \delta \bJ_{p_1\ldots n_a\ldots p_D}} \nonumber\\
=&\suml_{p_i\in \Z} \suml_{\mathbf{x}^1,\ldots,\mathbf{x}^k}G\hp{2k}_{\B}(\mathbf{x}^1,\ldots,\mathbf{x}^k)
\fder{}{J_{p_1\ldots m_a\ldots p_D}}\big( J_{\mathbf x^1}\cdots  J_{\mathbf x^k} \big) 
\fder{}{\bJ_{p_1\ldots n_a\ldots p_D}}\big( \bJ_{\mathbf w^1}\cdots  \bJ_{\mathbf w^k} \big)\,. \label{eqn:contr_prim}
\end{align}
Here, we have the freedom to choose the order of the labels of the
black vertices by the momenta $\mathbf w^i$ in such a way that
$w^r_a=x_a^r$. Hence $t(e^r_a)$ has the label $\mathbf w^r$, whereas
$s(e^r_a)$ is labeled\footnote{The condition $ s(e_c^r)=\mathbf x^r$
  is actually redundant, since by definition $e^r_c$ is the edge of
  color $c$ attached at $\mathbf x^r$, but we write this down for sake
  of clarity.} by $\mathbf x^r$.  We now focus on the product of the
two derivative terms:
\begin{align}  
\nonumber 
& \Big(\suml_{r=1}^k\big\{ \delta ^{x^r_1}_{p_1} 
\cdots 
\delta ^{x^r_{a-1}}_{p_{a-1} } 
\delta ^{x^r_a}_{m_a }
\delta ^{x^r_{a+1}}_{p_{a+1} } 
\cdots 
\delta ^{x^r_D}_{p_D} \big\}
 J_{\mathbf x^1}\cdots   \widehat{J_{\mathbf x^r}}\cdots J_{\mathbf x^k} \Big) 
 \\
& 
\Big(
 \suml_{l=1}^k\big\{ \delta ^{w^l_1}_{p_1} 
 \cdots \delta ^{w^l_{a-1 }}_{p_{a-1}}\delta ^{w^l_a}_{n_a }
 \delta ^{w^l_{a+1 }}_{p_{a+1}}\cdots
\delta ^{w^l_D}_{p_D} \big\}
\bJ_{\mathbf w^1}\cdots   \widehat{\bJ_{\mathbf w^l}}\cdots \bJ_{\mathbf w^k}\Big) \,.
\label{contribucion_Y}
\end{align}
For the term $ \big( \delta ^{x^r_1}_{p_1} 
\cdots \delta ^{x^r_a}_{m_a }\cdots \delta ^{x^r_D}_{p_D} J_{\mathbf
  x^1}\cdots \widehat{J_{\mathbf x^r}}\cdots J_{\mathbf x^k}
\big)\big( \delta ^{w^l_1}_{p_1} 
\cdots \delta ^{w^l_a}_{n_a }\cdots \delta ^{w^l_D}_{p_D} \ldots
\bJ_{\mathbf w^1}\cdots \widehat{\bJ_{\mathbf w^l}}\cdots \bJ_{\mathbf
  w^k} \big)$
to contribute to $\mathcal Y\hp a_{\dots}$ one requires $x^r_a=w_a^l$.
But by definition of $\mathbf w^l$, $w^l_a=x^l_a$.  Then from
\eqref{contribucion_Y}, the terms that contribute to
$\mathcal Y\hp a_{\dots}$ are
 \[\suml_{r=1}^k
 \Big(\big\{ \delta ^{x^r_1}_{p_1}  
\cdots \delta ^{x^r_a}_{m_a }\cdots 
\delta ^{x^r_D}_{p_D} \big\}
 J_{\mathbf x^1}\cdots   \widehat{J_{\mathbf x^r}}\cdots J_{\mathbf x^k} \Big)
 \Big(\big\{ \delta ^{w^r_1}_{p_1}  
 \cdots \delta ^{w^r_a}_{n_a }\cdots 
\delta ^{w^r_D}_{p_D}\big\}
\bJ_{\mathbf w^1}\cdots   \widehat{\bJ_{\mathbf w^r}}\cdots \bJ_{\mathbf w^k}\Big)
 \]
We rewrite the $\delta$ in the  $r$-th term of this sum as
\[
\delta_{m_a}^{n_a}\cdot\prod
_{j\in A_{t(e_a^r)}}\delta_{p_j}^{x_j^r} \delta_{p_j}^{w_j^r} 
\prod
_{i\in I}\delta_{p_i}^{x_i^r} \delta_{p_i}^{w_i^r} 
=
\delta_{m_a}^{n_a}\cdot\prod
_{j\in A_{t(e_a^r)}}\delta_{p_j}^{x_j^r} \delta_{p_j}^{x_j^{\xi(r,j,a)}} 
\prod
_{i\in I}\delta_{p_i}^{x_i^r} \delta_{p_i}^{w_i^r} 
\]
(see Def. \ref{def:def_z}). Then
\begin{align*}
Y\hp a_{m_a}[J,\bar J]&= \sum_{\substack{p_1,\ldots,\widehat{p_a},\ldots,p_D \in \Z}}\mathcal{Y}\hp a_{p_1,\ldots,{m_a},\ldots,p_D} \nonumber
\\
&=\suml_{k=1}^\infty \,\, \suml_{\B\in\partial(\fey (V))}\frac{1}{|\Autc(\B)|}\suml_{\substack{p_1,\ldots,\widehat{p_a},\ldots,p_D}}
\,\,
\suml_{r=1\vphantom{\widehat {x^r}}}^k
\,\,
\suml_{\mathbf x^1,\ldots,\widehat
{\mathbf x^r},\ldots,\mathbf x^k} 
 \\
&\quad\phantom\times 
\prod
_{j\in A_{t(e_a^r)}}\delta_{p_j}^{x_j^r} \delta_{p_j}^{x_j^{\xi(r,j,a)}} 
\prod
_{i\in I}\delta_{p_i}^{x_i^r} \delta_{p_i}^{w_i^r} 
G\hp{2k}_{\B}(\mathbf{x}^1,\ldots,\widehat{\mathbf{x}^r},\ldots, \mathbf{x}^D) \\
&\quad \phantom{\times} \times\big(J_{\mathbf x^1}\cdots   \widehat{J_{\mathbf x^r}}\cdots J_{\mathbf x^k}\big)
\cdot \big(\bJ_{\mathbf w^1}\cdots   \widehat{\bJ_{\mathbf w^r}}\cdots \bJ_{\mathbf w^k}\big)
\end{align*}
and by renaming the indices and using Definitions \ref{def:arista_weg}
and \ref{def:def_z} one finally gets
\allowdisplaybreaks[3]
\begin{align*}
Y\hp a_{m_a}[J,\bar J]&=
\suml_{k=1}^\infty \,\, \suml_{\B\in\partial(\fey (V))}\frac{1}{|\Autc(\B)|}
\,\,
\suml_{r=1\vphantom{\widehat {x^r}}}^k
\,\,
\suml_{\mathbf y^1,\ldots,\mathbf y^{k-1}} \!\!\!
G\hp{2k}_{\B}(\mathbf y^1,\ldots,\mathbf z^r,\ldots, \mathbf y^{k-1})\\
& \qquad\qquad \qquad \qquad\qquad \qquad \qquad\qquad \qquad
\cdot \mathbb J\big(\B\ominus {e_r^a}\big)(\mathbf y^1,\ldots,\mathbf y^{k-1})\\
&=\suml_{k=1}^\infty \,\, \suml_{\B\in\partial(\fey (V))}\frac{1}{|\Autc(\B)|}
\suml_{r=1\vphantom{\widehat {x^r}}}^k
(\Delta_{r,m_a}G\hp{2k}_{\B})(\mathbf y^1,\ldots, \mathbf y^{k-1})
\\
& \qquad\qquad \qquad \qquad\qquad \qquad \qquad\qquad \qquad
\cdot
\mathbb J\big((\B\ominus {e_r^a})\big)\{\mathbf y^1,\ldots,\mathbf y^{k-1}\}\\
&=\suml_{k=1}^\infty \,\, \suml_{\B\in\partial(\fey (V))}\frac{1}{|\Autc(\B)|}
\suml_{r=1\vphantom{\widehat {x^r}}}^k
(\Delta_{r,m_a}G\hp{2k}_{\B}) \star (\B\ominus e^r_a) \\
&=\suml_{k=1}^\infty \,\, \suml_{\B\in\partial(\fey (V))}\frac{1}{|\Autc(\B)|}
\langle\!\langle G_\B\hp{2k},\B\rangle\!\rangle_{m_a} \,.  \qedhere
\end{align*} 
\end{proof}

\allowdisplaybreaks[0]

An explicit expansion of $Y\hp a_{m_a}[J,\bar J]$ is, of course, also 
more useful. We derive it for $D=3$ and
assume throughout that $\{a,c,d\}=\{1,2,3\}$ as sets; 
also for sake of notation, we drop $\B$ 
in $\Delta_{m_a,r}^\B$ when this operator acts on the correlation function 
that already shows dependence on $\B$. The expansion\footnote{In \cite{SDE} we found the 
$\mathcal O (6)$ terms and treated also rank $D=4,5$ theories.} reads 
\begin{align}
 Y\hp a_{m_a} [J,\bar J] 
&= \nonumber
\Delta_{m_a,1}G\hp 2_{\raisebox{-.4\height}{\includegraphics[height=1.9ex]{gfx/Item2_Melon.pdf}}
} \star \boldsymbol{1} 
\\
&\,\,
+ \frac12
\sum\limits_{r=1}^2\big( \Delta_{m_a,r}
G\hp 4 _{|\raisebox{-.4\height}{\includegraphics[height=1.9ex]{gfx/Item2_Melon.pdf}}|
\raisebox{-.4\height}{\includegraphics[height=1.9ex]{gfx/Item2_Melon.pdf}}|}+
\Delta_{m_a,r}
G\hp4_{\raisebox{-.4\height}{\includegraphics[width=.6cm]{gfx/logo_vone.pdf}}}
+ 
\Delta_{m_a,r}
G\hp4_{\raisebox{-.4\height}{\includegraphics[width=.6cm]{gfx/logo_vtwo.pdf}}}
+ \nonumber
\Delta_{m_a,r}
G\hp4_{\raisebox{-.4\height}{\includegraphics[width=.6cm]{gfx/logo_vthree.pdf}}}\big)
\star \mathbb{J}  \big(
\raisebox{-.34\height}{\includegraphics[width=.55cm]{gfx/icono_melonM.pdf}}\big)  \nonumber  
\\  
&\,\, + \frac1{3}\sum\limits_{r=1}^3\sum\limits_{i=1}^3
\big(\Delta_{m_a,r}
G\hp 6_{\raisebox{.14\height}{\includegraphics[width=.4cm]{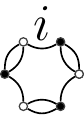}}}\big)
\star
\mathbb{J}  \big( 
\raisebox{-.34\height}{\includegraphics[width=.8cm]{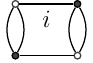}}\, \big) 
 + \frac13 \sum_{r=1}^3 
 (\Delta_{m_a,r}G\hp6_{\raisebox{-.4\height}{\includegraphics[width=.4cm]{gfx/logo_k33_blanco.pdf}}})
\star\mathbb{J}\big(
\raisebox{-.34\height}{\includegraphics[width=.8cm]{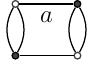}}\big) 
\nonumber
\\ 
\label{expansion_congraficasY}
& \,\,
+\sum\limits_{c\neq a}  
\big(\Delta_{m_a,1} G\hp 6_{
\raisebox{-.4\height}{\includegraphics[width=0.7cm]{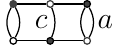}}}
\big)
\star\mathbb{J}\big(
\raisebox{-.34\height}{\includegraphics[width=.8cm]{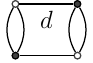}}\big) 
+\big(\Delta_{m_a,2}  G\hp 6_{
\raisebox{-.4\height}{\includegraphics[width=0.7cm]{gfx/logo_resto_6acostadoNew.pdf}}}
\big) 
\star \mathbb{J}
\big(
\raisebox{-.34\height}{\includegraphics[width=.8cm]{gfx/vertice_d.pdf}}\big) 
\\  
\nonumber
&\,\, 
+\big(
\Delta_{m_a,3}  G\hp 6_{
\raisebox{-.4\height}{\includegraphics[width=0.7cm]{gfx/logo_resto_6acostadoNew.pdf}}}
\big)
\star\mathbb{J}
\big(
\raisebox{-.34\height}{\includegraphics[width=.8cm]{gfx/vertice_a.pdf}} \big)
+ \big(
\Delta_{m_a,1}
G\hp 6_{
\raisebox{-.4\height}{\includegraphics[width=0.7cm]{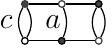}}}
\big)
\star\mathbb{J}
\big( 
\raisebox{-.34\height}{\includegraphics[width=.8cm]{gfx/vertice_c.pdf}}
\big)
+\big(
\Delta_{m_a,2}
G\hp 6_{
\raisebox{-.4\height}{\includegraphics[width=0.7cm]{gfx/logo_resto_6acostado_aHor.pdf}}}
\big)
\star\mathbb{J}
\big(\raisebox{-.34\height}{\includegraphics[width=.55cm]{gfx/icono_melonM.pdf}} \sqcup 
\raisebox{-.34\height}{\includegraphics[width=.55cm]{gfx/icono_melonM.pdf}}
\big) 
 \\
& \,\,
+\big(
\Delta_{m_a,3}
G\hp 6_{
\raisebox{-.4\height}{\includegraphics[width=0.7cm]{gfx/logo_resto_6acostado_aHor.pdf}}}
\big)
\star\mathbb{J}
\big( \nonumber 
\raisebox{-.34\height}{\includegraphics[width=.8cm]{gfx/vertice_d.pdf}}
\big)
+\frac1{ 3!}
\sum\limits_{r=1}^3\big(\Delta_{m_a,r}
G\hp 6 _{|\raisebox{-.4\height}{\includegraphics[height=1.9ex]{gfx/Item2_Melon.pdf}}|
\raisebox{-.4\height}{\includegraphics[height=1.9ex]{gfx/Item2_Melon.pdf}}|
\raisebox{-.4\height}{\includegraphics[height=1.9ex]{gfx/Item2_Melon.pdf}}|}\big)
\star\mathbb{J}
\big(
\raisebox{-.34\height}{\includegraphics[width=.55cm]{gfx/icono_melonM.pdf}} \sqcup 
\raisebox{-.34\height}{\includegraphics[width=.55cm]{gfx/icono_melonM.pdf}}\big) 
\\  
\nonumber
&\,\, 
+\nonumber 
\frac12 \big(\Delta_{m_a,1}G\hp 6 _{|\raisebox{-.4\height}{\includegraphics[height=1.9ex]{gfx/Item2_Melon.pdf}}|
\raisebox{-.4\height}{\includegraphics[width=.5cm]{gfx/logo_vertice_c.pdf}}|}\big)
\star \mathbb{J}
\big(
\raisebox{-.34\height}{\includegraphics[width=.8cm]{gfx/vertice_c.pdf}}
\big) 
+ \frac12\sum\limits_{r=2,3}\big(\Delta_{m_a,r}
G\hp 6 _{|\raisebox{-.4\height}{\includegraphics[height=1.9ex]{gfx/Item2_Melon.pdf}}|
\raisebox{-.4\height}{\includegraphics[width=.5cm]{gfx/logo_vertice_c.pdf}}|}\big)
\star \mathbb{J}
\big(
\raisebox{-.34\height}{\includegraphics[width=.55cm]{gfx/icono_melonM.pdf}} \sqcup 
\raisebox{-.34\height}{\includegraphics[width=.55cm]{gfx/icono_melonM.pdf}} \big)
+ \mathcal O(6)\,.
\end{align}
In this expression, for any two white vertices of a boundary graph
$\B$ that are not connected by an element $\tau\in\Sym_k$ that can be
lifted to $\hat\tau\in\Autc(\B)$, a convention regarding their
ordering should be set.  The ordering of the arguments
$(\mathbf x^1,\ldots, \mathbf x^k)$ of
$G_\B\hp {2k} (\mathbf x^1,\ldots, \mathbf x^k)$ and the labeling of
$\B\hp 0\wh$ by these obeys in eq. \eqref{expansion_congraficasY}
following convention:
\begin{quote}
  if a vertex $v_i\in\B\hp0\wh$ is labeled by the momentum
  $\mathbf x^i$ and appears to the left of the vertex
  $v_j\in\B\hp 0\wh$ labeled by $\mathbf x^j$, then $i<j\,$.
\end{quote}
Hence
$\raisebox{-.4\height}{\includegraphics[width=0.8cm]{gfx/logo_resto_6acostado_aHor.pdf}}$
and
$\raisebox{-.4\height}{\includegraphics[width=0.9cm]{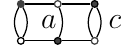}}$
are decidedly different and attention should be paid in the order of
the arguments.  For boundary graphs like
$\raisebox{-.24\height}{\includegraphics[width=.55cm]{gfx/logo_resto_6.pdf}}$
or
$\raisebox{-.20\height}{\includegraphics[width=.45cm]{gfx/logo_k33_blanco}}$
(whose drawing would lead to an ambiguous rule) we can dispense with
that convention, $\Autc(\B)$ ensures the well-definedness of
$G_\B\hp{2k}$.

\subsection{Two-Point function Schwinger-Dyson equations}

One gets the Schwinger-Dyson equation\footnote{ Our approach to SDEs
  differs from that by Gur\u au \cite{GurauSDE}.  He found
  constrictions in form of differential operators acting on
  $Z[J,\bar J]$ and showed that they satisfy a Lie algebra that
  generalizes Virasoro algebra \cite{GurauVirasoro} and is indexed by
  rooted colored trees.  We here work with $\log Z[J,\bar J]$ and our
  approach is not that general but we have those concrete operators
  derived from Theorem \ref{thm:full_Ward}.} for the two-point
function from
$
G\hp 2_{\includegraphics[height=1.6ex]{gfx/Item2_Melon.pdf}}(\mathbf a)
= \left. 
Z_0\inv \delta^2(Z[J,\bar J])/\delta \bar 
J_{\mathbf a} \delta J_{\mathbf a}  \right|_{J=\bJ=0}
$
and by using eq. \eqref{dobleder} with $m_a=n_a$
and $J=\bJ=0$. Here $Z_0=Z[0,0]$. 
If one formally performs the functional integral, one gets for
$\mathbf a\neq 0\in \Z^D$
 \begin{align} 
G\hp 2_{\includegraphics[height=1.9ex]{gfx/Item2_Melon.pdf}}(\mathbf a)
& 
=  \frac{1}{Z_0}\left\{\fder{ }{ \nonumber
J_{\mathbf a} } \left[
\exp\left(-S_{\mathrm{int}}(\delta/\delta \bar J,\delta/\delta J)\right)
\frac{1}{E_{\mathbf a}} J_{\mathbf a}\ee^{\sum_{\mathbf{q}} 
\bar J_{\mathbf q}
E _{\mathbf q}\inv J_{\mathbf q}}\right]\right\}_{J=\bJ=0}\\ \label{bemol}
&=\frac{1}{Z_0E_\mathbf{a}}\left[ 
\exp\left(-S_{\mathrm{int}}(\delta/\delta \bar J,\delta/\delta J)\right)
\ee^{\sum_{\mathbf{q}} 
\bar J_{\mathbf q}
E _{\mathbf q}\inv J_{\mathbf q}}
\right]_{J=\bJ=0} \\
  & \hphantom{=} + \frac{1}{Z_0E_\mathbf{a}}\left(
\exp\left(-S_{\mathrm{int}}(\delta/\delta \bar J,\delta/\delta J)\right) 
J_{\mathbf a}
\fder{ }{J_{\mathbf a}}
\ee^{\sum_{\mathbf{q}} 
\bar J_{\mathbf q}
E _{\mathbf q}\inv J_{\mathbf q}}\nonumber
\right)_{J=\bJ=0} \\
& = \frac{1}{E_{\mathbf a}} + \frac{1}{Z_0} \frac{1}{E_{\mathbf a}} 
\left(
\bar\phi_{\mathbf a} \dervpar{}{\bar \phi_{\mathbf a} 
}\left( 
S_{\mathrm{int}}(\phi,\bar\phi) 
\right)
\right)_{\phi^\flat\to\, \delta/\delta{J^\sharp}} Z[J,\bar J ]  \,, \nonumber
\end{align}
being $\{x^\flat,y ^ \sharp \}=\{ \bar x, y \}$ or $\{x,\bar y\}$.
Here we make a crucial assumption, which is not needed in the rank-$2$
theory (matrix models \cite[Sec. 2]{gw12}). We suppose that
the interaction $S_{\mathrm{int}}$
satisfies the following condition: \textit{each} (graph-)vertex of 
\textit{each} single interaction vertex
lies, for
certain color $a=1,\ldots, D$,
on a subgraph of the following type (in the tensor models parlance, ``melonic
insertion''):
\begin{equation}
\label{condicion_subgr}\phantom{=}
\raisebox{-.5\height}{
\includegraphics[width=2.79cm]{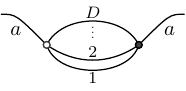}}
\end{equation}
Such is the case for the 
melonic $\phi^4_{\mathsf{m}}$-model in arbitrary rank.  
Then, whatever $S_{\mathrm{int}}(\phi,\bar\phi)$ is, 
in eq. \eqref{bemol} the term 
$\big(
\bar\phi_{\mathbf a} \dervpar{}{\bar \phi_{\mathbf a} 
}S_{\mathrm{int}}(\phi,\bar\phi) \big)_{{\phi^\flat\to \,  \delta/\delta J^\sharp }}
$ contains, for each order-$2r$ monomial, derivatives of $Z[J,\bJ]$ of the form
\[
 \fder{^{2r-2} }{\mathbb J(\B^\times)(\mathbf a^1,\ldots,\mathbf a^{r-1} )} \,\bigg(
 \sum_{ p_i\in \,\Z, i \neq a } 
 \fder{^2 Z[J,\bJ]}{J_{p_1\ldots p_{a-1}m_ap_{a+1} \ldots p_D}
 \delta\bJ_{p_1\ldots p_{a-1}n_ap_{a+1} \ldots p_D}}\bigg)
   \,,
 \]
where $\mathbb J(\B^\times)$ is a generic broken graph, for the moment irrelevant. Here the WTI for
the color $a$ is handy.

\subsection{Schwinger-Dyson equations for the $\phi^4_3$-model}
\label{cgft}
We give an example with a concrete theory, which can be connected with
Tensor Group Field Theory (TGFT).  CTMs with non-trivial kinetic term
usually are originated by TGFT-actions
 \begin{align*}
 S[\phi,\bar \phi]=& \frac{1}{2} \int_{\T^3\times \T^3}\! \! 
 \dif \boldsymbol{g} \dif \boldsymbol{g}'\, \bar\phi(g_1,g_2,g_{3})K(\boldsymbol g,\boldsymbol g')
 \phi(g'_1,g'_2,g'_{3})
 \\ 
&+
\frac{\lambda}{4} \int_{ \T^{12} }\prod_\alpha  \dif 
\boldsymbol{g}^{(\al)} \,V(\boldsymbol g^{(0)},\ldots,\boldsymbol{g}^{(3)})
\phi(\boldsymbol g^{(0)}) \bar\phi(\boldsymbol g^{(1)}) 
\phi(\boldsymbol g^{(2)}) \bar\phi(\boldsymbol g^{(3)})\, ,
 \end{align*}
 by Fourier-transforming the fields $\phi$ and $\bar\phi$ there.  Here
 we will set $K$ to be the Laplacian on
 $\T^3=\mathbb{S}^1\times \mathbb{S}^1\times \mathbb{S}^1$, but
 following analysis an can be with little effort carried on by picking
 a generic diagonal operator $K$.  The action in terms of the
 Fourier-modes is of the form
 $S[\phi,\bar \phi]=\Tr_2(\phi,E\bar \phi)+ \lambda (\Tr_{\mathcal V_1}(\phi,\bar\phi)
+\Tr_{\mathcal V_2}(\phi,\bar\phi)+\Tr_{\mathcal V_3}(\phi,\bar\phi))$, 
with $
\Tr_2{(\bar\phi,E\phi)}=\sum_{x_1,x_2,x_3} \bar\phi_{x_1x_2x_3} (x^2_1+x_2^2+x^2_3+m_0^2)\, \phi_{x_1x_2x_3}\,
$, being $E$ thus diagonal.
For this concrete theory 
\begin{align*}
\left.\bar\phi_{\mathbf x}\dervpar{ S_{\mathrm{int}}}{\bar\phi_{\mathbf x}}\right|_{{\phi^\flat\to \,  \delta/\delta J^\sharp }}
 =2\lambda\Bigg\{  \vphantom{ \sum\limits_{b_1}\fder{ }{\bJ_{b_1x_2x_3}}}  &  
\fder{ }{ J_{x_1x_2x_3}} \sum\limits_{b_1}\fder{ }{\bJ_{b_1x_2x_3}}
\sum\limits_{b_2,b_3}\fder{}{J_{b_1b_2b_3}}\fder{ }{\bJ_{x_1b_2b_3}} &
\\ & \!\!\!\!\!\!\!\! + \left.\fder{ }{ J_{x_1x_2x_3}} \sum\limits_{b_2}\fder{ }{\bJ_{x_1b_2x_3}}
\sum\limits_{b_1,b_3}\fder{}{J_{b_1b_2b_3}}\fder{ }{\bJ_{b_1x_2b_3}}  \right.  
&  
\\
 & \!\!\!\!\!\!\!\!+ \fder{ }{ J_{x_1x_2x_3}} \sum\limits_{b_3}\fder{ }{\bJ_{x_1x_2b_3}}
\sum\limits_{b_1,b_3}\fder{}{J_{b_1b_2b_3}}\fder{ }{\bJ_{b_1b_2x_3}}  
 & \!\! \!\!\!\!\! \!\!\!\!\!\left. \vphantom{\fder{ }{ J_{x_1x_2x_3}} }   \!\Bigg\}Z[J,\bJ] \right|_{J=\bJ=0}\,.
\end{align*}
One uses for each line the WTI for the color in
question. For the first line, $a=1$, this reads
\begin{align*}
2\lambda\left\{ 
\fder{ }{ J_{x_1x_2x_3}} \sum\limits_{b_1} 
\fder{}{ \bJ_{b_1x_2x_3}}\right.&\left.\bigg[ 
\delta_{x_1b_1}Y_{x_1}^{(1)}[J,\bJ] \, \cdot    \right. \\
&\!\!\!\!\!\!\!\!\!\!\!\!\!\!\!\!\!\!\!\!\!\!\!\!\!\!\!\left.\left.+
\sum_{b_2,b_3}\frac{1}{|b_1|^2-|x_1|^2}\left(
\bJ_{{b_1b_2b_3}}\fder{}{\bJ_{x_1b_2b_3}}
-J_{x_1b_2b_3}\fder{ }{ J_{b_1b_2b_3}}
\right)
\right] Z[J,\bJ] \right\} \left.\vphantom{\fder{}{J}}\right|_{J=\bJ=0} \,.
\end{align*}

The derivatives on the $Y_{x_1}\hp1$ yield
\allowdisplaybreaks[1]
\begin{align*}
 \fder{^2 Y_{x_1}\hp1[J,\bJ] }{J_{x_1x_2x_3}\delta{\bJ_{x_1x_2x_3}}}\bigg|_{J=\bar J=0} &= 
\nonumber
\dervpar{ Y\hp1_{x_1}[J,\bJ]}{
 \hspace{.21ex}\raisebox{-2pt}{\includegraphics[height=2.195ex]{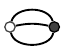}} (\mathbf x)} 
 \\ & = \big(
 \Delta_{x_1,1}
G\hp 4 _{\raisebox{-.33\height}{\includegraphics[height=1.8ex]{gfx_Corrigendum//Item2_Melon}}|
\raisebox{-.33\height}{\includegraphics[height=1.8ex]{gfx_Corrigendum//Item2_Melon}}}+ 
\Delta_{x_1,1} 
G\hp4_{\raisebox{-.4\height}{\includegraphics[width=.48cm]{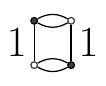}}}
+ 
\Delta_{x_1,1} \nonumber
G\hp4_{\raisebox{-.4\height}{\includegraphics[width=.48cm]{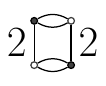}}}
+ \nonumber
\Delta_{x_1,1}
G\hp4_{\raisebox{-.4\height}{\includegraphics[width=.48cm]{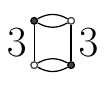}}} \big) (\mathbf x)\nonumber
\\
& =  \sum\limits_{q_2,q_3} G\hp4_{\raisebox{-.34\height}{\includegraphics[width=2.15ex]{gfx_Corrigendum//Item2_Melon.pdf}}|
\raisebox{-.34\height}{\includegraphics[width=2.15ex]{gfx_Corrigendum//Item2_Melon.pdf}}}( x_1,q_2,q_3 ; \mathbf x) 
+G\hp 4_{\includegraphics[width=3ex]{gfx_Corrigendum///logo_vone.pdf}}(\mathbf x,\mathbf x)  \nonumber
\\  
& \quad + \sum_{q_3}
  G\hp 4_{\includegraphics[width=3ex]{gfx_Corrigendum///logo_vtwo.pdf}}(x_1,x_2,q_3;\mathbf x)   
+  \sum_{q_2}G\hp 4_{\includegraphics[width=3ex]{gfx_Corrigendum///logo_vthree.pdf}}(x_1,q_2,x_3;\mathbf x) \,, 
\end{align*} 
and one obtains in similar way the terms concerning 
the two other colors. Straightforwardly one gets that 
for each $\mathbf x=(x_1,x_2,x_3)\in\Z^3$,
 \begin{align} \nonumber
  \hspace{-.1cm}
 G\hp 2_{\includegraphics[height=1.9ex]{gfx_Corrigendum//Item2_Melon.pdf}}(\mathbf{x})
& = \frac{1}{m^2+|\mathbf{x}|^2} \\ 
& \nonumber +
\frac{(-2\lambda) }{m^2+|\mathbf{x}|^2}
\left\{
G\hp 2_{\includegraphics[height=1.9ex]{gfx_Corrigendum//Item2_Melon.pdf}}(\mathbf{x}) \cdot 
\bigg[
\sum\limits_{q,r\in\Z}
G\hp 2_{\includegraphics[height=1.9ex]{gfx_Corrigendum//Item2_Melon.pdf}}(x_1,q,r)
+G\hp 2_{\includegraphics[height=1.9ex]{gfx_Corrigendum//Item2_Melon.pdf}}(q,x_2,r)
+G\hp 2_{\includegraphics[height=1.9ex]{gfx_Corrigendum//Item2_Melon.pdf}}(q,r,x_3)\bigg]
\right. 
\\  \nonumber
& +  \sum\limits_{q\in\Z} \bigg(
G\hp 4_{\includegraphics[width=3ex]{gfx_Corrigendum//logo_vone.pdf}}(x_1,q,x_3;\mathbf{x})+
G\hp 4_{\includegraphics[width=3ex]{gfx_Corrigendum//logo_vone.pdf}}(x_1,x_2,q;\mathbf{x}) +
G\hp 4_{\includegraphics[width=3ex]{gfx_Corrigendum//logo_vtwo.pdf}}(q,x_2,x_3;\mathbf{x}) 
\\
& \phantom{ + \sum\limits_{q\in\Z} }+G\hp 4_{\includegraphics[width=3ex]{gfx_Corrigendum//logo_vtwo.pdf}}(x_1,x_2,q;\mathbf{x}) 
+
G\hp 4_{\includegraphics[width=3ex]{gfx_Corrigendum//logo_vthree.pdf}}(q,x_2,x_3;\mathbf{x})  +
G\hp 4_{\includegraphics[width=3ex]{gfx_Corrigendum//logo_vthree.pdf}}(x_1,q,x_3;\mathbf{x})  
 \bigg) \nonumber 
\\ \hspace{-.0cm}
&+
\sum\limits_{q,r\,\in\,\Z}  
\bigg[ 
G\hp4 _{\raisebox{-.33\height}{\includegraphics[height=1.9ex]{gfx_Corrigendum//Item2_Melon}}|
\raisebox{-.33\height}{\includegraphics[height=1.9ex]{gfx_Corrigendum//Item2_Melon}}}
(x_1,q,r;\mathbf{x})
\label{corrig} 
 +
G\hp4 _{\raisebox{-.33\height}{\includegraphics[height=1.9ex]{gfx_Corrigendum//Item2_Melon}}|
\raisebox{-.33\height}{\includegraphics[height=1.9ex]{gfx_Corrigendum//Item2_Melon}}}
(q,x_2,r;\mathbf{x}) +  G\hp4 _{\raisebox{-.33\height}{\includegraphics[height=1.9ex]{gfx_Corrigendum//Item2_Melon}}|
\raisebox{-.33\height}{\includegraphics[height=1.9ex]{gfx_Corrigendum//Item2_Melon}}}
(q,r,x_3;\mathbf{x}) \bigg] 
\\ \hspace{-.4cm}
&   - \sum\limits_{q_1\in\Z} \bigg[
\frac{1}{q_1^2-x_1^2} \cdot 
\big(G\hp 2_{\includegraphics[height=1.9ex]{gfx_Corrigendum//Item2_Melon.pdf}}(q_1,x_2,x_3) - \nonumber
G\hp 2_{\includegraphics[height=1.9ex]{gfx_Corrigendum//Item2_Melon.pdf}}(\mathbf{x}) \big)
\bigg]   \\  \hspace{-.4cm}
&   -\sum\limits_{q_2\in\Z} \bigg[  \nonumber
\frac{1}{q_2^2-x_2^2} \cdot 
\big(G\hp 2_{\includegraphics[height=1.9ex]{gfx_Corrigendum//Item2_Melon.pdf}}(x_1,q_2,x_3) -
G\hp 2_{\includegraphics[height=1.9ex]{gfx_Corrigendum//Item2_Melon.pdf}}(\mathbf{x}) \big)
\bigg] \\ \nonumber \hspace{-.4cm}
&  -\sum\limits_{q_3\in\Z} \bigg[   
\frac{1}{q_3^2-x_3^2} \cdot 
\big(G\hp 2_{\includegraphics[height=1.9ex]{gfx_Corrigendum//Item2_Melon.pdf}}(x_1,x_2,q_3) -
G\hp 2_{\includegraphics[height=1.9ex]{gfx_Corrigendum//Item2_Melon.pdf}}(\mathbf{x}) \big)
\bigg]  +  \sum\limits_{c=1,2,3}  
G\hp 4_{\includegraphics[width=3ex]{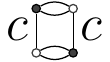}}(\mathbf{x},\mathbf{x})\Bigg\}\,  .
\end{align}
One can conveniently simplify the notation according to  
\begin{equation}
G\hp2 = 
 G\hp 2_{\includegraphics[height=1.75ex]{gfx_Corrigendum//Item2_Melon.pdf}} = \dervpar{\log Z}{
 \hspace{.21ex}\raisebox{-2pt}{\includegraphics[height=2.195ex]{gfx_Corrigendum//Item2_Melon.pdf}}}
 \,, \,\,
G\hp 4_{V_c}= G\hp 4_{\includegraphics[width=.5cm]{gfx_Corrigendum//Item4_Vcv.pdf}}
=
 \dervpar{\log Z}{
 \hspace{.21ex}\raisebox{-2pt}{\includegraphics[height=2.15ex]{gfx_Corrigendum//Item4_Vcv.pdf}}}\,,  \,\,
G\hp4 _{\mathrm m | \mathrm m}
=G\hp 4_
{\raisebox{-.33\height}{\includegraphics[height=1.75ex]{gfx_Corrigendum//Item2_Melon}}  |
\raisebox{-.33\height}{\includegraphics[height=1.75ex]{gfx_Corrigendum//Item2_Melon}}}=
\dervpar{\log Z}{
 \big(\hspace{.21ex}\raisebox{-2.43pt}{\includegraphics[height=2.5ex]{gfx_Corrigendum//Item2_Melon.pdf}} |
 \hspace{.21ex}\raisebox{-2.43pt}{\includegraphics[height=2.5ex]{gfx_Corrigendum//Item2_Melon.pdf}}\big)}\,,
 \end{equation}
and write the 2-point SDE in compact form:
 \begin{align}\nonumber
G (\mathbf{x}) &= \frac{1}{m^2+|\mathbf{x}|^2} + \frac{(-2\lambda)}{m^2+|\mathbf{x}|^2} \sum \limits_{c=1}^3
\Bigg\{\sum \limits_{q_a,q_b} G ({q}_a,q_b,x_c) \times G (\mathbf{x})  \\
 & \quad + \sum\limits_{d = a,b} \sum \limits_{q_c}G^{(4)}_{V_d}(x_a,x_b,q_c,\mathbf{x})  \label{eq:SDE2} \tag{\ref{corrig}$'$} 
     + \sum \limits_{q_a,q_b} G^{(4)}_{\mathrm m|\mathrm m}(q_a,q_b,x_c,\mathbf{x}) 
     \\ & \quad -\sum \limits_{q_c} \frac{1}{q_c^2-x_c^2 }  \vspace{-1cm}
   \Big[G (x_a,x_b,q_c)-G (\mathbf{x})\Big]  +G^{(4)}_{V_c}(\mathbf{x},\mathbf{x})  \Bigg\}\,,   \nonumber
\end{align} 
assuming set-equality $\{a,b,c\}=\{1,2,3\}$.
This 2-point SDE can be formally taken to a (diffeo-)integral equation, 
as in \cite{us}, by taking the continuum limit. 

\section{Conclusions} 
In the quest for the full Ward Takahashi Identity, we have shown that 
the correlation functions of 
rank-$D$ CTMs are classified by boundary ($D$-colored) graphs.  
For quartic melonic models the boundary sector is the set of \textit{all} 
$D$-colored graphs and 
correlation functions indexed by this set are, in their entirety, non-trivial.  
Concerning combinatorics, it would be also interesting to 
grasp, perhaps in terms of covering spaces, the 
counting itself of the \textit{connected} $D$-colored graphs in $2k$
vertices, which, as one can 
extrapolate from \cite{counting_invariants,OEIS}, gives also the number of 
conjugacy classes of subgroups of index $k$ in free group $F_{D-1}$.
\par
A similar, more intricate organization by boundary graphs is very likely to hold for
multi-orientable tensor models; this is worth exploring, in
particular if one is interested in allowing non-orientable manifolds.
\par 
The full WTI works also for $\mathrm U(1)$-TGFT and it would be
interesting to derive a Ward Takahashi Identity for the
$\mathrm{SU}(2)$-TGFT models, or $\mathrm{SU}(2)$-related models like
Boulatov's and Ooguri's.\par
The culmination of this work would be the construction of the
$\phi^4_3$-theory in an Osterwalder-Schrader manner, properly along
the lines of the Tensor Track.  The $\phi_3^4$-theory is
superrenormalizable and its renormalization has been studied
constructively \cite{RivDele} using the multiscale loop vertex
expansion \cite{MLVE}.  The addition of $\phi^6$-interaction vertices
can make the theory from the renormalization viewpoint even more
interesting.  In rank-$4$ such theory would very likely convey the
interesting properties of the Ben Geloun-Rivasseau model.  The new
theory would have, of course, the same boundary sector and therefore
the same expansion of the free energy (obviously with different
solutions for the correlation functions indexed by the same boundary
graph). Hence for melonic $(\phi^4 + \phi^6)$-theories the present 
results hold.  \par

The next obvious step is to solve the equation for the $2$-point
function derived here.  
We shall begin by studying in depth, for  
general theories, the \textit{discrete permutational symmetry} axiom
\cite[Sec. 5, Rule 2]{track_update} stated by Rivasseau. Although all
the models treated here are in fact $\mathfrak S_D$-invariant, there
are elements of CTMs that are not manifestly $\Sym_D$-invariant
(e.g. the homology of graphs defined in \cite{Gurau:2009tw}). After
showing their invariance one would be entitled to state an equivalence of the form
\[
G\hp4_{\raisebox{-.4\height}{\includegraphics[width=.6cm]{gfx/logo_vone.pdf}}} 
\sim 
G\hp4_{\raisebox{-.4\height}{\includegraphics[width=.6cm]{gfx/logo_vtwo.pdf}}}
\sim
G\hp4_{\raisebox{-.4\height}{\includegraphics[width=.6cm]{gfx/logo_vthree.pdf}}} ,
\quad 
G\hp 6_{\raisebox{-.4\height}{\includegraphics[width=.6cm]{gfx/logo_resto_6.pdf}}} 
\sim
G\hp 6_{\raisebox{-.4\height}{\includegraphics[width=.6cm]{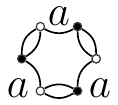}}}  
\mbox{ for }  a\neq c\,,
\]
and similar rules for higher multi-point functions indexed by graphs
lying on the same $\mathfrak S_D$-orbit of the action
$\tau:G_{\B}\hp{2k} \mapsto (\Delta\tau )^*G_{\tau(\B)}\hp{2k } $
where 
$\Delta\tau$ is the diagonal action on $(\Z^D)^k$, $\tau\in\Sym_D$, and the action
of $\Sym_D$ on $\Z^ D$ and
$\amalg\Grph{D}$ is in both cases permutation of
colors. This equalities will noticeably simplify the SDEs, as is
evident in eq.  \eqref{eq:SDE2}.  Based on the WTI exposed here, 
the  full\footnote{As a technicality, the SDE tower was obtained for connected 
correlation functions indexed by connected boundary graphs. 
The SDE obtained there were actually works any model
having \eqref{condicion_subgr} as subgraph in the 
interaction vertices.} 
tower of Schwinger-Dyson equations for quartic theories in 
arbitrary rank was obtained in \cite{SDE}. 
Their renormalized version should also be derived
and we can proceed as in \cite{us}. 
\par
Nonetheless, in order to solve the equations, this approach 
might still not be enough and, in order to obtain a closed equation,
it can be complemented as follows.
Fixed a correlation function $G_{\B} \hp {2p}$, Gur\u au's degree's
range of graphs contributing to $G_{\B}\hp{2p}$ is
$\rho(\B)=\{t(\B) + n\cdot (D-1)!/2 \,|\, n\in \Z_{\geq 0}\}$ being
$t{(\B)}\in\Z_{\geq 0} $ a lower bound depending on $\B$.  As done in
\cite{gw12} for matrix models, one can further possibly decouple the
equation \eqref{eq:SDE2} for the 2-point function by expanding any
correlation function occurring there in subsectors that share the same
value $\alpha$ of the degree,
$G_{\B}\hp{2p}=\sum_{\alpha\in\rho(\B)} G_{\B}\hp{2p,\alpha}$ , in
order to obtain a closed equation. Furthermore, since mainly the
degree conveys the geometrical information, this is a sum over all the
geometries bounded by $|\Delta(\B)|$ that a fixed model (here a
quartic) triangulates.  These subsectors generalize the matrix models'
genus-expansion \eqref{eq:exp_genus_matrix} and in turn justify Figure
\ref{fig:expgurau}.\par
Lemmas \ref{thm:completeness3} and \ref{thm:BS_D}, intended here first
as auxiliary results, are important on their own, if one wants to
understand the geometry of the spaces generated by the $\phi^4_{\mathsf m}$-theories.
In rank $3$ and $4$, for
instance, they realize the triviality of $\Omega^{\mathrm{SO}}_2$ and $\Omega^{\mathrm{SO}}_3$, respectively,
the orientable bordism groups.
This could be an accident due to the equivalence of 
the topological, PL and smooth categories in low dimensions. For 
higher dimensions one should rather compare with
the piecewise linear bordism groups $\Omega_{*}^{\mathrm{PL}}$,
which is beyond the scope of this study but worth analyzing.
\par
The graph-operation $\#$ introduced in Definition \ref{thm:surg_ribb} turned out 
to be optimally-behaved (additive) with respect to Gur\u au's degree. The fact
that melonic graphs are spheres leads us to conjecture that $\#$
\textit{should} be indeed the QFT-compatible graph-realization of the
connected sum in arbitrary rank (a direct proof of which is found in \cite{cips}),
when the colored graphs represent manifolds.

\begin{acknowledgements} The author wishes to thank Raimar Wulkenhaar for enlightening
  comments, and Fabien Vignes-Tourneret motivating discussions. 
  Thanks likewise also to Joseph Ben Geloun and Dine
  Ousmane Samary---whose visit to University of M\"unster was possible
  by means of the Collaborative Research Center ``Groups, Geometry \&
  Actions'' (SFB 878)---for enriching discussions.  The
  \textit{Deutscher Akademischer Austauschdienst} (DAAD) is
  acknowledged for financial support.
\end{acknowledgements}
\appendix

\section{The fundamental group of
  crystallizations} \label{app:fundamental} We expose the computation
of the fundamental group of two crystallizations, i.e. colored graphs
whose number of $(D-1)$-bubbles is exactly $D$ (for
non-crystallizations, Gur\u{a}u has found a representation of the
fundamental group of general colored graphs
\cite[Eq. 40]{Gurau:2009tw}).  The next algorithm by Gagliardi
\cite{gagliardi} works \textit{only for crystallizations}.  Namely,
let $\B$ be a crystallization, for simplicity, of a manifold of
dimension at least $3$ (hence $\B$ has $D>3$ colors).  The fundamental
group of $\B$, $\pi_1(\B)$, is isomorphic to the fundamental group
$\pi_1(|\Delta(\mathcal B)|)$ of the (path connected) space $\mtc B$
represents. The former group is constructed from generators $X$ and
relations $R$ dictated by certain bubbles of $\B$ as follows. Choose
two arbitrary different colors $i,j$ and consider all $D-1$ bubbles
$\{\mathcal B^{\widehat{ i j}}_\alpha\}_{\alpha=1}^n$ without the two
colors $i,j$. Let
\[
  X=\{x_{1},\ldots,x_{ {n }}\}, \qquad x_\alpha = \mbox{generator associated to }
 \mathcal B ^{\widehat{ij}}_\alpha . 
\] Consider the set
$\{\mathcal B^{ i j }_\gamma\}_{\gamma=1}^m$ of 
the $ij$-bicolored $2$-bubbles of $\mathcal B$. 
For each $\gamma$, each vertex of the loop $\mathcal B^{ i j }_\gamma$
intersects certain $\B^{ \widehat{i j} }_\alpha$,
and in that case write $x_\alpha^{\epsilon}$ 
according to the following rule:
set $\epsilon=1$ if the vertex at which $\mathcal B^{ i j }_\gamma$
intersect $\mathcal B ^{\widehat{ij}}_\alpha$ is black and 
$\epsilon=-1$ if it is white.
For  
each $\gamma$ let $R(\mathcal B^{ i j }_\gamma)$ 
be the
word on $X$ defined by juxtaposing all such elements $x_\alpha^{\epsilon}$ 
for each vertex of $\mathcal B^{ i j }_\gamma$, in order of occurrence.
Then
\[
\pi_1(\mathcal B)\cong\langle x_1, \ldots, x_{n-1},x_n \,| \,
x_n, \{ R( \mathcal B^{ i j }_\gamma) : \gamma=1,\ldots, m-1\}\rangle  
\]
Gagliardi's algorithm \cite{gagliardi} states that neither the 
choice of the $D-1$ bubble
$x_n$ that one sets to the identity is important, 
nor the relation $R(\B^{ i j }_m)$ that does not appear is,
nor the two colors $i,j$ are. 

\begin{example}
  We put Gagliardi's algorithm to work for the crystallization. Here
  we come back to the color-set $\{0,1,2,3\}$, instead of
  $\{1,2,3,4\}$ and consider $\Gamma\in\Grph{3+1}$ given in Figure
  \ref{fig:Gagliardi_algorithm} lens space $L_{3,1}$, already
  mentioned in Section \ref{sec:geom_interpretation}.
We choose first the two colors 
$i=2,j=3$, whose corresponding $\{2,3\}^c$-colored $(D-1)$-bubbles are 
$\Gamma_1 ^{01}$ and $\Gamma_2 ^{01}$ depicted below. 
One  associates to all but one of the 
2-bubbles with chosen colors $  \{2,3\}$, a relation.
There are two such 2-bubbles and we drop the inner bubble 
and pick the outer one 
 $\Gamma_1^{23}$, as shown in Figure \ref{fig:Gagliardi_algorithm}. 
The rule says that the (only non-trivial) relation 
$R(\Gamma_1^{23})$
corresponds to $\Gamma_1 ^{23}$ and is 
given by $R(\Gamma_1 ^{23}) = x_1^{+1}x_2\inv  x_1^{+1}x_2\inv x_1^{+1}x_2\inv$.
Notice, incidentally, that if we had instead 
chosen to derive a relation for the inner bubble 
the corresponding relation would be 
$R(\Gamma_2 ^{23})=x_2^{+1}x_1\inv  x_2^{+1}x_1\inv x_2^{+1}x_1\inv$
which is just $(R(\Gamma_1 ^{23}))\inv$. Thus 
$\pi_1(\Gamma)\cong \langle x_1,x_2\, |\, x_2 , R(\Gamma_1^{23}) \rangle = 
\langle  x_1 \, |\, x_1^3   \rangle \cong \Z_3 \cong \pi_1(L_{3,1})\,.
$
\begin{figure}[H]
\centering
\includegraphics[height=3.3cm]{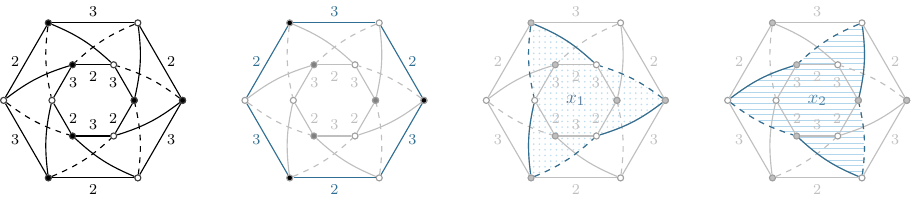} 
\caption{ \label{fig:Gagliardi_algorithm} From left to right: $\Gamma$
  and three of its bubbles $\Gamma_1^{23}$, $\Gamma_1^{01}$ and
  $\Gamma_2^{01}$, used to compute its fundamental group.  For the
  last two bubbles the generators $x_1$ and $x_2$ are depicted }
\end{figure}
\end{example}

\begin{example}\label{ex:cyclic}
 The same algorithm applied to the graph $\mathcal C$ 
 of example \ref{ex:4bordism}:
 \[
 \mathcal C=\hspace{-.5cm}\raisebox{-.45\height}{
 \includegraphics[height=2.6cm]{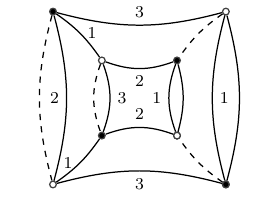}}
 \]
 Namely, one has for the chosen colors $i=2,j=3$ that the relation
 $R(\mathcal C^{23}_{\mathrm{outer}})$ corresponding to the outer
 $23$-bicolored bubble is given by $x_2x_2\inv x_1x_1\inv$ and
 therefore is trivial.  Here $x_1$ (resp. $x_2$) is the leftmost
 (resp. rightmost) $01$-bicolored bubble.  Thus
 $ \pi_1(\mathcal C)=\langle x_1,x_2\, | \,x_2 , R(\mathcal
 C^{23}_{\mathrm{outer}})\rangle = \langle x_1 | \, \emptyset \rangle
 = \Z.  $
\end{example}

\allowdisplaybreaks[1]
\section{The term $Y_{m_a}\hp{a}$ for disconnected graphs}\label{app:B}
A useful formula in order to compute higher order terms in
$Y_{m_a}\hp a$ is presented in this last appendix. We also obtain an
expression that shows that our expansion of $W$ in boundary graphs,
eq. \eqref{expansion_W}, is genuinely the generalization of the
longest-cycle expansion for matrix models of Section
\ref{general_matrix}.  \par We now split an arbitrary boundary graph
$\B$ in its (say $B$) connected components $\R^\beta \in\Grph{D}$,
$\B=\coprod_{\beta=1}^B \mathcal R^\beta$.  Then there exist a
non-negative integer partition $\{n_i\}_{i \geq 0}$ of $B$, i.e.
$B =\sum_{j=1}^\ell n_j$, with $ k= \sum_{j=1}^\ell j\cdot n_j $ and
$\mathcal N =2 k$, where $k$ is the number of $J$-sources (and
therefore also the number of $\bar J$-sources) and $\mathcal N$ the
order of the correlation function, and $n_k$, as before, is number of
boundary components with exactly $k$ $J$-sources (or, equivalently,
$k$ $\bJ$-sources).  We can associate to $\R^\beta$ a product of
sources, $\mathbb J(\R^\beta)$, as above.  Of course
$\mathbb J(\B)=\prod_\beta \mathbb J(\R^\beta)$.  In that case
\[
G^{(2k)}_{\B}\star \mathbb J (\B )=\sum\limits_{\mathbf{a,b,\ldots,c}}
G^{(2k)}_{|\R^1|\R^2|\ldots|\R^B|} (\{\mathbf a\},\{\mathbf
b\},\ldots,\{\mathbf c\} )\cdot \mathbb J( \R^1)(\{\mathbf a\})\mathbb
J(\R^2)(\{\mathbf b\}) \cdots \mathbb J(\R^B)(\{\mathbf c\}).\]
Now, some of the $\R^\beta$'s might be repeated. Let $s$ be the number
of \textit{different} graphs at the disconnected components of the
boundary
$\{\R_\beta\}_{\beta=1}^B=\cup_{b=1}^s\{ \R_1^b,\ldots,\R_{m_b}^b\}$,
where $m_b$ is the number of copies of $\R^b$. Moreover, we order
$
\R^b_{\!\!\phantom{a}_{\bullet}} \mbox{ ascending: } \R^b\leq\R^{b'}, 
\mbox{ if }(\R^b)^{(0)}\leq (\R^{b'})^{(0)}    
$
(which is equivalent to say, that the boundary component of type
$\R^{b'}$ has the same number of external lines or more than $\R^b$).
Then, of course $\sum_b m_b=B$. We rewrite then $W$ as
\begin{align} 
W[J,\bar J] &= \sum\limits_{\ell=1}^\infty \sum\limits_{\substack{m_s\geq 1 \label{ansatz2}  \\ 
m_{1},\ldots,m_{\ell-1}\geq 0}}\!\!\!\!\!\!\!
\!\!\!{\vphantom{\sum}}' \left(\frac{1}{m_1!\ldots m_s! (|\Autc(\R^1)|)^{m_1}\ldots  
(|\Autc(\R^s)|)^{m_s}}
 \right) \\
 & \quad  \cdot
G^{(2k)}_{|\R_1^1|\R^1_2|\ldots|\R^1_{m_1}|\cdots|\R^s_{1}|\ldots |\R^s_{m_s}|} 
 \star \mathbb J
\big(( {\R_1^1}\sqcup {\R_2^1}  \nonumber
\cdots \sqcup {\R_{m_1}^1})\sqcup \cdots \sqcup({\R_1^{s}}  \sqcup
\cdots  \sqcup{\R^s_{m_s}})\big)\, .
\end{align} 
The symmetry factors are consequence of eq. \eqref{eq:corona}.
The prime in the sum denotes the following two restrictions:
\begin{equation} \label{eqn:suma_prime}
2k = \sum_{p=1}^s m_p \cdot |(\R^p)\hp0|\qquad  \mbox{ and } \qquad 2\ell=|(\R^s)\hp 0|\,,
\end{equation}
so $\ell$ is the half of the ``largest number of vertices'' of the
components of the boundary graph.

The sum in eq. \eqref{termY} is over all possible disconnected graphs.
We now obtain the expression for $Y\hp a_{m_a}$ in terms of the
connected components of the graphs.  For the boundary graph
$\B=( {\R_1^1}\sqcup {\R_2^1} \cdots \sqcup {\R_{m_1}^1})\sqcup \cdots
\sqcup({\R_1^{s}} \sqcup \cdots \sqcup{\R^s_{m_s}})$,
where $\R_1^b,\ldots,\R_{m_b}^b$ are copies of the same graph, having
$k_i$ white vertices. Now, the operators $\Delta_{r,m_a}$ select the
$r$-th white vertex. We give the white vertices of $\B$ the order
according to the occurrence of these copies of the connected parts of
$\B$, i.e.  the first $k_1$ white vertices are in $\R_1^1$; the next
vertices, from the $(k_1+1)$-th until the $(2k_1)$-th in $\R_2^1$ and
so on. The last $k_s$ vertices are in $\R_{m_s}^s$.)  For sake of
notation $\mathcal N=\mathcal N(\{k_i\})$, the order of the Green's
function, is not made explicit, given by
$\mathcal N=2\cdot (\sum_{i=1}^s m_i k_i)$ one derives the following
expression for $Y\hp a_{m_a}[J,\bar J]$:
\begin{align}  
& \sum\limits_{\ell=1}^\infty \sum\limits_{\substack{m_s\geq 1   \\ 
m_{1},\ldots,m_{\ell-1}\geq 0}}\!\!\!\!\!\!\!
\!\!\!{\vphantom{\sum}}' \left(\frac{1}{m_1!\ldots m_s! (|\Autc(\R^1)|)^{m_1}\ldots  
(|\Autc(\R^s)|)^{m_s}}
 \right) \\
 &\!\!\!\times \left[
\sum_{r=1}^{k_1} (\Delta_{r,m_a}G^{(\mathcal N)}_{|\R_1^1|\R^1_2|\ldots|\R^1_{m_1}|\cdots|\R^s_{1}| \ldots |\R^s_{m_s}|}) 
\star \left(\R_1^1\ominus e_a^r \sqcup {\R_2^1} \sqcup \ldots \sqcup {\R_1^{s}}\sqcup \cdots  \sqcup{\R^s_{m_s}} \right) \right.  \nonumber
\\
&+ \sum_{r=1}^{k_1} (\Delta_{k_1+r,a}G^{(\mathcal N)}_{|\R_1^1|\R^1_2|\ldots|\R^1_{m_1}|\cdots|\R^s_{1}|\ldots |\R^s_{m_s}|}) 
\star \left(\R_1^1 \sqcup {\R_2^1}\ominus e_a^r \sqcup \ldots \sqcup {\R_1^{s}}\sqcup \cdots  \sqcup{\R^s_{m_s}} \right) +\ldots\nonumber 
\\
 &+ \sum_{r=1}^{k_1} (\Delta_{m_1\cdot k_1+r,a}G^{(\mathcal N)}_{|\R_1^1|\ldots|\R^1_{m_1}|\cdots|\R^s_{1}|\ldots |\R^s_{m_s}|}) 
\star \left(\R_1^1 \sqcup \ldots \sqcup {\R_{m_1}^1}\ominus e_a^r \sqcup \cdots \sqcup {\R_1^{s}} \cdots  \sqcup{\R^s_{m_s}} \right) \nonumber \\
& \nonumber
+\sum_{r=1}^{k_s} (\Delta_{m_1\cdot k_1+m_2\cdot k_2+\ldots+m_{s-1}k_{s-1}+r,a}
G^{(\mathcal N)}_{|\R_1^1|\ldots|\R^1_{m_1}|\cdots|\R^s_{1}|\ldots |\R^s_{m_s}|})
\\
& \hspace{7cm}\star \left(\R_1^1 \sqcup  \ldots \sqcup {\R_{m_1}^1}
\sqcup \cdots \sqcup {\R_1^{s}}\ominus e_a^r \sqcup \cdots  \sqcup{\R^s_{m_s}} \right) \nonumber 
\\
&+\ldots  \nonumber 
+\sum_{r=1}^{k_s} (\Delta_{m_1k_1+m_2k_2\ldots+m_{s-1}k_{s-1}+(m_s-1)k_s+r,a}
G^{(\mathcal N)}_{|\R_1^1|\R^1_2|\ldots|\R^1_{m_1}|\cdots|\R^s_{1}|\ldots |\R^s_{m_s}|}) 
\\
&\hspace{7cm}\star \left(\R_1^1  \sqcup \ldots 
\sqcup {\R_1^{s}}\sqcup \cdots  \sqcup{\R^s_{m_s}\ominus e_a^r } 
\right) \Big]\, ,  \nonumber
\end{align} 
with the prime, as before, meaning the restrictions of the sum by
eqs. \eqref{eqn:suma_prime}.

\end{document}